\documentclass[12pt]{article}
\usepackage[left=1.25in, right=1.25in, top=1.25in, bottom=1.25in]{geometry}
\usepackage{amssymb}
\usepackage{mathrsfs} 
\usepackage{amsmath}
\usepackage{calligra}
\usepackage{bold-extra}
\DeclareFontFamily{U}{mathc}{}
\DeclareFontShape{U}{mathc}{m}{it}%
{<->s*[1.03] mathc10}{}
\DeclareMathAlphabet{\mathcal}{U}{mathc}{m}{it}
\DeclareMathAlphabet{\mathpzc}{OT1}{pzc}{m}{it}

\newcommand{\V}[1]{``{#1}"}
\usepackage{epsfig}
\usepackage{natbib}
\usepackage{rotating}
\usepackage{latexsym}
\usepackage{boxedminipage}
\usepackage{verbatim}
\usepackage{amsmath}
\usepackage{nicefrac}
\usepackage{authblk}
\usepackage{bm}
\usepackage{import}
\usepackage{xcolor}
\usepackage{enumitem}
\usepackage{xr}
\usepackage[scr]{rsfso}
\usepackage{yfonts}
\def\be{\begin{equation}}
\def\ee{\end{equation}}
\def\ber{\begin{eqnarray}}
\def\eer{\end{eqnarray}}
\def\berr{\begin{eqnarray*}}
\def\eerr{\end{eqnarray*}}
\def\mc{\bm{\mathfrak{C}}}
\long\def\symbolfootnote[#1]#2{\begingroup
\def\thefootnote{\fnsymbol{footnote}}\footnote[#1]{#2}\endgroup}
\newcommand{\field}[1]{\mathbb{#1}}
\long\def\symbolfootnote[#1]#2{\begingroup
\def\thefootnote{\fnsymbol{footnote}}\footnote[#1]{#2}\endgroup}
\def\sc{\bm{\mathcal{C}}}

\def\cov{{\rm cov }}
\def\cum{{\rm cum }}

\def\tr{{\rm tr }}
\def\diag{{\rm diag }}
\def\I{\bm{I}}
\def\norm{\mathcal{N}}

\newcommand{\prob}{\field{P}}

\newcommand{\E}{\field{E}}

\def\A{\bf{A}}
\def\P{\bm{P}}

\def\tr{{\rm tr }}
\def\lb{\left\|}
\def\rb{\right\|}
\def\ka{\kappa}
\def\epi{\bm{\varepsilon_i}}
\def\epj{\bm{\varepsilon_j}}

\def\ei{\bm{\epsilon_i}}
\def\ej{\bm{\epsilon_j}}
\def\Gi{\bm{\Gamma_{i}}}

\def\Gj{\bm{\Gamma_{j}}}

\def\f{\bm{\mathrm{F}}}
\def\ff{\bm{\mathcal{F}}}
\def\g{\bm{\mathrm{G}}}
\def\gg{\bm{\mathcal{G}}}
\def\x{\bm{\mathrm{X}_i}}
\def\xx{\bm{\mathcal{X}_i}}
\def\xxt{\bm{\mathcal{X}_t}}
\def\xt{{\bm X}_t}

\def\xxj{\bm{\mathcal{X}_j}}
\def\z{\bm{\mathrm{Z}}}
\def\zi{\bf{Z}_i}
\def\zz{\bm{\mathrm{z}}}

\def\yy{\bm{\mathcal{y}_i}}
\def\yyt{\bm{\mathcal{y}_t}}
\def\y{\bm{\mathrm{y}}_i}
\def\v{{\bf{V}}_{\bm{i}}}
\def\vj{{\bf{V}}_{\bm{j}}}

\def\vv{\bm{\mathcal{V}}_{\bm{i}}}
\def\vvj{\bm{\mathcal{V}_{j}}}
\def\vvv{\bm{\mathrm{v}}}
\def\B{\bm{ \mathrm {B}_N}}
\def\b{\bm{\mathrm{b}_i}}
\def\bj{\bm{\mathrm{b}_j}}
\def\uu{\bm{\mathcal{u}_i}}
\def\uuj{\bm{\mathcal{u}_j}}
\def\huu{\bm{\hat{\mathcal{u}}_i}}
\def\huut{\bm{\hat{\mathcal{u}}_t}}
\def\ui{\bm{\mathrm{u}_i}}
\def\hui{\bm{\hat{\mathrm{u}}_i}}
\def\ut{\bm{\mathrm{u}_t}}
\def\hut{\bm{\hat{\mathrm{u}}_t}}

\def\d{\bm{\mathrm{D}}}
\def\db{\bm{\mathrm{D}}_{\bot}}
\def\mzi{\mathscr{M}_{\zi}}
\def\pd{\mathscr{P}_{\d}}
\def\md{\mathscr{M}_{\d}}

\def\mxx{\mathscr{M}_{\xx}}
\def\mxxt{\mathscr{M}_{\xxt}}
\def\projp{\mathscr{P}}
\def\projm{\mathscr{M}}
\newcommand{\m}{\text{\calligra m}}
\def\si{\psi^{-}}
\def\sth{\bm{\mathrm{S}}_N}
\def\sthi{\bm{\mathrm{S}}^{-1}_N}
\def\ssth{\bm{\mathcal{S}}_N}
\def\ssthi{\bm{\mathcal{S}}^{-1}_N}

\def\SO{\bm{\Sigma}}

\def\SS{\bm{\Omega_N}}
\def\SSI{\bm{\Omega_N}^{-1}}
\def\SSIb{\bm{\bar{\Omega}}}
\def\sh{\bm{\hat{\mathcal{S}}_N}}
\def\shi{\bm{\hat{\mathcal{S}}_N}^{-1}}
\DeclareMathOperator{\sa}{\mathscr{F}}
\def\sxi{\bm{\Sigma_{\xx ' \xx}}^{-1}}
\def\sx{\bm{\Sigma_{\xx ' \xx }}}
\def\sxf{\bm{\Sigma_{\xx ' \ff }}}

\def\Hi{\bm{\Xi_{i}}}

\def\Hn{\bm{\Xi_{N}}}
\def\Hninv{\bm{\Xi^{-1}_{N}}}
\def\txx{\frac{\xx'\xx}{T}}

\def\sxf{\bm{\Sigma_{\xx'\ff}}}
\def\sf{\bm{\Sigma_{\ff'\ff}}}
\def\sv{\bm{\Sigma_{\v'\v}}}
\def\dif{\left(\ps-\txx\right)}
\def\xp{\xx^+}
\def\ixx{\left(\frac{\xx'\xx}{T}\right)^{-1}}
\def\sh{\bm{\hat{\mathcal{S}}_N}}
\def\shi{\bm{\hat{\mathcal{S}}_N}^{-1}}

\def\st{\sqrt{T}}
\def\MA{\bm{\mathcal{H}_N}}
\def\MAI{\bm{\mathcal{H}^{-1}_N}}

\def\MAIN{\bm{C^{-1}_N}}

\def\cn{\bm{\mathrm{C}_N}}
\def\cni{\bm{\mathrm{C}_N}^{-1}}

\def\bcn{\bm{\breve{\mathrm{C}}_N}}

\def\stldp{\bm{\tilde{S}^+}}
\def\oi{\bar{\omega}}
\def\bb{\boldsymbol{\beta}}
\def\ba{\boldsymbol{\alpha}}
\def\Hbar{\bm{\bar{\Xi}_{N}}}
\def\Hbari{\bm{\bar{\Xi}^{-1}_{N}}}
\def\hbar{\bar{\xi}}

\def\T{\bm{\Theta_i}}
\def\BTN{\bm{\breve{\Theta}}_N}
\def\BT{\bm{\breve{\Theta}_i}}
\def\ps{\bm{\Psi_i}}
\def\psinv{\bm{\Psi_i}^{-1}}

\def\di{\mathscr{D}}
\def\pp{\mathcal{P}}

\def\dic{\mathcal{D}^c}

\def\fft{\bm{\chi_i}}
\def\ft{\chi}
\def\MAB{\bm{\bar{C}_N}}
\def\xii{\bar{\xi}}
\def\cb{\bar{c}_N}

\def\stldp{\bm{\tilde{S}^+}}
\def\norm{\mathcal{N}}
\def\aa{\approx}
\def\rem{\bm{\mathfrak{R}}}
\def\an{\bm{A_N}}
\def\han{\bm{\hat{A}_N}}
\newtheorem{dummytheorem}{Dummy-Theorem}[section]
\newtheorem{assumptions}{Assumptions}[section]
\newtheorem{remark}{Remark}[section]
\newcommand{\proofendsign}{$\Box$} 

\newtheorem{lemma}[dummytheorem]{Lemma}
\newtheorem{theorem}[dummytheorem]{Theorem}
\newtheorem{prop}[dummytheorem]{Proposition}

\newtheorem{corollary}[dummytheorem]{Corollary}

\newtheorem{example}[dummytheorem]{Example}
\newtheorem{assumption}[assumptions]{Assumption}

\newenvironment{proof}{{\noindent \bf Proof }}
 {{\hspace*{\fill}\proofendsign\par\bigskip}}

\def\an{\bm{A_N}}
\def\han{\bm{\hat{A}_N}}
\usepackage[title]{appendix}
\renewcommand\appendix{\par
\setcounter{section}{0}%
\setcounter{subsection}{0}%
\setcounter{table}{0}
\setcounter{figure}{0}
\setcounter{equation}{0}
\numberwithin{equation}{section}
\gdef\thetable{\Alph{table}}
\gdef\thefigure{\Alph{figure}}
\gdef\theequation{\Alph{section}.\arabic{equation}}
\gdef\thesection{\Alph{section}}
}

\newlist{Part}{enumerate}{1}
\setlist[Part,1]{
  label={Part~(\alph*)},
  leftmargin=*,
  align=left,
  labelsep=1mm,
  itemindent=\dimexpr\labelsep+\labelwidth+18pt\relax
}

\newlist{Partn}{enumerate}{1}
\setlist[Partn,1]{
  label={Part~(\roman*)},
  leftmargin=*,
  align=left,
  labelsep=1mm,
  itemindent=\dimexpr\labelsep+\labelwidth+18pt\relax
}

\title{ Robust  Nearly-Efficient Estimation  \\  of Large Panels   with Factor Structures}

\author[a]{Marco Avarucci}
\affil[a]{University of Glasgow}
\author[b]{Paolo Zaffaroni}
\affil[b]{Imperial College London}
\providecommand{\keywords}[1]{\textbf{Keywords:} #1}
\begin{document}
\maketitle
\begin{abstract}
This paper  studies   estimation of linear panel regression models  with heterogeneous coefficients, when both the regressors and the residual contain a possibly common, latent, factor structure. Our theory is  (nearly) efficient, because based on the GLS principle, and also robust to the specification of such factor structure, because it does not require any information on the number of factors nor estimation of the factor structure itself.
  We first show how the unfeasible GLS estimator not only affords an efficiency improvement but, more importantly, provides a bias-adjusted estimator with the
    conventional limiting distribution, for situations where the OLS is affected by a first-order bias.  The technical  challenge resolved in the paper is to show how these properties are preserved for a class of feasible GLS estimators in a double-asymptotics setting. Our theory is illustrated by means of Monte Carlo exercises and, then, with an empirical application using individual asset returns and firms' characteristics data.
\end{abstract}


\section{Introduction}







This paper  considers  (nearly) {\em efficient}  estimation of linear panel regression models  with heterogeneous coefficients, when both the regressors and the residual contain a common, latent, factor structure. At the same time, our estimation procedure does not require any knowledge of such  latent factor structure, not even the maximum possible number of latent factors, let alone the latent factors themselves.  This qualifies our procedure as {\em robust}.
Factor models represent one of the most  popular and successful way to capture cross-sectional and temporal dependence, especially when facing a large number of units ($N$)	 and time periods ($T$), although in our context factors and their loadings represent nuisance parameters.

  However, the possibility of a common factor structure in  both regressors and residuals, which
 would  typically arise when omitting relevant regressors, leads  to endogeneity, making estimation by ordinary least squares (OLS)  invalid such that  all  its asymptotic properties are not holding any longer.

We first consider   an unfeasible generalized least squares (UGLS) estimator for the regression coefficients,
based on the presumption that the covariance matrix of the residuals, evaluated conditional on the latent factors, is known. It turns out that, regardless of the possibility of endogeneity (that is when regressors and residuals  are correlated),  the UGLS is  $\sqrt{T}$-consistent and asymptotically normal
distributed {\em without }  requiring any  information   on the factor structure, such as the number of factors or the factors themselves and their loadings.  
This contrasts with the asymptotic bias plaguing the OLS estimator, under the same circumstances.
  In other words, the UGLS does not only represent a more efficient estimator  but  provides an {\em automatic } biased-adjusted estimator with desirable asymptotic properties.   This  result  is  due to an important    insight, namely the existence of a  form of {\em asymptotic orthogonality } between the common factors, that affect the residuals,  and the {\em inverse } of the residuals'   covariance matrix.  Most importantly, such asymptotic orthogonality  is manifested  at a very fast rate, namely the squared norm of the product between the covariance matrix and the factors is $O(T^{-1})$.

  The  challenge  arises when considering a 
 feasible version of the UGLS  estimator. A natural approach, here followed, is to  
 make  use of the panel dimension,  considering the sample (across $N$)  covariance
 matrix of the OLS residuals, which in turn have been obtained by a (time-series) regression with $T$ observations.
%
    Unlike the OLS and UGLS cases, the asymptotic theory for the GLS requires  both $N$ and $T$ to diverge.   Lack of consistency of the  OLS estimator for the regression coefficients  unavoidably implies  that such $T \times T $ sample covariance matrix  is not consistent, element by element, for the true residuals' covariance matrix but will converge (element by element) instead to a pseudo-true covariance matrix.
      The  surprising, crucial,  result here established is that such pseudo-true covariance matrix is also  asymptotically orthogonal
 to the latent factors, and at the same rate of convergence $O(T^{-1})$.  Indeed,  there is an entire class of
 matrices, rather than  a unique matrix, that is asymptotically
 orthogonal to the factors.  This is the most intriguing aspect of our theory.  The technical achievement of this paper is to show that the {\em feasible }
 GLS (henceforth GLS)
 estimator for the regression coefficients,  is $\sqrt{T}$-consistent and  asymptotically normal, as {\em both }  $N,T$ diverge to
 infinity. Again, this holds even when the  OLS remains an invalid estimator.
At the same time, since  the pseudo-true value differs  in general from the true covariance matrix,  the GLS might not be as efficient  as
the UGLS.  However, evaluation of the GLS iteratively, as explained below, permits  to make it  close to the UGLS estimator. 


In summary, the GLS estimator exhibits four main  desirable, compelling, properties.
First,   it permits to carry out inference on the regression coefficients  based on conventional asymptotic distributions.  In particular, the GLS estimator of the regression coefficients has a  mixed-normal asymptotic distribution, implying the possibility of   inference  by means of chi-squared criteria. 
    Second, as in classical estimation theory, it delivers (nearly) efficient estimation.
 Third,   the GLS estimator does not require any knowledge of the exact number  of latent
  factors, or even an upper bound of such number. In particular,  the number of factors can be either smaller, equal  or larger than the number of  regressors.
 Fourth, the GLS is computationally easy to handle since it simply requires to perform ($N$) linear
 regressions, without invoking any nonlinear numerical optimizations. 
Our approach can be also applied to the dual case of  cross-sectional regressions with time-varying coefficients.

%
%
%
%
%
%
%
%
%

This paper belongs to, and extends, two different strands of literature.


 First,    it has been demonstrated, in various contexts, that efficient estimation techniques not only lead to an improvement of precision but, most importantly, resurrect the required asymptotic properties, in terms of bias, rate of convergence and distribution, in situations where these are not warranted by non-efficient approaches.  
 
 In the context of cointegrated systems,  \cite{P91b} and \cite{P91a} 
show that use of the efficient, full system, maximum likelihood (ML)  goes beyond an efficiency improvement: it solves the well-known issues of specification and inference in cointegrated systems, that plagues unrestricted VAR estimation such as the presence of asymptotic biases and non-standard asymptotic distributions (i.e. the Dickey-Fuller distribution).\footnote{Indeed, 
especially   \cite{P91b}   provides a detailed explanation of  these properties,  namely removing second order bias, dealing with endogeneity, absence of nuisance parameters and, obviously, achievement of full efficiency.} 
   Note that  ML is  asymptotically equivalent to GLS in that set-up.   Although our  theoretical framework is not one of cointegrated systems, strong analogies emerge with  \cite{P91b} and \cite{P91a}: in both cases, a (local)
 mixed-normal distribution arises and efficient estimation mitigates the lack of strong exogeneity. Moreover, such deficiency (i.e. lack of exogeneity) is manifested through the form of the residuals' covariance matrix:  non-block diagonality for (triangular) cointegrated systems of    \cite{P91b}, \cite{P91a}  and a factor structure such as ours, which also rules out block-diagonality, for our framework. Second, \cite{P91b} demonstrates how  these  remarkable properties of  efficient estimation are warranted by full system regressions but not by single-equation regressions. Likewise, our method requires the full information arising from the panel, namely one needs  both $T$ and $N$  to diverge.
 
 \cite{RH97}  study estimation of time series regression models, when both the regressors and the residual exhibit long-memory, and in fact spectral singularities can arise at any frequency. Under these circumstances, in particular when the spectral singularities of the regressors     and residuals arise at the same frequency with sufficient intensity, the  OLS estimator is no longer $\sqrt{T}$-consistent and asymptotically normal.  However,  under the same circumstances,  \cite{RH97} show that a class of weighted least squares estimates, which includes GLS as a special case, has standard asymptotic properties. 
 
 An important difference  between  our approach and \cite{P91b}, \cite{P91a}  and   \cite{RH97} is that their estimation procedure is affected by a second-order bias, that is their estimators are consistent (although with non-standard rate of convergence and asymptotic distribution), whereas in    our context a first-order bias arises, leading for instance to inconsistency of the  OLS estimator. Therefore, our GLS  adjustment appears compelling in our framework.

Second,  inference of panel data model with a latent factor structure in the residuals and  {\em heteroreneous} regression coefficients has been studied, initially, from a purely econometric perspective  and,   more recently,  from an empirical finance angle.

    In a linear cross-sectional
regression  \cite{A05} shows  that, when residuals and  regressors share  a factor
structure, $\sqrt{N}$-consistency of the OLS estimator is preserved only  with uncorrelated factor loadings.\footnote{Although not spelled out, \cite{A05} can be readily applied to time-varying coefficients.}
 Within a linear time regression,
  \cite{P06}  shows that
 heterogeneous  regression coefficients can be $\sqrt{T}$-consistently
estimated  OLS by  augmenting the regressors with cross-sectional averages
of the dependent variable  and individual-specific regressors.  \cite{AB15} 
 consider  a panel model with (sparse) heterogenous coefficients  and establish the asymptotics of a penalized OLS estimator.
 Maintaining the assumption of a common latent factor structure in the residuals of a  panel data model with heterogeneous coefficients,  \cite{vel17}  allow for the possibility that 
 the idiosyncratic innovation is non-stationary, in particular  exhibiting long memory.  
 
Motivated by empirical asset pricing,  new methods to conduct  robust  inference on panel data models with  a latent factor structure have been recently developed. 
\cite{GX18}  derive the asymptotics for a  procedure to estimate the risk-premium of an observed factor, robust to the omission of the set of relevant (i.e. priced)  factors. 
 Like us,  they adopt a double-asymptotics approach. However,  \cite{GX18} differ from us because they focus on estimation of the parameters of the second-pass regression, 
 that is when the asset-pricing restriction is imposed,   whereas we ignore any asset-pricing content (i.e., from the point of view of the two-pass methodology, we focus on the parameters of the first-pass regression).  Moreover, their procedure relies on estimating the   complete space spanned by the latent factors driving the model whereas our method can avoid this aspect altogether. Gagliardini, Ossola and Scaillet (2018)
  study the  properties of 
 a diagnostic criterion to detect an approximate factor structure in the residuals of large, unbalanced, panel data models. Like us, they consider a double-asymptotic setting and ignore any 
 asset-pricing restrictions on the parameters of the panel data model. Moreover,  \cite{GOS18}  method is robust, in the sense that it does not need to explicitly estimate the latent factor structure embedded in the residuals, just like us. However  their focus is specifically to check whether the unobserved residuals have a factor structure whereas our method focuses on estimation of the regression coefficients to the observed, possibly heterogenous, regressors. 

Unlike the previous papers, the large majority of contributions to this literature focused on the case of  constant regression coefficients. 
  \cite{P06}   shows that  a faster rate of convergence is achieved with  constant  regression coefficients.
   \cite{B09} considers joint estimation of the constant regression coefficients and of the residuals'  factor structure components  through an iterative OLS procedure.
   The same estimator has been studied by   \cite{MW09}  under weaker conditions on the observed regressors. 
 %
  \cite{MW13} show that    \cite{B09}  and  \cite{MW09}  results  hold, with no loss of efficiency, when the exact number of latent  factors $M$  is unknown and only  an upper bound is specified. \cite{BAIGLS}  show that  GLS estimation leads to an efficiency improvement over the   \cite{B09}  and  \cite{MW09} OLS-type estimator. 
	     \cite{GHS12}  establish $\sqrt{NT}$-asymptotics for the OLS estimator by  augmenting the regressors with the principal component estimator of the common factors extracted from the observable data.
 \footnote{Several generalizations of the aforementioned results have been considered. 
 \cite{PT11} and \cite{CP13} confirm the same asymptotic results of     \cite{P06} when  spatial-dependence  in the idiosyncratic component of the innovation's factor structure as well as dynamic panel, respectively, are allowed for.   \cite{KUW15}  show  that Pesaran's  estimator retains its asymptotic properties under weaker conditions, 
 allowing for either correlated loadings or for the number of latent factor  $m$ to be larger than the number of observables, whereas  \cite{WU15} discuss some limitations.   \cite{S13} extends \cite{B09} to the case of non-constant regression coefficients  establishing $\sqrt{T}$-asymptotics when $T/N^2 \rightarrow 0 $.  Dynamic panel are  permitted.}
  Other contributions to this literature include  Holtz-Eakin, Newey \& Rosen (1988), Ahn, Hoon Lee \& Schmidt (2001),
\cite{BN04}, \cite{PS03},  \cite{MP04} and \cite{PS07}.

%
%
%
%
%
None of these papers address the issue of efficient estimation, except for \cite{BAIGLS}\footnote{\cite{BAIGLS} focus on the {\em homogeneous} parameter case, unlike us, and considers joint estimation of the latent factors and parameters, generalizing \cite{B09} and \cite{MW09}. 
More importantly, the motivation of \cite{BAIGLS} differs drastically from ours because they focus on the GLS approach for an efficiency improvement of an estimator that already exhibits the conventional asymptotic properties under \cite{BAIGLS} assumptions, in particular iid-ness across time. In our case, our GLS approach mitigates the {\em first-order} bias  affecting the  OLS estimator, where we allow for both serial  and cross-sectional correlation as well as heteroskedasticity of the residuals. 
},  but rather focus on  various, ingenious, ways to mitigate the bias induced by the correlation between regressors and innovations. 
 In contrast, our GLS approach allows to tackle both issues,  at the same time, without requiring {\em any } knowledge of the factor structures affecting the regressors and innovations.\footnote{In particular, given that we can afford to be completely agnostic about the need to conduct inference on the latent factor structure  affecting the model, our work differ, both in terms of focus and in terms of the techniques developed, from  the multitude of papers developing inference methods on latent factor structures
(on estimating the number of latent factors see  \cite{BN02}, \cite{HL07},  \cite{AW07}, \cite{O09},  \cite{O10}, \cite{AH13}   and  on estimating latent factor structures see \cite{FHLR00},  \cite{SW02},  \cite{BN02}, \cite{B09}  among  others.}
  Our  asymptotic  distribution theory requires  $ T^3/N^2 \rightarrow 0 $ whereas the milder $ T/N \rightarrow 0 $ ensures consistency. This relative speed  spells out a neat dichotomy  in terms of the role of $T$ and $N$: the faster rate of divergence for $N$ is asked for to estimate accurately the (inverse of the) sample-covariance matrix  required by the GLS formula, 
which in turn mitigates the asymptotic bias.
 Instead,  the slower divergence of $T$ controls the asymptotic variance of the GLS estimator, dictating ultimately the estimator's rate of convergence. Noticeably, the relative speed requested by our estimator  differs from the relative speed requested by the alternative procedures described above,
 suggesting that our result can also be viewed as complementary to the others, for example  more suitable to short panels where $N$ is  much larger than $T$.\footnote{For example, \cite{P06}  requires  $T/N^2 \rightarrow 0$ for asymptotic normality but the weaker condition  $T/N \rightarrow 0 $ is required for homogeneous regression coefficients, where the faster  $\sqrt{NT}$ -rate of convergence is achieved.
 Moreover, one needs  the number of heterogeneous regressors  to be greater than number of latent factors. \cite{B09} shows that the regression coefficients' estimator is also $ \sqrt{NT}$-consistent,   when $T/N \rightarrow \kappa > 0 $ for some constant $ \kappa $. \cite{B09} estimator is asymptotically biased, in general, but an asymptotically valid bias-correction is established under slightly stronger conditions. 
 \cite{MW09}  establish  $ \sqrt{NT} $-asymptotics, again when  $N$ and $T$  diverge at the same speed (i.e. $N/T  \rightarrow \kappa > 0 $).}

%

This paper proceeds as follows.
    Section~\ref{def} illustrates the general model and
 the assumptions required for
estimation of regressions with unit-specific parameters. The asymptotic results for the OLS, UGLS and GLS estimators are presented in Section~\ref{res}. 
Section~\ref{common} describes estimation and inference of the coefficients to common regressors. 
The technical contributions of the paper are discussed and highlighted in Section~\ref{LEMMI}
Section~\ref{various} discusses various issues related to the GLS estimator. In particular,  
we first explore the case when the regressors and the residuals do not depend on the same set of factors.  
Second, we discuss the conditions under which the feasible GLS will still work in the context of dynamic panels.
Third, despite the inefficiency of the feasible GLS, we explain how substantial efficiency gains can be achieved by a multi-step version of the GLS estimator. Fourth, we  describe how consistent estimation of the asymptotic covariance matrix can be obtained.
Fifth,  we describe how to implement our estimator to cross-sectional regressions with time-varying coefficients.
Our theoretical results are corroborated by a set of Monte Carlo
experiments described in Section~\ref{MC}. An empirical application, which investigates whether firms' characteristics are relevant to individual stock returns, is presented in  Section~\ref{EMP}. 
 Section~\ref{CONCL}
concludes. The proofs of our theorems are reported in Appendix \ref{proofmain}, relying on  three  technical results, enunciated in Appendix \ref{centralemmas}. Appendix \ref{auxiliary} defines some quantities of interest for the construction of the GLS estimator, in particular regarding the (inverse of the) covariance matrix of the residuals. The Supplement contains appendices \ref{someresults}-\ref{app3}  with the proofs of additional material, that serve out main results.  

Hereafter we use the following notation.  Let $\bm{A}(R \times C)$ denotes a generic real $R\times C$ matrix with entries $a_{nm}$; in short $\bm{A}=[a_{nm}]_{n,m=1}^{R,C}$, or simply $\bm{A}=[a_{nm}]$ when the matrix's dimension is clear.  Similarly, $\bm{a}$ denotes a generic column vectors of length $R$ with element $a_n$; in short $\bm{a}=[a_n]_{n=1}^R$. The transpose of a  $\bm{A}$ is denoted by $\bm{A}'$. If $R=C$, $\lambda_1(\bm{A})$ and $\lambda_N(\bm{A})$ denote the minimum and the maximum eigenvalue of $\A$, respectively.
  With $\bm{A} > (\geq)  0 $ we mean that  $\bm{A}$ is positive definite (positive semi positive).  
Let  $\lb\bm{A}\rb_{sp}=\sqrt{\lambda_N\left(\bm{A}'\bm{A}\right)}$ 
denotes the spectral norm of $\bm{A}$, and $\|\bm{A}\|=\sqrt{\tr\left(\bm{A}'\bm{A}\right)}$, where $\tr(\cdot)$ denotes the trace, is the Frobenius norm. When $R=C$  we define the column and row norm of  $\bm{A} $ as $\| \bm{A} \|_{col} =  \max_{1\leq n \leq C} \sum_{n'=1}^C |a_{nn'}|$ and
  $ \| \bm{A} \|_{row} =  \max_{1\leq n' \leq C} \sum_{n=1}^C |a_{nn'}|$, respectively.  
Furthermore, for $C>R$, we use $\projp_{\bm{A}}=\bm{A(A'A)^{+}A'}$, where $\bm{A}^+$ denotes the Moore-Penrose generalized inverse of $\bm{A}$ and $\projm_A=\I_C-\projp_A$, where $\I_C$ is the identity matrix of dimension $C \times C $. If $\bm{A}$ has full column rank,  $\bm{A}_{\bot}$ denotes the $C\times (C-R)$ matrix satisfying $\bm{A}_{\bot}'\bm{A}=\bm{0}$, where $\bm{0}$ is a matrix of zeros, and $\bm{A}'_{\bot}\bm{A}_{\bot}=\I_{C-R}$. 
We use $ \V{\xrightarrow{p}} $, $ \V{\xrightarrow{d}}$ to denote
convergence in probability and  convergence in
distribution, respectively, and  $ \mathcal{N}\left(\bm{a},\bm{B}\right)$ denote a random vector normally distributed with mean and covariance matrix equal to $\bm{a},\bm{B} $, respectively.
  For $\bm{A}(C\times R),\bm{B}(C\times C ) $ and $\bm{C}(C \times P)$ being three random matrices, then $  \bm{\Sigma}_{A'BC}$ denotes the probability limit  (when finite) of $ C^{-1} \bm{ A'BC} $ as $C \rightarrow \infty $. For the $C\times R$ random  matrices $\bm{A}_N,\bm{B}_N$ that are functions of $N$ we write $\bm{A}_N \approx\bm{B}_N$ if $\lb\bm{A}_N-\bm{B}_N\rb\xrightarrow{p}0$ when $N\to\infty$. 
 $\sa(\bm{A})$ denotes the sigma-algebra generated by the random matrix $\bm{A}$, $\prob(\cdot)$ and $\E(\cdot)$ indicate the probability of an event and the expectation of a random variable, respectively. In the sequel, $\ka$  denotes a generic, positive constant, which need not to be the same every time we use it.

\section{Model: definitions  and assumptions} \label{def}

Assume that the observed variables obey a linear
regression model with $S$ common observed regressors  ${\bf d }_t =(d_{t1},\dots,d_{tS})'$ and $K$ heterogeneous regressors ${\bf x }_{it}=(\mathrm{x}_{it1} \hdots \mathrm{x}_{itK})'$.   
Following  the  convenient  specification put forward  by \cite{P06}, the
model for the $i$th unit can be expressed, in matrix form, as
 \be \label{eq:hetero}
\y= \d \bm{ \alpha}_{i} +  \x  \bm{\beta _{i}} + {\bf u_i},
 \ee
for an observed $T \times 1 $ vector $ {\bf y }_i = ( \mathrm{y}_{i1}, \dots, \mathrm{y}_{iT})' $,  an observed $T \times S $ matrix ${\bf D } =  ( {\bf d }_{1} \hdots {\bf d }_{T})' $ of  common
regressors, an observed $T \times K $ matrix $ {\bf
X }_i = ( {\bf x }_{i1} \hdots {\bf x }_{iT})' $ of unit-specific regressors, and
an unobserved $T \times 1 $ vector $ {\bf u_i} = ( \mathrm{u}_{i1} \hdots \mathrm{u}_{iT})' $.  In turn,  the innovation  vector satisfies the factor structure:
\begin{equation}\label{eq:uu}
  \ui =
{\bf F } {\bf b }_i + \bm{\varepsilon }_i,\qquad\textrm{with}\quad \Hi:=\E\bm{\varepsilon_{i}
\varepsilon'_{i}},
\end{equation} for an unobserved $ M \times 1 $ vector of factor loadings $ {\bf b }_{i} $, an
unobserved $T \times M $ matrix of common factors $ {\bf F } = (
{\bf f}_{1},...,{\bf f }_{T})' $  and an unobserved $ T \times 1 $
vector of idiosyncratic innovations  $ {\bm \varepsilon }_i=(
\varepsilon_{i1} \hdots \varepsilon_{iT})' $. The unit specific regressors  satisfy:
 \be \label{eq:ref_def}
  {\bf X }_i = {\bf D } {\bf \Delta }_i +
{\bf F } {\bf \Gamma }_i + {\bf V }_i, \ee for an unobserved $
S \times K $ matrix of factor loadings $ {\bf \Delta }_{i} = ( \boldsymbol{ \delta }_{i1} \hdots  \boldsymbol{ \delta }_{iS} )'$ with $\boldsymbol{\delta }_{il}=(\delta_{il1} \hdots \delta_{ilK}   )' $,
 an unobserved $ M \times K $ matrix of factor loadings $ {\bf \Gamma }_{i} = ( \boldsymbol{ \gamma }_{i1} \hdots \boldsymbol{ \gamma }_{iM}  )' $ with $ \boldsymbol{ \gamma }_{il} = ( \gamma_{il1} \hdots \gamma_{ilK}  )' $, and
an unobserved $T \times K $ matrix of idiosyncratic innovations  $ {\bf V }_i=(
\vvv_{i1} \hdots \vvv_{iT})' $ with $ \vvv_{it} = (v_{it1}, \dots v_{itK} )' $.  
 The maintained
assumption here is that $K$, $S$  and $M$ {\em do not vary  } with $T$ and
$N$. Moreover, we do not need to impose {\em any } relationship between them so that, in particular, $M$ can be either {\em smaller, equal or bigger} than $K$.  Although  model (\ref{eq:hetero}) is written as a single regression
across time for a given $i$, we assume that in fact a panel of observations $\{{\bf y}_1 \hdots {\bf y}_N, {\bf X
}_1 \hdots  {\bf X }_N  \}$  is available and fully used within our methodology.


As explained below,  throughout our analysis we always {\em de-mean } the data by $\md $. This allows to avoid making any assumptions on $\bm{ \Delta_i} $. 
   We now present  our assumptions which, thank to the detailed specification of model (\ref{eq:hetero})-(\ref{eq:ref_def}), appear relatively primitive.

\medskip

\begin{assumption}[idiosyncratic innovation $ \varepsilon_{it}  $]\label{ass eps}  The $N\times 1$ vector $\bm{\varepsilon}_{t}= ( {\varepsilon}_{1t}...{\varepsilon}_{Nt})'$ satisfies the following equation
\be \label{eq:11CHL} 
\bm{\varepsilon}_{t}= \bm{R } \bm{ a }_{t},\qquad\qquad \textrm{for}\quad t=1,\dots T,
\ee
where  the $N\times N$ matrix  of constants $\bm{R} = [ r_{ij} ] $ satisfies $ \parallel \!\! \bm{R} \!\!  \parallel_{row} +\parallel \!\!  \bm{R} \!\!  \parallel_{col}  \,   < \infty $, $\inf_i\sum_{j=1}^N |r_{ij}|>\ka$  for some $\ka>0$, 
 and the elements of  the $N \times 1 $ vector $
  \bm{ a_{t}}  =  \left(a_{1t},a_{2t},\dots,a_{Nt}\right)' $ follow a linear process:
\be \label{eq:11} 
a_{it}  =\sum_{s=0}^{\infty} \phi_{is}\eta_{i,t-s}, \quad \sup_i\sum_{s=0}^{\infty}s^2|\phi_{is}|<\infty,\quad \textrm{with}\quad \phi_{i0}=1,
\ee
 where the sequence $\{\eta_{it}\}$ is independent and identically distributed across $i$ and $t$ with $\E \eta_{it}=0 $ and $ \E | \eta_{it} |^6< \infty $. Moreover, 
for every complex number $z\in\mathbb{C}$, 
\be\label{eq:11ap}
\inf_i|\phi_i(z)|>\ka, \quad |z|\leq 1, \qquad \textrm{where}\quad\phi_i(z)=\sum_{s=0}^{\infty} \phi_{is}z^s.
\ee
\end{assumption}
\medskip

\begin{remark}
Assumption \ref{ass eps} is similar to Assumptions 1 and 2 in \cite{PT11} and, with same variations, 
this form of  cross-sectional and time dependence has been adopted  also by  \cite{MW09}, \cite{MW13} and \cite{Onatski2015388}. 
The above assumption turns out to be extremely convenient for establishing the asymptotic distribution of the feasible and unfeasible GLS estimators along the lines of 
Theorem 1 in \cite{RH97}.
\end{remark}

\begin{remark}\label{remcum}  Assumption \ref{ass eps} implies that, for every $ 2 \le h,\ell\le 6$:
$$
\sup_{i_1} \sup_{t_1}   \sum_{i_2 \cdots  i_{\ell}=1}^N  \sum_{t_2 \cdots t_h =1}^T | \cum_h  (\varepsilon_{i_1 t_1}, \varepsilon_{i_2 t_2} \cdots ,   \varepsilon_{i_{\ell}  t_h})| < \infty ,  
$$
where the summands are the cumulants of order $h$ of $\varepsilon_{i_1 t_1}, \varepsilon_{i_2 t_2} \cdots ,   \varepsilon_{i_{\ell} t_h}$.
\end{remark}
\medskip
\begin{remark}\label{unif lower bound}  By \cite{BD91}, Proposition 4.5.3,  (\ref{eq:11ap}) implies that the eigenvalues of the covariance matrices of $\bm{a_j}=(a_{j1},\dots a_{jT})'$ are bounded, and greater than $\ka$ for every $i$. Easy calculations give
$\Hi=  \sum_{j=1}^N r_{ij}^2  \E \bm{a}_j \bm{a}_j'$,
implying that $ \inf_i \lambda_1(\Hi)>\ka $ and $ \sup_i \lambda_T(\Hi)<\infty$.
\end{remark}

\medskip

\begin{assumption}[regressor innovation $\v $]\label{ass V}

 The sequence $\{{ v}_{itk}\} $ have zero mean,  $ \sup_i \sup_k \sup_t  \E |  \!  v_{itk} \!  |^{12} < \infty $ and  they  satisfy, for every $2  \le h,\ell,s\le 14$ and $2\le j\le h$:
$$
  \sup_{k_1 \cdots k_s }  \sup_{i_1} \sup_{t_1}   \sum_{i_2 \cdots  i_{\ell} =1}^N  \sum_{t_2 \cdots t_h =1}^T (1 + t_j^2)  | \cum_h  (v_{i_1 t_1 k_1}, \cdots ,   v_{i_{\ell}  t_h,k_s})| \le\infty .
$$
Moreover,  $\inf_i\lambda_1\left(\E\bm{\mathrm{v}}_{it}\bm{\mathrm{v}}'_{it}\right)>\ka$, where $\bm{\mathrm{v}'_{it}}=(v_{1it},\dots,v_{Kit})$. 
\end{assumption}
\begin{remark}\label{inverseV}
Assumption \ref{ass V} implies that $T^{-1}\v'\v\xrightarrow{p}T^{-1}\sum_{t=1}^T\E\bm{\mathrm{v}}_{it}\bm{\mathrm{v}}'_{it}=: \SO_{\v'\v}$, with $\inf_i \lambda_1\left(\SO_{\v'\v}\right)>\ka$. It follows that $\sup_i\lb\left(\v'\v/T\right)^{-1}\rb=O_p(1)$.
\end{remark}
\begin{remark}
The $ v_{itk} $ can be interpreted as the high-rank components of the regressors $\mathrm{x}_{itk}$, adopting  \cite{MW13} terminology, as opposed to the $ \d $ which represent the low-rank components. For instance, if for each $k$ the  $ v_{itk} $ are generated as $\varepsilon_{it}$ in Assumption \ref{ass eps}, one obtains $\bm{V_k}=\left[v_{kti}\right]_{t,i=1}^{T,N}= O_p(  \sqrt{ \max ( N,T) } ) $
 for every $k$  (see the discussion in \cite{MW13}, Appendix 1 and \cite{Onatski2015388}).   In contrast, $ (  \sum_{i=1}^N  \parallel \!\! \d \!\!  \parallel^2 )^{1 \over 2} =
 (  N  \parallel \!\! \d \!\!  \parallel^2 )^{1 \over 2} =    O_p(\sqrt{(NT} )   $.
\end{remark}

\medskip

\begin{assumption}[latent  and observed  factors]\label{ass factors}
Set $\z = ( \d , \f ) = [ \mathrm{z}_{ti} ]$
for $ 1 \le t \le T $ and $ 1 \le j \le M+S < \infty  $.  
Then, 
%
\be\label{eq:maggio1}
\frac{\z'\z}{T}\xrightarrow{p}\bm{\Sigma_{\z'\z}},\quad\textrm{with}\quad \bm{\Sigma_{\z'\z}}:=\left[
\begin{array}{cc}
\bm{\Sigma}_{\d'\d} & \bm{\Sigma}_{\d'\d}\\
\bm{\Sigma}_{\f'\d} & \bm{\Sigma}_{\f'\f}
\end{array}
\right]>0,
\ee
and $ \bm{\Sigma}_{\d' \d}>0 $,  $ \bm{\Sigma}_{\f' \f}>0 $.  Moreover, we assume
$
\E\lb\bm{\mathrm{z}_{t}}\rb^4<\infty.
$
where  $\bm{\mathrm{z}_{t}}=(\mathrm{z}_{t,1},\dots,\mathrm{z}_{t,M+S})'$
\end{assumption}
\medskip
\begin{remark}\label{fpd}
Equation (\ref{eq:maggio1}) implies that $\f'\md\f>0$ (see \cite{ltk}, Result (4), Section 9.11.2).
\end{remark}
\medskip
\begin{remark}
Although not strictly necessary, we are ruling out  trending behaviours in  $ \d $ and $ \f $. However, $\d$ and $\f $ are allowed to be cross-correlated as well as serially correlated. although not perfectly collinear.  For instance, the joint dynamics of $\z $ could  be described by  a  multivariate stationary ARMA.
\end{remark}
\medskip

\begin{assumption}[regressors]\label{ass X} For every $i$, the matrix of unit specific regressors $\x $  and the matrix of common regressors $\d $ have full row rank. Moreover, setting $ \zi := [ \d , \x ] $,
 $ N^{-1}\sum_{i=1}^N \mzi \ui  \ui' \mzi$ has always rank $T-S$ for  sufficiently large $N$ and $T$. 
\end{assumption}
\begin{remark}\label{remark ass X}
Assumption \ref{ass X} requires enough cross section heterogeneity of the $\x$'s across individuals. Simple manipulations show that
$$
\db'\left(\frac{1}{N}\sum_{i=1}^N \mzi \ui  \ui' \mzi\right)\db=\frac{1}{N}\sum_{i=1}^N \bm{M_{(\db'\x)}} \db'\ui  \ui'\db \bm{M_{(\db'\x)}}>0 , 
$$
implying that the empirical covariance matrix $\sh$ defined in (\ref{eq:charlie2}) is invertible.
\end{remark}
\medskip

\begin{assumption}[loadings $\b$ and $ {\bf \Gamma }_i  $]\label{ass loading}
$\Gi$ and $\b$ are non-random such that $\lb \Gi \rb<\infty$ and $\lb\b\rb<\infty$ and, for $N> M$,
\be
\B:={1 \over N } \sum_{i=1}^N \b \b' >0 . \label{eq:seconmomb}
\ee
and
\ber
\an&:=&\frac{1}{N}\sum_{i=1}^N\left(\I_M-\Gi\psinv\Gi'\frac{\f\md\f}{T}\right)
\b\b'
\left(\I_M-\frac{\f'\md\f}{T}\Gi\psinv\Gi'\right), \label{eq:seconmomb2}
\eer
is positive definite with
\be\label{eq:trenta}
\bm{\Psi}_i:=\Gi'\frac{\f'\md\f}{T}\Gi+\SO_{\v'\v}.
\ee

\end{assumption}
\begin{remark}
 Condition (\ref{eq:seconmomb}) implies that the factor structure (\ref{eq:uu}) is strong, as defined in \cite{PT11}.  This is commonly assumed in the literature.
The technical condition (\ref{eq:seconmomb2}) is used in the proof of Theorem \ref{Theorem_GLS}. As shown  in Section \ref{lemmata1} in the Supp. Material,  Lemma \ref{lemmata1}\ref{zerone}, the matrices in brackets are of full rank. Hence, (\ref{eq:seconmomb2}) will be satisfied when there is enough cross-sectional heterogeneity in the sample. Finally, our results will not change if random loadings are assumed (and cross-sectionally independent from other parameters). 
\end{remark}
\begin{assumption}[independence]\label{independence} The $f_{mt}, v_{ksi},\varepsilon_{uj} $ are mutually independent for every $i,j$ and $t,s,u$ and $m,k$.
\end{assumption}
\begin{remark}
We are not allowing for any correlation between any entries of $ \epj$ and  $ \x $. This rules out  the possibility that $ \x $ contains a weakly exogenous component, and in this respect we are similar to \cite{P06} and \cite{B09}.   The implications from generalizing this assumption, in particular when considering dynamic panels where one element of $\x$ represents the lagged dependent variable, are discussed in Section~\ref{dynamic}. 
\end{remark}

\begin{remark}\label{mreg}
Assumptions \ref{ass V}, \ref{ass factors} and \ref{independence} and Remark \ref{fpd}  imply that $T^{-1}\x'\x\xrightarrow{p}\bm{\Sigma}_{\x'\x}>0$ and   $T^{-1}\x'\md\x\xrightarrow{p}\bm{\Sigma}_{\x'\md\x}>0$,  for every $i$.  Hence $\lb(T^{-1}\x'\x)^{-1}\rb=O_p(1)$ and  $\lb(T^{-1}\x'\md\x)^{-1}\rb=O_p(1)$ for $T$ large enough. 
\end{remark}
\section{Estimators: definitions and asymptotics}  \label{res}

Our main objective is to estimate the heterogeneous slope coefficients $ \bm{ \beta }_{i}$ of (\ref{eq:hetero}). However,  estimation of the coefficients  $ {\bm \alpha }_{i} $ of the common regressors is also discussed in Section \ref{common}. Hence, without loss of generality,  we premultiply both sides of 
(\ref{eq:hetero}) by the projection matrix $\md$, obtaining
\ber
\md  \y=  \md \x\bm{\beta_i}+\md \ui . \label{eq:mm11} 
\eer
We consider three different estimators for the  parameters $ {\bm \beta }_{i}$, namely the OLS, the unfeasible  and  feasible GLS estimators.
Regarding the OLS estimator for $ \bm{ \beta _{i} }$:
\begin{equation}\label{eq:charlie1}
 \bm{\hat{ \beta }_i^{OLS}} := ( \x' \md \x )^{-1} \x \md \y.
\end{equation}

We now  consider   GLS estimation.   
Define the   cross-sectional averages of the individual covariance   matrices of the $\md \ui$, conditional on sigma algebra generated by $ \bm{Z}$, defined in Assumption \ref{ass factors}:
\be\label{eq:deffeb}
\md\sth\md,\;  \mbox{  setting  }\;\sth:= \f\B \f'+ \Hn,\quad\mbox{with}\quad
\Hn:=\frac{1}{N}\sum_{i=1}^N\Hi.
\ee
We assume  without loss of generality  that $\bf{d_t}$ includes an element equal to one, i.e. we allow for an intercept term,  leading to  $ \E \ui =0 $. 
The  presence of $\md $  could cause some complications in the definition of the GLS estimator since  the $\md \ui $ have a singular covariance matrix. We  show how to solve this issue and obtain  a model with a  non-singular residual covariance matrix that can be used to construct the  GLS estimator.

Proceeding along the lines of \cite{MN88}, Section 11 in Chapter 13, one gets the  UGLS  estimator  when the  residual covariance matrix to model (\ref{eq:mm11}) is singular:
\begin{equation}\label{eq:11giugno}
\bm{\hat{\beta}}_i^{UGLS}:=\left(\x'\md \left(\md \sth \md\right)^+ \md  \x\right)^{-1}\x'
\md  (\md\sth\md)^+ \md \y.
\end{equation}
By Lemma \ref{fact6410} in the Supp. Material
$$ (\md\sthi \md)^+ =\db\left(\db'\sth \db\right)^{-1}\db' , $$ where $\db$ is the $T\times (T\!-\!S)$ full rank matrix such that
 $ \md = \db \db' $ where $\db'\db =\I_{T-S} $. Assumption \ref{ass X} and display (\ref{eq:deffeb}) imply that the inverse in (\ref{eq:11giugno}) is well defined for  any $T$. By substitution, setting  for simplicity 
\begin{equation} \label{eq:storti}
\yy=\db' \y,\quad \xx= \db'\x,\quad 
\ei=\db' \epi,\quad   \ff=\db' \f,\quad    \uu=  \db'  \ui, 
\end{equation}
  one obtains
 \ber
\bm{\hat{\beta}}_i^{UGLS} && 
 = \left(\x'\db
\left(\db'\sthi\db\right)^{-1}
\db'\x'\right)^{-1}
\x'\db\left(\db'\sthi \db\right)^{-1}\db\y \nonumber\\
&& =  \left(\xx'  \ssthi \xx \right)^{-1}
\xx' \ssthi \yy,\label{eq:UGLSdef}
\eer
where we set  
$
\ssth= \db'\sth\db.
$
This means that the UGLS has now the more  conventional expression of the  generalized least squares  for the model
\ber\label{eq:modstorto}
\yy=\xx\bm{\beta}_i+\uu,\qquad\textrm{with}\qquad \uu=\ff\b+\uu,
\eer
 without involving Moore-Penrose matrices.
 Pre-multiplying the data by  $ \db'  $ reduces the sample size by $S$ units since now  the $ \yy $ and the $ \xx $ have $T-S$ rows. Likewise, considering again model (\ref{eq:modstorto}), an equivalent representation of (\ref{eq:charlie1}) is $  \bm{\hat{\beta }_i^{OLS}} = (\xx ' \xx )^{-1} \xx ' \yy   $.

Along the same lines,  our proposed   {\em feasible } GLS estimator is given by 
\begin{equation}\label{eq:bgls}
\bm{\hat{ \beta }_i^{GLS}} := \left(\xx'  \shi \xx \right)^{-1}
\xx' \shi \yy,
\end{equation}
where
 \begin{equation}\label{eq:charlie2}
 \sh:=
N^{-1} \sum_{i=1}^N \huu\huu'   , \,\,\, \mbox{ with } \huu:=
 \yy-\xx \bm{\hat{\beta}}_i^{OLS} = \mxx \uu,
\end{equation}
for $N$ and $T$  large enough, by Assumption~\ref{ass X} and Remark~\ref{remark ass X}, $ \sh $
 has full rank.
The following two theorems  enunciates the asymptotic distribution of the OLS,  UGLS and GLS estimators, respectively. The  proofs are given in  Appendixes \ref{thuno} and \ref{proofT2}, respectively. Further details are provided in the Supp. Material.

\begin{theorem}  \label{Theorem_OLS} 
 When  Assumptions \ref{ass eps}, \ref{ass V}, \ref{ass factors}, \ref{ass X}, \ref{ass loading} and \ref{independence} hold, for any $N$ and as $T\to\infty$

\noindent
(i)   (OLS estimator)
$$
T^{1 \over 2 } ( \bm{\hat{ \beta }}_i^{OLS} - \bm{ \beta }_{i} - {\bf
\tau }_i^{ OLS } ) \xrightarrow{d}
\mathcal{N}\left(\bm{0},\bm{\Sigma_i}\right),\quad
$$
where 
\begin{equation}\label{eq:bias17}
  {\bf \tau }_i^{OLS} :=  \sxi  \sxf  \b , 
  \end{equation}
is the bias term, and the asymptotic covariance matrix equals
\begin{equation}\label{eq:asycm17}
\bm{\Sigma_i}
 :=  \sxi \bm{\Sigma_{  \xx '  \db'\Hi\db \xx  }}   \sxi ,
\end{equation}
setting
\ber
 \bm{  \Sigma }_{\xx' \ff }  &:=& \Gi' \Big( \bm{\Sigma }_{\f'\f} - \bm{\Sigma }_{\f'\d} \bm{\Sigma }_{\d'\d}^{-1} \bm{\Sigma }_{\d'\f} \Big)  , \label{eq:sigmaXF}  \\
 \bm{  \Sigma }_{\xx' \xx }  &:=&    {\bf \Gamma}_i'  \Big( \bm{\Sigma }_{\f'\f} - \bm{\Sigma }_{\f'\d} \bm{\Sigma }_{\d'\d}^{-1} \bm{\Sigma }_{\d'\f}   \Big)  {\bf \Gamma}_i + {\bf \Sigma }_{ \v'\v} , \label{eq:sigmaXX}  \\
 \bm{ \Sigma }_{\xx'  \db'\Hi\db \xx }   &:=&  \Gi ' ( -  \bm{ \Sigma }_{\f'\d} \bm{ \Sigma }_{\d'\d}^{-1} , \bm{  I }_m  )  \bm{ \Sigma }_{\z' \Hi \z}  ( -  \bm{ \Sigma }_{\f'\d}  \bm{ \Sigma }_{\d'\d}^{-1} , \bm{ I }_m )' \Gi \label{eq:sigmaXHX}\\
 &&+  \bm{ \Sigma }_{\v' \Hi \v}  .  \nonumber 
 \eer

\medskip

 \noindent (ii)  (UGLS estimator)
 \be\label{eq:theorem1part2}
T^{1 \over 2 } \left( \bm{\hat{ \beta }_i^{UGLS}} - \bm{ \beta }_{i}   \right) \xrightarrow{d}
\mathcal{N}\left(\bm{0},\SO^{\star}_N\right),
\ee
with $\SO^{\star}_N:=\SO_{\v'\Hninv\v}^{-1}\SO_{\v'\Hninv\Hi\Hninv\v}\SO_{\v'\Hninv\v}^{-1}$.

 \end{theorem}  
\medskip

\begin{remark}
  The OLS estimator is affected by a {\em first-order bias}. It  will be asymptotically unbiased if either $ \b = \bm{0} $ or $ \Gi = \bm{0} $ or, alternatively, for diagonal $\bm{ \Sigma }_{\ff' \ff } $ as well as  with $ \Gi $  and $ \b $  satisfying $ \gamma_{il} b_{il} = 0 $ for every  $l$ and $i$. Essentially, this means that the entries of $ \Gi $ are non zero whenever the corresponding entries of $ \b $ are zero, for the same row $l$, and viceversa. More in general, 
no bias arises if  ${\bf b}_i $ belongs to the null space of  $  \bm{  \Sigma }_{\ff' \ff } \Gi  $, assuming  $M>K$.
\end{remark}

\begin{remark}
  One can assume without loss of generality that the same latent factors $\f $ enter into $ \bm{X}_i $ and  $\ui $. In fact, assume  $ \ui =  \g  \bm{b}_i + \bm{\varepsilon }_i  $ with the rows of $ \g $ correlated, but not identical to the rows of $ \f $. Then the bias  takes the form
\[
  {\bf \tau }_i^{OLS} =  \sxi  \bm{\Sigma }_{\f' \projp_F \g}  \b ,
\]
exploiting the decomposition $ \g = \projp_F\g+ \projm_F\g$. Hence the bias will only be non-zero due to the portion of $\g$ correlated with $\f$. The same consideration applies to the GLS estimator. In Section~\ref{diff-fac} we explore more in details the implications of having different, yet correlated, factor structures for regressors and innovations. 
\end{remark}

\bigskip


\begin{remark}
The UGLS estimator is asymptotically unbiased, consistent and asymptotically normal as $T \rightarrow \infty $. Moreover,   the  UGLS estimator can be  efficient in the GLS sense. In particular, when 
the $ \bm{\varepsilon }_i  $ are not (unconditionally) heteroskedastic, namely   $  \Hi = {\bm \Xi } $, then the UGLS  asymptotic covariance matrix does  not have the  sandwich form,  unlike for the OLS estimator. One can define the UGLS differently, for instance replacing $\sth $   with  $\f\B \f'+ \Hi $ in (\ref{eq:11giugno}). However, our  definition of the UGLS estimator 
makes it closer to the population counterpart to the class of feasible GLS estimators  here studied.
\end{remark}

We now present the main result of the paper.

\begin{theorem}  \label{Theorem_GLS} When  Assumptions \ref{ass eps}, \ref{ass V}, \ref{ass factors}, \ref{ass X}, \ref{ass loading} and \ref{independence}  hold,  
as $ 1 / T   +  T /  N  \rightarrow 0 $, 
\be\label{eq:theorem2A}
  \bm{\hat{ \beta }_i^{GLS}} \xrightarrow{p}\bm{\beta _{i}},\qquad    
\ee
and, as $(1/T)+(T^3/N^2)\to 0$, then
\be\label{eq:theorem2B}
\left(\v'\MAIN\Hi\MAIN\v\right)^{-\frac{1}{2}}
\left(\v'\MAIN\v\right)
 ( \bm{\hat{ \beta }_i^{GLS}} - \bm{ \beta _{i}}   ) \xrightarrow{d}
\mathcal{N}\left(\bm{0},\I_{\bm{K}}\right)  ,
\ee
%
%
%
%
%
where
\be\label{eq:MMM2}
\cn:= \frac{1}{N}\sum_{i=1}^N \left(\Hi+\bm{\Theta_i}\right)  \mbox{with}\;
\bm{\Theta_i}:=\E\left[\v\sxi \Gi\sf\b\b'\sf\Gi'\sxi \v'\right],
\ee
with $\sx$ defined  in  (\ref{eq:sigmaXX}).
\end{theorem}

\begin{remark}
 The GLS estimator is asymptotically unbiased, consistent and asymptotically normal as both $N,T \rightarrow \infty $ such that $ T^3/N^2 \rightarrow 0 $.  The  feasible GLS estimator is not efficient in  general. A multi-step generalization achieves substantial efficiency gains, see Section~\ref{various}.
\end{remark}

\section{Common regressors} \label{common}

We now  consider estimation of the coefficients  $ \bm{ \alpha }_{i} $ to the common regressors  $ \d  $ in model (\ref{eq:hetero}). A natural generalization of the GLS estimator  would be
\[
\Big(   \begin{array}{c}  \tilde{\bm \alpha }_i^{GLS}  \\  \tilde{\bm \beta }_i^{GLS}  \end{array}  \Big) := \left(\zi'  \tilde{\bm{S}}^{+}_N\zi \right)^{+}
\zi ' \tilde{\bm{S}}^{+}_N \y,
\]
where $ \zi$ has been defined in Assumption \ref{ass X} and 
\be\label{eq:gidio}
 \tilde{\bm{S}}_N = N^{-1} \sum_{i=1}^N \hui\hui',\qquad  
  \hui:=  
  \y - \d \hat{\bm{ \alpha }}_i^{OLS} -  \x \hat{\bm{ \beta }}_i^{OLS} ,
\ee
that is $\hui$ are the OLS residuals, for  $ (  \hat{\bm{ \alpha }}_i^{OLS'}  ,  \hat{\bm{ \beta }}_i^{OLS'}  )':= (\zi' \zi )^{-1} \zi'  \y$.
 However, we show in Theorem \ref{TI1} (Supp. Material, Appendix~\ref{app2})  that $  \tilde{\bm{ \alpha }}_i^{GLS} =  {\bm 0}_M  $ and $   \tilde{\bm{ \beta }}_i^{GLS} =    \hat{\bm{ \beta }}_i^{GLS} $ due to a cancellation that occurs as a consequence of $ \d $ being common across units. 
If the joint distribution for the estimators of $   ( {\bm{ \alpha }}_{i}' ,  {\bm{ \beta }}_{i}' )' $  is not required, one can estimate the   $ \bm{ \alpha }_{i} $
 as the projection of
 $ \y- \x \hat{\bm{ \beta }}_i^{GLS} $ on $ \d $ yielding
\be\label{eq:gidio2}
\widetilde{\bm  \alpha }_i := (\d' \d )^{-1} \d' (  \y - \x \hat{\bm{ \beta }}_i^{GLS}  ) .
\ee
Using our theory, its asymptotic distribution follows (see Theorem \ref{TI2} ins the Spp. Material for further details).
 Note that the additional assumption $ \f ' \d = {\bm  0 } $ is required. For example, if we are interested in a model with  an intercept term, heterogenous across units, 
 such as $ \d  {\bm \alpha }_{i} = \bm{ \iota }_T \alpha_{i1} + \d_2 \bm{ \alpha }_{i2} $, with $ \bm{ \alpha }_{i} = (\alpha_{i1} ,\bm{ \alpha }_{i2} ')' , \d = ( \bm{ \iota }_T , \d_2 ) $,
 then one of the restrictions $\f' \d= {\bm 0 } $  is simply $ \sum_{t=1}^T \bm{f}_t = \bm{0} $.  If, moreover, a grand-mean is  also allowed for, such as $ \d \bm{ \alpha }_{i} = \bm{ \iota }_T  \alpha_{3} + \bm{ \iota}_T \alpha_{i1} + \d_2 \bm{ \alpha }_{i2} $, then the additional
 restriction $ \sum_{i=1}^N  \alpha_{i1}  = 0 $ is needed.
 Similar identification conditions are discussed in \cite{B09} and \cite{MW09}.\footnote{Most of the papers on estimation of panel regressions with so-called interactive fixed effects, such as ours, focus exclusively on the coefficients to the  heterogeneous time-varying regressors. Among the few exceptions, is \cite{B09} who shows that,  without further identification 
assumption,  estimation of the coefficient to common regressors is possible only for constant parameters.  In contrast,  for the case of  non-constant   coefficients further  identification assumptions  similar to ours are needed.  \cite{MW09} study the same estimator of \cite{B09} under weaker conditions on the regressors, allowing for instance for pre-determinatedness. Our identification  condition for the coefficients to common regressors implies their weaker  corresponding assumption. They focus exclusively on the case of constant regression coefficients.}

If instead the joint distribution for estimators of $ \bm{\alpha }_{i} $ and $ \bm{ \beta }_{i} $ is required, this can be achieved by  a slight modification of our GLS estimator, namely
\be
\left(   \begin{array}{c}  \breve{\bm \alpha }_i^{GLS  }  \\  \breve{\bm \beta }_i^{GLS}  \end{array}  \right) := \left(\zi'  \breve{\bm{S}}_N^{-1} \zi \right)^{-1}
\zi' \breve{\bm{S}}_N^{-1} \y, \label{dagger}
\ee
for the non-singular matrix
\be\label{eq:gidio3}
\breve{\bm{S}}_N := \tilde{\bm{S}}_N + \left( { \tr( \tilde{\bm{S}}_N ) \over N }\right) \pd .
\ee
Non-singularity of  $\breve{\bm{S}}_N $ follows by augmenting the matrix $\breve{\bm{S}}_N $, of rank $T-S$, with the projection matrix $ \pd $ of rank $S$. Scaling by 
$N^{-1}\tr\left( \tilde{\bm{S}}_N\right)   $ in not required by the asymptotic theory but could be relevant in finite-samples to ensure the same order of magnitude of the two terms in $\breve{\bm{S}}_N  $.  It turns out that  the same identification condition  $ \f ' \d = {\bm  0 } $, discussed above, is required.  Monte Carlo experiments are reported in Section~\ref{MC} to assess the small-sample properties of these estimators.

\section{Technical contributions} \label{LEMMI}
The asymptotics for the GLS estimator requires four key auxiliary results, enunciated in Appendix \ref{centralemmas}, which   could be useful in a broader set of statistical problems. 
 The main reason for this  complexity is that, unlike most of the existing  theoretical results on GLS estimation, we are not restricting the number of free elements  of the  weighting matrix to be finite. Indeed, in our case the number of free elements of the weighting matrix is $O(T^2)$ and hence rapidly increasing with $T$. To tackle the curse-of-dimensionality issue, we exploit the approximate factor structure of the weighted matrix, that we write as  $\bm{E=FAF'+C}$, for  (possibly random)  $ M_1 \times M_1 $  matrix $ \bm{A}>0 $ and $ T \times T $ matrix $ \bm{C} > 0 $ for every finite $T$.  The inverse has a convenient form  thanks to the Sherman-Morrison formula (see Appendix \ref{smw}). 
 
Lemma \ref{PZ} establishes the asymptotic orthogonality between the inverse of the   matrix $\bm{E}$, and the factor $\f$. More precisely, when $ \bm{A} $ and $ \bm{C} $ satisfy a set of mild regularity conditions
\be
\parallel \bm{E}^{-1} \bm{F} \parallel^2 = O_p (T^{-1}),  \label{asyorth}
\ee
This is a remarkably fast  rate given that 
$\bm{E}^{-1} \f $ is $T \times M_1 $ dimensional, with $M_1$ fixed, hence with its number of rows increasing with $T$. It implies that for a large class of  $ T \times M_2 $ matrices $\bm{P}$ (satisfying the mild regularity conditions of Lemma \ref{PZ}) possibly unrelated to both $\bm{E}$ and $\bm{F} $, then  $ \bm{P}' \bm{E}^{-1} \bm{F} = O_p(1)  $ and, when the entries of $\bm{P}$ have zero mean and are stochastically independent of $\bm{E} $ and $\bm{F} $, then  $ \bm{P}' \bm{E}^{-1} \bm{F} = O_p(T^{-\nicefrac{1}{2}}) $. These rates are very different from the usual case, arising when $ \bm{E} $  and $ \bm{F} $ are unrelated. For example, when $ \bm{A}=\bm{0} $, under the same assumptions on $\bm{P}$, one gets  that  $ \bm{P}' \bm{E}^{-1} \bm{F}   $ is of order $ O_p(T^{\nicefrac{1}{2}} )  $ or $ O_p(T) $,  depending  on whether  $ \bm{P}$ has zero or non-zero mean, respectively, assuming that $ \bm{F} $, $\bm{A} $  and $\bm{P}$ are mutually independent. 

 The  asymptotic orthogonality  (\ref{asyorth}) plays a crucial role in establishing the asymptotics for the GLS (and UGLS)  estimator. To better understand this, consider the following decomposition of the GLS estimator:
\ber
\hat{\bm \beta }_i^{GLS}  -  {\bm \beta }_{i}  && =  \left(\xx'  \shi \xx \right)^{-1}\nonumber
\xx' \shi  \uu \\ &&  =   \left(\xx'  \shi \xx \right)^{-1}
\xx' \shi  \ff \b +  \left(\xx'  \shi \xx \right)^{-1}\label{eq:franzferd}
\xx' \shi \ei.
\eer
Since $   \sh
 = \ff \han \ff' + \bm{\hat{\sc}_N} $, for some random matrices $\han, \bm{\hat{ \sc}_N} $ (specified in Appendix~\ref{auxiliary})  function of both $T$ and $N$,  we show that Lemma \ref{PZ}  applies to  the bias term, namely the first term in (\ref{eq:franzferd}). In particular, by (\ref{asyorth}) term $    \xx' \shi \ff  $ 
 is of a smaller order of magnitude (and vanishes asymptotically) than $\xx' \shi  \ei $ as long as $N$ is diverging faster than $T$.
 Note that both the dimension of  $ \bm{\hat{ \sc}}_N $, as well as its elements,  are changing with  $T$.   The faster rate for $N$ is demanded  for by the need to  have  $\han $ and $ \bm{\hat{\sc}}_N $ with  the desired limiting properties. 

To save notation we rename  $\bm{\hat{\beta}_i^{GLS}}$ in (\ref{eq:bgls}) as $\bm{\hat{\beta}_i(\sh)}$ and define $\bm{\hat{\beta}_i(\MA)}$ the estimator obtained replacing the weighting matrix $\sh$ with $\MA$ in (\ref{eq:bgls}). 
The matrix  $ \MA =\ff \an \ff' + \bm{\mathcal{C}_N}$, defined in Appendix~\ref{proofT2} (display (\ref{eq:MMM1})),
 is non-stochastic, if $\z$ is fixed, and, most importantly,  satisfies the assumption of Lemma \ref{PZ}. Existing conditions  for the asymptotic equivalence of $\bm{\hat{\beta}_i(\sh)}$ and $\bm{\hat{\beta}_i(\MA)}$ cannot be applied here. Although  the inverse operator of a matrix is an analytic function,  one cannot  rely on the delta method due to the curse of dimensionality, namely the fact that the elements as well as the size of $\sh $ are varying with $T$ and $N$. For similar reasons, element-wise convergence of $\sh$ cannot be combined with the Slutsky's Theorem, as discussed in \cite{mf94}. Considering the absolute convergence $\lb\shi-\MAI\rb_{sp}$ and use random matrix theory is not a viable option. Even if we were able to obtain the optimal convergence rate $\lb\sh-\MA\rb_{sp}=O_p(\sqrt{T/N})$ established by for i.i.d. data, it would be not sufficient to obtain $\sqrt{T}\left(\bm{\hat{\beta}_i(\sh)}-\bm{\hat{\beta}_i(\MA)}\right)=o_p(1)$ without strengthening  the Assumptions of Theorem \ref{Theorem_GLS}. The latter convergence requires indeed proving that $T^{-\nicefrac{1}{2}}\xx'\left(\shi-\MAI\right)\uu=o_p(1)$. We would have
\be\label{eq:casa}
\lb T^{-\nicefrac{1}{2}}\xx'\left(\shi-\MAI\right)\uu\rb \leq \lb T^{-\nicefrac{1}{2}}\xx\rb
\lb \shi-\MAI \rb_{sp} \lb \uu \rb =O_p\left(\frac{T}{\sqrt{N}}\right).
\ee
 To prove our results we found convenient proceeding in two steps: firstly we prove that 
$T^{-\nicefrac{1}{2}}\xx'\shi\uu$
is asymptotically equivalent
$T^{-\nicefrac{1}{2}}\xx'\SSI\uu$, with $\SS$ as defined in  Appendix~\ref{auxiliary}, equation (\ref{eq:silver}) ,and secondly that the latter term is in turn asymptotically equivalent to 
$T^{-\nicefrac{1}{2}}\xx'\MAI\uu$.
Lemma \ref{approx inv}, a simple extension of a well known result in matrix algebra, entails that
\ber\label{eq:formula}
\shi 
= \SSI  - \SSI   ( \sh  - \SS) \SSI+ 
\SSI  (   \sh  - \SS ) \shi    ( \sh   - \SS ) \SSI,
\eer
implying that $
\xx' \shi  \uu  $ can be re-written as the (algebric) sum of $ \xx' \SSI  \uu  $, $ \xx' 
\SSI (  \sh  - \SS ) \SSI  \uu $ and 
 $ \xx' 
\SSI   (   \sh  - \SS ) \SSI    (   \sh  - \SS )
\SSI  \uu  $. In the proof  we show that the second last term is of order $O_p(T/N^{\nicefrac{1}{2}}) $ and that  the last term is of order $O_p(T^2/N) $, respectively. Instead, the first leading term will exhibit the usual $O_p(T^{\nicefrac{1}{2}}) $ rate of convergence.  For the second and third  terms to be asymptotically negligible, in terms of asymptotic distribution, one requires that $T^{3 \over 2}/N $ goes to zero as $T$ increases. Unfortunately our approach requires lengthy calculations involving high order cumulants that are bounded using the diagram formula (see Appendix \ref{auxiliary1} in the Supp. Material). 

A further step necessary to derive the convergence in distribution of the estimator is the derivation of the asymptotic distribution of $T^{-\nicefrac{1}{2}}\xx'\MAI\uu$. 
We first show that the latter term is equivalent to $T^{-\nicefrac{1}{2}}\v'\cni\epi$. In order to exploit the results in \cite{RH97}, the absolute row/column summability of $\cni$ needs to be shown. To accomplish this task we first approximate  $\cn$ with a circulant symmetric matrix, as further discussed in the Supp. Material, Appendix \ref{invcov}. By Lemma \ref{paolo}, that extends a result in \cite{Z93}, we prove that the inverse of the latter matrix has indeed bounded row norm.

\section{Discussion and generalizations}  \label{various}
 In this section we describe various generalizations to our framework. In particular, we explain the consequences of allowing for different, yet related, factor structures in the regressors and 
 residuals,  respectively. Then, we show to derive a consistent estimator for the asymptotic covariance matrix of the GLS estimator. We also discuss how to achieve efficiency improvements 
 of the GLS by an iterative procedure. Finally, we explain how our results apply to cross-sectional regressions with time-varying coefficients. 
 
\subsection{Different factor structures}  \label{diff-fac}

So far we have assumed that the unit-specific regressors $\x$ and the true residuals $\ui $ of (\ref{eq:hetero}) share the same common, latent, factors.  We now explore the implications of allowing that possibly different, yet correlated, set of factors affect the regressors and the residuals, respectively.  Let us here illustrate the UGLS case and then provide more details in the Supp. Material (Appendix~\ref{app3}) for the (feasible) GLS. 
To simplify the exposition we assume  that $\d=\bm{0}$, and
\be\label{eq:diffact}
\y=\x\bm{\beta}_{i}+\ui, \quad \x=\f_1\Gi+\v, \quad \ui=\f_2\b+\epi ,
\ee
where $\f_1(T\times M_1)$ and $\f_2(T\times M_2)$ satisfy Assumption \ref{ass factors}. By Remark \ref{remarkPZ},
\be\label{eq:tikka}
T^{-\nicefrac{1}{2}}\f_1'\sthi\ui \approx  
T^{-\nicefrac{1}{2}}\f_1'\left[\I_T-\Hninv\f_2\left(\f_2'\Hninv\f_2\right)^{-1}\f_2'\right]\Hninv\ui .
\ee
Let $\bm{W}:=\Hninv\f$, then 
\berr
&&\f_1'\left[\I_T-\Hninv\f_2\left(\f_2'\Hninv\f_2\right)^{-1}\f_2'\right]\Hninv\ui=
\f_1'\left[\I_T-\bm{W}\left(\f_2'\bm{W}\right)^{-1}\f_2'\right]\Hninv\ui\\
&=& \f_1'\f_{2\bot}\left(\bm{W}'_{\bot}\f_{2\bot}\right)^{-1}\bm{W}'_{\bot}\Hninv\ui
=\f_1'\f_{2\bot}\left(\bm{W}'_{\bot}\f_{2\bot}\right)^{-1}\bm{W}'_{\bot}\Hninv\epi,
\eerr
where $\f'_{2\bot}\f_2=\bm{0}$.

It follows that even if $\f_1\notin\textrm{sp}(\f_2)$, where $\textrm{sp}(\f_2)$ denote the space spanned by $\f_2$, the UGLS estimator is still consistent but the term $\breve{\f}_1'\bm{\Xi_i}^{-1}\epi$, with $\breve{\f}_1:=\bm{W}_{\bot}\left(\f_{2\bot}'\bm{W}_{\bot}\right)^{-1}\f_2'\f_1\in\textrm{sp}(\bm{W}_{\bot})$, will contribute to the asymptotic distribution of the estimator.\footnote{We conjecture that one needs to assume on the $\breve{\mathrm{g}}_t$ the same conditions assumed by Robinson and Hidalgo (1997) about their $ x_t$ (see their Condition 7) .}
  In Appendix~\ref{app3} (Supp. Material) we show, heuristically, that the FGLS estimator of (\ref{eq:diffact}) enjoys the same asymptotic properties stated in Theorem \ref{Theorem_GLS}.

 

\subsection{Dynamic models} \label{dynamic}

Although our set-up allows for dynamics,  through the dynamic autocorrelation of either the factors and the idiosyncratic error, 
   our results extend to the case of dynamic panel with factor structure such as
%
%
\be\label{eq:dynamic1}
\y=    \x\bm{\beta}_{i}+ \bm{y}_{-1,i}  {\rho}_{i}+  \ui, 
\ee
  where $ \x $ and $ \ui $ satisfy  (\ref{eq:ref_def}) and  (\ref{eq:uu}), respectively, and we set $ \bm{y}_{-1,i} =  ( \mathrm{y}_{i0}, \dots, \mathrm{y}_{iT-1})' $, with  first-order autoregressive coefficients  satisfying  $ -1 < \rho_i < 1 $ for every $i$. We set $ \bm{D} = \bm{0} $  to simplify the exposition.
    Obviously one can re-write (\ref{eq:dynamic1}) as 
  \be  \label{wrongdyn}
\y=\x^*  \bm{\beta}_{i}^* +  \ui \mbox{ setting }  \x^* := ( \x , \bm{y}_{-1,i} )  \mbox{ and } \bm{\beta}_{i}^* := ( \bm{\beta}_{i} ' ,  {\rho}_{i} )'.
  \ee
  It turns out that applying our GLS estimator (\ref{eq:bgls})  to specification (\ref{wrongdyn})  will still work when 
   further conditions are assumed on the idiosyncratic part of the residuals $ \ui = {\bf F } {\bf b }_i + \bm{\varepsilon }_i  $, namely  that 
   the $ \varepsilon_{it} $ are i.i.d. across time but have some degree of cross-correlation across $i$. A similar assumption is made by \citet{CP13}, Assumption 1,  also in the context of dynamic panel data models.  Notice that our result is rather strong because we are {\em not } ruling out that the model residuals $u_{it}$ are dependent across time (and across $i$), through the factors ${\bf f }_t $. Moreover, again thanks to the factors ${\bf f }_t $, regressors and residuals are correlated and thus we are violating the classical strong-exogeneity assumption typically advocated in a GLS framework.

    A technical proof goes beyond the scope, and the page limit, of the present paper but details, 
  corroborated by Monte Carlo simulations, are  available upon request.

\subsection{Efficiency improvements}


  The form of the asymptotic covariance matrix of the GLS, indicated in Theorem~\ref{Theorem_GLS}, denotes lack of efficiency, unlike for the UGLS estimator case (in the special sense discussed). This arises because although $  \sh$  is {\em approximated } by  the  matrix $ \MA$ defined in Appendix~\ref{proofT2}, (in the sense that $\xx'\sh\uu = \xx'\MAI\uu + o_p(\sqrt{T})$),  the latter does not coincide with the true covariance matrix $\ssth $. This is to be expected  since $\MA$ is constructed based on the OLS residuals $ \huu = \yy - \xx \bm{\hat{ \beta }}_i^{OLS} $  where $\bm{\hat{ \beta }}_i^{OLS} $ is non consistent for $ \bm{\beta }_{i} $. However, a multi-step procedure can be envisaged that could achieve (near) asymptotic efficiency, or more precisely an estimator with an asymptotic distribution arbitrarily close to the UGLS estimator. We shall call the outcome of this procedure the iterated-GLS estimator. The first step would be to construct the GLS estimator as explained in the previous sections, which we now denominate as $\hat{\bm \beta }_{i}^{(1)} $. We then construct the associated residuals $ \huu^{(1)} = \yy - \xx \bm{\hat{ \beta }}_i^{(1)} $. Notice that now  $\bm{\hat{ \beta }}_i^{(1)}  $ is a consistent estimator for $\bm{\beta }_{i} $. The second step entails constructing 
$\sh^{(1)} = N^{-1} \sum_{i=1}^N    \db'   \huu^{(1)} \huu^{(1)'}  \db  $ and using it to obtain 
$
\bm{\hat{ \beta }}_i^{(2)} = 
 \left(
 \xx'  \left(\sh^{(1)}\right)^{-1} \xx \right)^{-1} \xx'
 \left(\sh^{(1)} \right)^{-1}  \yy.
$ 
In general the $h$th step entails constructing  $  \bm{\hat{ \beta }}_i^{(h)} =   \left(\xx'  \left( \sh^{(h-1)} \right)^{-1} \xx \right)^{-1} \xx' \left(\sh^{(h-1)} \right)^{-1}  \yy$,  
where $ 
\sh^{(h-1)} $ is obtained based on $ \bm{\hat{ \beta }}_i^{(h-1)} $. We conjecture that as $h$ increases, the asymptotic distribution of  $\bm{\hat{ \beta }}_i^{(h)} $ 
is getting arbitrarily close to the one of the UGLS.  This is confirmed by the Monte Carlo experiments presented in Section~\ref{MC}.  Although the theoretical  analysis of this iterated-GLS is not developed here, techniques along the lines of the ones developed in the current paper would allow to establish the asymptotics.  Indeed, since $\bm{\hat{ \beta }}_i^{(h)} $  is consistent for $ \bm{ \beta }_{i} $  for any $h \ge 2 $, the asymptotics should follow more easily than for the GLS estimator.

%
 
\subsection{  Estimation of asymptotic covariance matrix}   \label{SE}
  Consistent estimation of the GLS asymptotic covariance matrix    can be obtained in different ways, depending on the type of heteroskedasticity and correlation assumed for the $ \epi $. For instance, using the  results of 
  \cite{NW87}, Theorem 2,  one obtains the covariance matrix estimator for $\hat{\bm \beta }_i^{GLS} $
 \be
 \left(   { \xx' \shi  \xx \over T }  \right)^{-1}
 \left( \hat{\bm A }_{i0}  + \sum_{h=1}^n {\scriptstyle (1 - {h \over (n+1) })} \left(  \hat{\bm A }_{ih} + \hat{\bm  A }_{ih}'               \right) \right) 
  \left(   { \xx' \shi  \xx \over T }  \right)^{-1},  \label{con_est}
\ee
setting $
 \hat{\bm A }_{ih}:=  T^{-1} \sum_{t=h+1}^T  \hat{u}_{it}^{GLS}  \hat{u}_{it-h}^{GLS}  \hat{\mathcal{x} }_{it} \hat{\bf \mathcal{x} }_{it-h}', \,\,\, h=0,1,...T-1$,
where $ \hat{ \xx }:= ( \hat{\mathcal{x} }_{i1} ... \hat{\mathcal{x} }_{iT} )' =    \shi   \xx $ and $   \hat{ \uu }^{GLS}:= ( \hat{u}_{1i}^{GLS} ... \hat{u }_{iT}^{GLS} )' = \yy  - \xx  \hat{\bm \beta }_i^{GLS} $    and the bandwidth $n=n(T,N)$  grows slowly with $N$ and $T$.
The same approach has been used in \cite{P06}, eq. (51) and (52).   \cite{B09}, Section 7, provides estimators of the asymptotic covariance matrix when correlation and heteroskedasticity of either series or cross-section form is allowed for, using  \cite{NW87} and a partial-samplig approach, respectively. Note that these approaches cannot be applied to our case since require constant regression coefficients, involving averaging across both $N$ and $T$. Similar approaches have been used by   \cite{MW09} and  \cite{MW13} under  more restrictive  dependence assumptions.
%
%
%
 
  Notice that, although  $\hat{\uu}^{GLS} $ contains a factor structure,  as it is evident by from  its population counterpart  $ \uu  = \ff  \b  + \ei $, the contribution of 
  $  \shi \ff \b   $ to $ \shi  \hat{\uu }^{GLS} $ is  (asymptotically) negligible with respect to the contribution of $\shi \ei $, whose asymptotic variance is require for consistent 
  estimation of the asymptotic covariance matrix of $\bm{\hat{ \beta }}_i^{GLS}  $.
    An alternative approach consists of   estimating the idiosyncratic component of $\hat{\uu}^{GLS} $ directly, for example by  principal components, yielding $ \hat{\ei}^{GLS}:=(  \hat{\epsilon }_{i1}^{GLS}  \cdots   \hat{\epsilon }_{iT}^{GLS}  )' $ and then replacing $ \hat{\bm A }_{ih} $ by $
   T^{-1} \sum_{t=h+1}^T  \hat{\epsilon }_{it}^{GLS}  \hat{\epsilon }_{it-h}^{GLS}  \hat{\bf x }_{it} \hat{\bf x }_{it-h}' $ into (\ref{con_est}). Preliminary testing for the number   of factors $M$  is required in this case, making it less appealing.

   Consistent estimation of the  asymptotic covariance matrix  for 
$   (  \breve{\bm \alpha }_i^{GLS  \prime} , \breve{\bm \beta }_i^{GLS \prime } )' $
    of (\ref{dagger}) follows along the same lines,  leading to:
 \[
 \left({ \zi'  \breve{\bm{S}}_N^{-1} \zi \over T } \right)^{-1}
  \left( \breve{\bm A }_{i0}  + \sum_{h=1}^n {\scriptstyle (1 - {h \over (n+1) })} \left(  \breve{\bm A }_{ih} + \breve{\bm  A }_{ih}' \right) \right)
   \left({ \zi'  \breve{\bm{S}}_N^{-1} \zi  \over T } \right)^{-1},
\]
 where $
 \breve{\bm A }_{ih} =  {1 \over T } \sum_{t=h+1}^T  \breve{u}_{it}^{GLS}  \breve{u}_{it-h}^{GLS}  \breve{\bf z }_{it} \breve{\bf z }_{it-h}', \,\,\, h=0,1,...T-1$,
where $ \breve{ \z }_i = ( \breve{\bf z }_{i1} ... \breve{\bf z }_{iT} )' =   \breve{\bm{S}}_N^{-1}   \z_i $ and $   \breve{ \uu }^{GLS} = ( \breve{u}_{1i}^{GLS} ... \breve{u }_{iT}^{GLS} )' = 
\y  - \d  \breve{\bm \alpha }_i^{GLS  }  - \x \breve{\bm \beta }_i^{GLS } $.


 \subsection{  Cross-sectional regressions} 

As an example of a cross-sectional regression with factor structure,  consider \cite{A05}  model:
 \be y_{it} = {\bm \vartheta }_{t} ' (1 \, {\bf x }_{it}'
)' + u_{it}, \label{fac0} \ee where $(y_{it}, {\bf x }_{it}  )$  are
assumed $i.i.d.$ across units conditional on $ {\bf c }_{1t}, \,
{\bf C }_{2t}$ by \cite{A05}, Assumption 1, with \ber u_{it}
&=& {\bf c }_{1t}' {\bf u }_i^* + \varepsilon_{it}, \label{fac1} \\
{\bf x }_{it} &=& {\bf C }_{2t} {\bf x }_i^* + {\bf v }_{it} ,
\label{fac2}
 \eer with ${\bf c
}_{1t}, \, {\bf u }_i^* $ are $ d_1 \times 1 $ random vectors and $
{\bf C }_{2t}, \, {\bf x }_i^* $  respectively a random matrix of
dimension $k \times d_2 $, with $ d_2 \ge k $, and a random vector
of dimension $ d_2 \times 1 $ and $ \varepsilon_{it} $ and $ {\bf v
}_{it} $ are $i.i.d.$  innovations across $i$ and $t$, respectively
scalar and $k \times 1 $, with zero mean and variances $ \xi_{i,t} $
and $ {\bf \Sigma }_{V_t'V_t}  $, respectively. We focus here on \cite{A05}'s {\em
standard factor} structure, spelled out in his Assumption SF1, here
 slightly extended to allow for an
idiosyncratic component in both the regression error $u_{it}$ and
the regressors ${\bf x }_{it}$ as well as time-variation in
parameters, common factors and covariance matrices.   The first
extension is unavoidable for us since when $ \varepsilon_{it}  = 0 \,\,\,
a.s.$ our theory does not apply.    Model  (\ref{fac0})-(\ref{fac1})-(\ref{fac2}) can be rewritten as   
 \be \label{time}
{\bf y}_t = {\bf D } {\bf \alpha }_{t} +  {\bf X }_t {\bf \beta }_{t} + \ut,
 \ee
where we set  $ {\bf y }_t = ( y_{1t}  \hdots y_{Nt})' $ and $ {\bf D } = {\bm \iota }_N ,  {\bf X }_t =    {\bf X }^*   {\bf C }_{2t}'   $,
with $  {\bf X }^* = ( {\bf x }_1^*  \hdots {\bf x }_N^* )' $ and parameters $ {\bm \vartheta }_{t} = ( {\bm \alpha }_{t}' , {\bm \beta }_{t}'  )' $, and  the $ {\bm u}_t = ( u_{1t} \hdots u_{Nt})' $  satisfy the factor structure  (\ref{eq:uu})
 \[ 
 {\bm u}_t  =
{\bf B } {\bf f }_t + {\bm \varepsilon }_t,
\]
with ${\bf f}_t = {\bf c }_{1t} $, $ {\bf B } = ( {\bf u }_1^*  \hdots {\bf u }_N^* )' $ and 
  $ \bm{\varepsilon }_t=( \varepsilon_{1t}  \hdots \varepsilon_{Nt} )' $.

  In analogy with Section~\ref{def}, 
%
%
  our proposed   {\em feasible } GLS estimator is 
\[
\hat{\bm \beta }_t^{GLS} = \left(\xxt'  \hat{\bm{\mathcal{S}}}_T^{-1} \xxt \right)^{-1}
\xxt' \hat{\bm{\mathcal{S}}}_T^{-1}\yyt,
\] 
setting $ \yyt =  \bm{D}_{\bot}' \bm{Y}_t, \xxt = \bm{D}_{\bot}' \bm{X}_t $,  
 assuming large enough $N$ and $T$ to ensure invertibility of $ \hat{\bm{\mathcal{S}_T}} $, given by
\begin{equation}\label{eq:charlie2T}
 \bm{\hat{\mathcal{S}}_T} =
T^{-1} \sum_{t=1}^T  \huut \huut'   , \,\,\, \mbox{ with } \huut=
 \yyt-\xxt \bm{\hat{\beta}}_t^{OLS} = \mxxt \hut ,
\end{equation}
where $ \bm{\hat{\beta}}_t^{OLS} =  \left(\xxt'   \xxt \right)^{-1}
\xxt' \yyt $ and  ${\hut}$ are the OLS estimator and the OLS regression residuals, respectively,  of regression (\ref{time}).
%
%
%
%
%
Notice that now  $\md = \bm{I}_N-\d(\d'\d)^{-1}\d' = \bm{D}_{\bot} \bm{D}_{\bot}' $ is a $N \times N $ matrix and $ \d_{\bot} $ is a $N-1 \times N $ matrix.
Given the duality between $\hat{\bm \beta }_t^{GLS}  $ and $ \hat{\bm \beta }_i^{GLS}  $,  we conjecture that under a set of regularity conditions analogous to  Assumptions 2.1-2.6 one obtains consistency of $\hat{\bm \beta }_t^{GLS}  $ for  $ 1/N + N/T \rightarrow 0 $ and asymptotic normality of $ \sqrt{N}( \hat{\bm \beta }_t^{GLS}  -  {\bm \beta }_{t}  ) $ for $  1/N + N^3/T^2  \rightarrow 0 $. Extension to a more general form of common observed regressors, other than $ {\bf D } = {\bm \iota }_N $,  can be obtained along the lines of  Section~\ref{common}.

\section{Monte Carlo analysis}  \label{MC}

We conduct a set of  Monte Carlo  experiments to appreciate the relevance of our asymptotic results for the GLS estimator in finite samples. 

\subsection{Design}

The data generating process is
\ber \label{unit}
y_{it} && = \alpha_{i0} + \beta_{i0} x_{it} + b_{i10} f_{1t} + b_{i20} f_{2t}  +\varepsilon_{it},
\eer
where the single regressor satisfies
\ber \label{reg}
x_{it}  = 0.5 + \delta_{i10} f_{1t} + \delta_{i30} f_{3t}  + v_{it}.
\eer
 Note that the model implies an observed common factor equal to $1$ for all observations. The single regressor is allowed to be contemporaneously correlated with the innovation through one of the latent common factors (whenever $ b_{i10} \delta_{i10} \neq 0  $). The factor loadings are normally distributed random variables, $i.i.d.$ across unit:
  \ber
&& \left( \begin{array}{c} b_{i10} \\ b_{i20} \end{array} \right) \sim NID \left( \left( \begin{array}{c} 1 \\ 0 \end{array} \right), \left( \begin{array}{cc} 0.2 & 0 \\ 0 & 0.2 \end{array} \right) \right), \\
&& \left( \begin{array}{c} \delta_{i10} \\ \delta_{i30} \end{array} \right)  \sim NID \left( \left( \begin{array}{c} 0.5 \\ 0 \end{array} \right), \left( \begin{array}{cc} 0.5 & 0 \\ 0 & 0.5 \end{array} \right) \right),
\label{rank}
  \eer
and the latent common factors and the idiosyncratic components are stationary stochastic processes, mutually independent to each other, satisfying
\[
  { f }_{j,t}  = 0.5  { f }_{j,t-1} + \sqrt{0.5} { \eta }_{jf,t}, \,  j=1,2,3,
 \]
 where each $  { \eta }_{jf,t}  \sim NID(0,1) $, mutually independent for $j=1,2,3$, and
\berr
&&    \varepsilon_{it}  =  \rho_{i\varepsilon }   { \varepsilon }_{it-1} + { \eta }_{i\varepsilon,t}, \,\,\,
 { \eta }_{i\varepsilon,t}  \sim NID( 0,   \sigma_i^2 ( 1 -  \rho_{i\varepsilon }^2) ),   i=1,...,N, \\
&&    v_{it}  =  \rho_{iv }
  { v }_{it-1} + { \eta }_{iv,t}, \,\,\,
 { \eta }_{iv ,t}  \sim NID( 0,   ( 1 -  \rho_{i v }^2) ), i=1,...,N,
 \eerr
with $ \rho_{i \varepsilon } \sim UID(0.05,0.95), \, \rho_{i v } \sim UID(0.05,0.95), \sigma_{i \varepsilon }^2 \sim UID(0.5,1.5) $ where $NID, UID$ means $iid$ normally and uniformly distributed respectively.
Finally, the  parameters of interest are  constant across replications and equal to $ \alpha_{i0} = 1,  \gamma_{i0} = 0.5 $ and, assuming $N$ even,
\[
\beta_{i0} = \left\{  \begin{array}{ll} 1 & \mbox{ for } i=1,...,{N \over 2} , \\ 3  & \mbox{ for } i={N \over 2}+1,...,N. \end{array}    \right.
\]
This Monte Carlo design is a  simplified version of \cite{P06}, designed in such a way that (through (\ref{rank})) the rank condition in \cite{P06},  eq. (21), is not satisfied. \cite{P06} shows that
under this circumstance his individual specific estimator for $ \beta_{i0} $ is invalid  whereas his pooled estimators for $\beta_{0} = \E \beta_{i0} $ remains consistent.

%

We consider $2000$ Monte Carlo replications with sample sizes $(N,T) \in \{ 60,200,600 \} \times  \{ 30,100,300 \}$, where  $N>T$.

The results  are summarized in Tables 4,5, and 6, where we report the sample mean  and the root mean square error
 for the estimates of the parameter ${\bf \alpha }_{i0}, {\bf \beta }_{i0} $,
 averaged across the Monte Carlo iterations. We consider four estimators which corresponds to four panels of each table: the GLS, 
 the multi-step GLS (described in Section~\ref{various}) where the iteration is carried out $J=4$ times, the OLS 
 and the UGLS 
  estimators.
 In particular,    for each of these four estimators,  we report the average  across all $N$ units of the sample mean (denoted by \texttt{mean})
 $ MM^{-1} \sum_{m=1}^{MM} \hat{\alpha}_i^m $ and of  the root mean square error (denoted by (denoted by \texttt{rmse}) $  \left( MM^{-1} \sum_{m=1}^{MM} (   \hat{\alpha}_i^m -   1 )^2 \right)^{1 \over 2} $
and the average across the units $i=N/2+1,...,N$ of
 $ MM^{-1} \sum_{m=1}^{MM} \hat{\beta}_i^m $ and $  \left(MM^{-1} \sum_{m=1}^{MM} (   \hat{\beta}_i^m -   3 )^2 \right)^{1 \over 2} $ with $MM=2,000$. Recall that we assumed that the true intercept coefficients are constant across units whereas the regression coefficients take two different values for the first half and second half of the $N$ units.
  Here
$\hat{\alpha}_i^m $ and $ \hat{\beta}_i^m $ denote, respectively, the estimates of the intercept and regression coefficients corresponding to the $m$th Monte Carlo iteration for a generic estimator.

\subsection{Results}


Our comments below apply to each table, with minor differences. Since the GLS and multi-step GLS estimators requires $N \ge T $, each panel is made by a lower triangular matrix. Obviously, the OLS and the UGLS estimator do not require this constraint since they can be also
evaluated when $N<T$ but we did not report the results for this case.
The upper left panel describes the GLS results. One can see how the bias diminishes as both $N,T$ grow or when $N$ increases for a given $T$. This is because the
inverse of the pseudo-covariance matrix  is better estimated in these circumstances. In contrast, although still negligible in absolute terms, the bias, if any,  tends to increase when $T$ grows for a given $N$.
Instead, as expected, the \texttt{rmse} always diminishes when  $T$ increases for a given $N$ or when they both increase. 
 In general these  results suggest that the bias of the estimates varies mainly with $N$ and their variance varies with $T$.
 The same pattern is observed with respect to the multi-step GLS results, reported in the upper right panel.
The only difference is that now the bias and the \texttt{rmse} are always much smaller than the GLS case. The lower right panel reports the results for the UGLS which is unfeasible in practice since it involves the
true covariance matrix ${\bf S }_N $. As a consequence, the results do not depend on $N$ but only on $T$.  The bias is negligible even for small samples and, for larger sample sizes,
it is remarkably comparable to the iterated GLS although the latter exhibit a slightly larger \texttt{rmse}. Finally, the lower left panel reports the OLS results which also do not depend on $N$, as expected.
 Under our design, the OLS estimator is non-consistent obtaining  a bias which is much larger than for any other estimators and, more importantly, only marginally varying as $N$ or $T$ increases.
The \texttt{rmse} diminishes suggesting that the variance of the OLS estimator is converging to zero with the squared bias converging to the squared of $  \tau_i^{OLS} $.

\section{Empirical Application: Firms' Characteristics and Expected Returns}  \label{EMP}

We present an empirical application of our methodology, inspired by  asset-pricing theory. According to  so-called beta-pricing models, asset returns follow a factor model:
\be
R_{i,t} = \tilde{\alpha}_{i} + \tilde{\boldsymbol {\gamma }}_{i}' \d_t +  u_{i,t} ,  \label{eq:facmod}
\ee
where $ R_{i,t}$ defines the  rate of return for asset $i$, in excess of the risk-free rate, and $ \d_t $ is a vector of observed factors, with coefficients $ \tilde{\boldsymbol \gamma }_{i} $.
Important, special, cases of model (\ref{eq:facmod}) are the Capital Asset Pricing Model (CAPM) of \cite{S64}  and \cite{L65}, when $\d_t$ is the (scalar) excess  market return with $ \tilde{\alpha}_{i} =0 $ for every $i$, and the Arbitrage Pricing Theory (APT)  of \cite{R76}, when $\d_t$ is a vector of possibly non-traded factors.\footnote{Focusing on the special case when $\d_t$ are the excess returns of traded assets, the APT  holds when  the  $ \tilde{ \alpha}_{i} $, although not zero, satisfy the condition $ \tilde{\boldsymbol \alpha }'  ( var( \ut ) )^{-1}  \tilde{\boldsymbol \alpha}  < \infty $, setting $  \tilde{\boldsymbol \alpha }  = ( \tilde{\alpha}_{1} , \cdots , \tilde{ \alpha}_{N} )' $.}

 Model (\ref{eq:facmod}), together with some form of no-arbitrage and some constraints of the covariance matrix of the $u_{i,t}$, implies that expected  excess returns $\E(R_{it}) $ are linear in the coefficients $  \boldsymbol{ \alpha }_i  $ only, namely that the $\d_t$ are the only source of risk (see Corollary 1, \cite{C83}). However, this fundamental paradigm has been challenged  empirically. For instance, \cite{DT97} 
 and    \cite{DFF}  provide strong evidence according to which  stocks characteristics, such as market capitalization (size), valuation (book-to-market) and other characteristics do influence expected returns well beyond the betas.   One can  extend model (\ref{eq:facmod}) to allow for   characteristics by specifying:
 \be
 R_{i,t} = \tilde{\alpha}_{i} + \tilde{\boldsymbol {\gamma}}_{i}' {\bf  D}_t +  \boldsymbol{\beta}_i' {\bf X}_{i,t} +  u_{i,t} =   \boldsymbol{ \alpha}_i'  (1, {\bf D}_t')' +   \boldsymbol{  \beta}_i' {\bf  X}_{i,t} +  u_{i,t},  \label{eq:facmod2}
  \ee
 where now $X_{i,t}$ defines a vector of characteristics associated with the $i$th stock, setting $ \boldsymbol{\alpha}_i = (   \tilde{{\alpha}}_{i}, \tilde{\boldsymbol{ \gamma}}_{i}'  )' $.  Model (\ref{eq:facmod2}) can be interpreted
  as, and in fact is equivalent to, our basic model (\ref{eq:hetero}). Moreover, it is conceivable that the error term has a factor structure, such as (\ref{eq:uu}), possibly correlated with both the $\d_t$ and the $ \xt$. For instance, this is arises whenever one suspects the possibility of missing, pervasive, factors.   
  Our asymptotic distribution theory can be used to assess whether the   $\boldsymbol{ \alpha }_i $ or the $ \boldsymbol{ \beta }_i $ or both are significant or not.

 We use a data set of monthly observations,  from January 1966   to December 1994, of individual asset returns extracted from CRSP and of firms' characteristics extracted from COMPUSTAT.\footnote{See \cite{BCS98}   for details}. In particular,  the eight characteristics that we consider are SIZE (the natural logarithm of the market value of the equity of the firm as of
the end of the second to last month), BM  (the natural logarithm of the ratio of the book value of equity plus deferred taxes to the market value of equity, using the end of the previous year market and book values)\footnote{As in Fama and French (1993), the value of BM for July of year t to June of year t+1 was computed using accounting data at the end of year t-1.},  DVOL  (the natural logarithm of the dollar volume of trading in the security in the second to last month),   PRICE (the natural logarithm of the reciprocal of the share price as reported at the end of the second to last month),  YLD  (the dividend yield as measured by the sum of all dividends paid over the previous 12 months, divided by the share price at the end of the second to last month),    RET2-3  (the natural logarithm of the cumulative return over the two months ending at the beginning of the previous month), RET4-6 (the natural logarithm of the cumulative return over the three months ending three months previously), 
 RET7-12  (the natural logarithm of the cumulative return over the 6 months ending 6 months previously).\footnote{Lagged return variables were constructed to exclude the return during the immediate prior month in order to avoid any spurious association between the prior month return and the current month return caused by thin trading or bid-ask spread effects.}

 We report the results in Table 1,2 and 3. In particular,  we consider three different  factor models, depending on the set of common factors. Table 1 refers to the CAPM model augmented with the eight characteristics. 
 We report the average, across the $N=356 $ assets, of the  GLS estimates  $ (\breve{\boldsymbol{ \alpha }}_i^{GLS \prime } , \breve{\boldsymbol \beta }_i^{GLS \prime} )' $ in (\ref{dagger})  for each regression parameter, together with their $10$th and $90$th percentiles, out of the $N$ assets. 
 Similarly, we report the average, across the $N=356 $ assets, of the t-$ratio$s  for each regression parameter, together with their $10$th and $90$th percentiles, out of the $N$ assets.
 Finally, we report the  F test statistics corresponding to three different joint hypotheses, namely for  all $ \boldsymbol{ \alpha}_i = 0 $, or  all $ \boldsymbol{\beta}_i=0 $ or both.  Again, we report
  the average across the $N$ assets of the F test statistics, and  their $10$th and $90$th percentile. Table 2 refers to the 
  three-factor model of \cite{FF3}, augmented with the eight characteristics, whereby the elements of $\d_t$ are the market, the small-minus-large (SML) and the high-minus-low (HML) portfolio returns, respectively. Finally, Table 3 refers to the five-factor model of \cite{FF5}, augmented with the eight characteristics, whereby the elements of $\d_t$, with respect to the three-factor model, are augmented by  the profitability (RMW)  and investment (CMA)  portfolio returns.

 Across all the three asset-pricing models, the results strongly indicate that characteristics  influences excess returns, and highly significantly so. This emerges both by considering individual t-$ratios$ as well   as the F test for the joint hypothesis that the coefficients to the characteristics (i.e. the $ \boldsymbol{ \beta}_i $)   are all zero.  Noticeably,  the effects of the common factors, for example the market return 
 for the CAPM, are also strongly significant,  across the three asset-pricing models. Indeed, their effects appear unambiguously stronger than for the characteristics, although 
 they are both highly significant.

\section{Concluding remarks} \label{CONCL}

This paper proposes a feasible GLS estimator for linear panel with common factor structure in both the regressors and the innovation.
We establish our results   for time regressions with unit-specific coefficients, and present  several  generalisations such as dynamic panels,  cross-section regressions with time varying coefficients and different factor structures for regressors and residuals.
The GLS estimator is consistent and asymptotically normal, when both the cross-section $N$ and time series  $T$ dimensions diverge to infinity  where, under the same 
circumstances, the OLS is first-order biased. 
In summary, the GLS estimator exhibits four main   properties: 
first,   it permits to carry out inference on the regression coefficients  based on conventional distributions; 
    second, as in classical estimation theory, it delivers (almost) efficient estimation; 
 third,   it does not require any knowledge of the exact number  of latent
  factors, or even an upper bound of such number; and 
 fourth, the GLS is computationally easy to handle without invoking any nonlinear numerical optimizations. 
 Our results are corroborated by a set of Monte Carlo experiments and illustrated by an asset-pricing empirical application.
 
  



%
%
%
%
%
%
%
%
%
%
%

\vspace{-2in}

\begin{table}  \label{capmtable}
\centering
 {\footnotesize\bf $\,\,\,\,\,\,\,\,\,\,\,\,\,\,\,\,\,\,\,\,\,\,\,$  Table 1: \\    $\,\,\,\,\,\,\,\,\,\,\,\,\,\,\,\,\,\,\,\,\,\,\,$ CAPM: \\ Testing  the effect of characteristics  }
  \footnotesize
\begin{center}
\begin{tabular}{ |c|ccc|ccc| } 
  \hline
  \multicolumn{7}{|c|}{Panel A} \\ \hline \\
  \multicolumn{1}{|c|}{}  & \multicolumn{3}{|c|}{t-ratios} & \multicolumn{3}{|c|}{GLS estimates}  \\
     parameter:  & 10-th quantile  &  average   &  90-th quantile  	 & 10-th quantile  &  average   &  90-th quantile 	\\		
    intercept  &-182.9  &    -32.6 &     125.2			&		  -0.123  & -0.024   & 0.063	\\
  Mkt &        1507.2  &       2488.4   &    3568.2			&		    0.061   & 0.010   & 0.014		\\
   SIZE &      -239.5    &    -83.6   &      77.1				&	   -0.217   & -0.069  &  0.073		\\
  BM &       -118.4 &           14.7 &           153.5 &						   -0.063 &    0.006 &   0.076 \\		
   DVOL &      -12.1 &        118.9 &        272.4 &						   -0.01   &  0.038 &   0.085 \\		
   PRICE &      -266.2 &         -126.7 &          4.6 &						   -0.088 &   -0.044 &    0.001 \\		
   YLD &      -301.9 &           -103.6 &          59.9 &						   -0.150 &   -0.051 &    0.010 \\		
  RET23 &       -289.1 &         -143.4 &      -26.4 &						   -0.021 &   -0.012 &  -0.002 \\		
 RET46 &        -264.9 &         -139.9 &       -14.2 &						   -0.022 &   -0.012 &  -0.001 \\		
   RET712 &      -234.4 &         -121.4 &         -18.7 &						   -0.019 &  -0.011 &   -0.001 \\		
  \hline 
  \multicolumn{7}{|c|}{Panel B} \\ \hline 
  \multicolumn{1}{|c|}{test statistic:}  & \multicolumn{2}{|c}{10-th quantile} & \multicolumn{2}{c}{average}  & \multicolumn{2}{c|}{90-th quantile} \\
   \multicolumn{1}{|c|}{ $F_\gamma$  } &  \multicolumn{2}{|c}{1444899.20}    &   \multicolumn{2}{c}{ 5135622.34}      &   \multicolumn{2}{c|}{9905946.98	}	\\	
      \multicolumn{1}{|c|}{$ F_\beta  $} &  \multicolumn{2}{|c}{82958.69}    &   \multicolumn{2}{c}{ 296084.32}      &   \multicolumn{2}{c|}{ 612749.73	}	\\	
      \multicolumn{1}{|c|}{$ F_{ \beta , \gamma }$  } &  \multicolumn{2}{|c}{7690568.38}    &   \multicolumn{2}{c}{20110261.61}      &   \multicolumn{2}{c|}{ 37936804.13	}	\\
      	\hline 
 \end{tabular}
\end{center}
\end{table}
\noindent \footnotesize{{\bf Note to Table~1}:  Panel A reports t-$ratios$ and parameter estimates corresponding to the CAPM model, augmented with characteristics SIZE,
BM,  DVOL,  PRICE,  YLD,  RET23,  RET46 and   RET712:
\[
 R_{it}  =  \tilde{\alpha }_{i} +  \tilde{\gamma }_{i} R_{Mkt,t} +    \boldsymbol{ \beta}_{i}'  {\bf X}_{it}   + u_{it}, \,\,\, t=1, \cdots , T, \,\,\, i=1, \cdots , N ,
 \]
where $R_{it}$ defines the excess return on asset $i$,  $R_{Mkt,t}  $ is the $S\& P 500 $ excess return and $X_{it}$  the $8 \times 1 $ vector of characteristics.
Panel B reports the $F$ test statistics corresponding to the null hypotheses $H_0:  \tilde{ \gamma }_{i} = 0 $,   
 $H_0:   \boldsymbol{ \beta}_{i} = 0 $ and   $H_0:  \boldsymbol{ \beta}_i = 0 ,  \tilde{ \gamma }=  0 $,  given by $F_\gamma , F_\beta $ and $ F_{\beta , \gamma } $ respectively.

The data are monthly and makes a panel of monthly observations with $T=348, N=356 $. The characteristics have been cross-sectionally standartized.
Column 2 to 4 of Panel A report the $10$-th decile, the average and the $90$th decile of the t-$ratios$ across the $ N$ assets. Columns 5 to 7  of Panel A report the same quantities with respect to the 
parameter estimates, using the GLS estimator $ (\breve{ \boldsymbol \alpha }_i^{GLS \prime} , \breve{ \boldsymbol \beta }_i^{GLS \prime} )' $ in (\ref{dagger}). Their covariance matrix 
 is estimated using the approach described in Section~\ref{SE}. Column 2 to 4 of Panel B report  the $10$-th decile, the average and the $90$th decile of the three F test statistics
 across the $ N$ assets.

\noindent 

}

\newpage

\begin{table}  \label{ff3table}
\centering
 {\footnotesize\bf $\,\,\,\,\,\,\,\,\,\,\,\,\,\,\,\,\,\,\,\,\,\,\,$  Table 2: \\    $\,\,\,\,\,\,\,\,\,\,\,\,\,\,\,\,\,\,\,\,\,\,\,$ Fama French (1993)  3-factor model: \\    $\,\,\,\,\,\,\,\,\,\,\,\,\,\,\,\,\,\,\,\,\,\,\,$ Testing  the effect of characteristics  }
  \footnotesize
\begin{center}
\begin{tabular}{ |c|ccc|ccc| } 
  \hline
  \multicolumn{7}{|c|}{Panel A} \\ \hline \\
  \multicolumn{1}{|c|}{}  & \multicolumn{3}{|c|}{t-ratios} & \multicolumn{3}{|c|}{GLS estimates}  \\
     parameter:  & 10-th quantile  &  average   &  90-th quantile  	 & 10-th quantile  &  average   &  90-th quantile 	\\		
    intercept  &    -210.6   &   -36.6 &        133.5  &    -0.116 &    -0.021 &     0.069  \\					
       Mkt &       1473.1  &        2325   &     3349.8   &  0.006 &    0.009 &     0.0128 \\					
     SMB    &    -541.1   &    256.6   &    1072.5  &  -0.003  &   0.003 &     0.010 \\					
HML &       -340.2 &      348.3 &       1068.7  &  -0.002 &    0.002 &    0.006 \\					
       SIZE    &    -252.3 &    -91.7 &       81.6 &     -0.205 &    -0.068 &     0.065 \\					
     BM &        -138.3 &       14.6 &       169.3 &    -0.06 &    0.004 &     0.076 \\					
     DVOL    &    -12.8 &      125.1 &       278.6 &   -0.003 &     0.036 &     0.081 \\					
    PRICE & -303.3 &   -140.4 &      -2.1 &    -0.089 &    -0.042 &  -0.001 \\
       YLD &         -327.1 &    -113.7  &      72.7  &    -0.148 &    -0.050 &    0.009 \\
  RET23  &    -305.8 &      -155.2 &       -30.8 &    -0.021 &    -0.011 &   -0.002 \\
 RET46 &   -316.6 &     -160.3 &      -25.8 &    -0.021 &   -0.011 &   -0.002 \\
	 RET712 &      -272.5 &     -141.3 &      -28.5 &    -0.018 &    -0.010 &   -0.002 \\
  \hline 
  \multicolumn{7}{|c|}{Panel B} \\ \hline 
  \multicolumn{1}{|c|}{test statistic:}  & \multicolumn{2}{|c}{10-th quantile} & \multicolumn{2}{c}{average}  & \multicolumn{2}{c|}{90-th quantile} \\
    \multicolumn{1}{|c|}{ $F_\gamma$  } &  \multicolumn{2}{|c}{3148154.05}    &   \multicolumn{2}{c}{8766878.36}      &   \multicolumn{2}{c|}{16270455.79}	\\	
      \multicolumn{1}{|c|}{$ F_\beta  $} &  \multicolumn{2}{|c}{87651.83}    &   \multicolumn{2}{c}{ 315163.52}      &   \multicolumn{2}{c|}{ 634901.58}	\\	
      \multicolumn{1}{|c|}{$ F_{ \beta , \gamma }$  } &  \multicolumn{2}{|c}{ 12894492.39}    &   \multicolumn{2}{c}{ 29647189.57}      &   \multicolumn{2}{c|}{   55497464.04}	\\
      	\hline 
 \end{tabular}
\end{center}
\end{table}
\noindent \footnotesize{{\bf Note to Table~2}:  Panel A reports t-$ratios$ and parameter estimates corresponding to the \cite{FF3}  3-factor model, augmented with characteristics SIZE,
BM,  DVOL,  PRICE,  YLD,  RET23,  RET46 and   RET712:
\[
 R_{it}  = \tilde{ \alpha}_{i} + \tilde{\gamma }_{i1} R_{Mkt,t} +   \tilde{ \gamma }_{i2} R_{SMB,t}    +   \tilde{\gamma }_{i3}  R_{HML,t} +  \boldsymbol{ \beta}_{i}'  X_{it}   + u_{it}, \,\,\, t=1, \cdots , T, \,\,\, i=1, \cdots , N ,
 \]
where $R_{it}$ defines the excess return on asset $i$,  $R_{Mkt,t}  $ is the $S\& P 500 $ excess return, $R_{SMB,t}$ is the size factor, $ R_{HML,t} $ is
 the value factor  and ${\bf X}_{it}$ the the $8 \times 1 $ vector of characteristics.
  Panel B reports the $F$ test statistics corresponding to the null hypotheses $H_0:  \tilde{ \boldsymbol \gamma}_{i} = 0 $,    $H_0:  \boldsymbol{ \beta}_{i} = 0 $ and 
  $H_0:   \boldsymbol{ \beta}_i = 0 , \tilde{ \boldsymbol \gamma }_i = 0 $,  given by $F_\gamma , F_\beta $ and $ F_{  \beta  , \gamma } $ respectively, setting $  \tilde{ \boldsymbol  \gamma }_i = (  \tilde{ \gamma }_{i1},  
  \tilde{ \gamma }_{i2},  \tilde{ \gamma }_{i3})' $.  For details refer to the notes to Table~1.


}

\newpage

\begin{table}  \label{ff5table}
\centering
 {\footnotesize\bf $\,\,\,\,\,\,\,\,\,\,\,\,\,\,\,\,\,\,\,\,\,\,\,$  Table 3: \\    $\,\,\,\,\,\,\,\,\,\,\,\,\,\,\,\,\,\,\,\,\,\,\,$ Fama French (2105)  5-factor model:  \\     $\,\,\,\,\,\,\,\,\,\,\,\,\,\,\,\,\,\,\,\,\,\,\,$ Testing  the effect of characteristics  }
  \footnotesize
\begin{center}
\begin{tabular}{ |c|ccc|ccc| } 
  \hline
  \multicolumn{7}{|c|}{Panel A} \\ \hline \\
  \multicolumn{1}{|c|}{}  & \multicolumn{3}{|c|}{t-ratios} & \multicolumn{3}{|c|}{GLS estimates}  \\
     parameter:  & 10-th quantile  &  average   &  90-th quantile  	 & 10-th quantile  &  average   &  90-th quantile 	\\		
    intercept  &   -208.7 &   -36.6 &       129.1 &     -0.120 &    -0.021 &     0.075 \\
   Mkt &         1303    &    2172.2  &      3207.8  &   0.006 &    0.009 &     0.012 \\
SMB &  -605.5 &       259.9 &       1143.2   &  -0.003 &    0.002 &     0.010 \\						
     HML &       -326.7 &    202.5 &        863.8  &  -0.003 &   0.001 &    0.007 \\
     RMW &  -440.7 &     -57.1 &       344.4 &  -0.006 &  -0.001 &    0.006 \\
CMA &   -488.1 &     -0.661 &       484.4 &   -0.007 &    0.0001  &  0.007 \\
   SIZE   &  -255.8 &      -90.3 &       78.5 &     -0.206 &    -0.067 &     0.062 \\
  BM &     -141.7 &      16.4 &       163.8 &    -0.063  &    0.005 &     0.075 \\
   DVOL &   -11.1 &     124.7 &       266.5 &   -0.002 &     0.036 &      0.081 \\
   PRICE &     -300.1 &      -138.9 &      0.813 &     -0.087 &    -0.042 &  -0.001 \\
   YLD &       -319.1 &    -113.1 &       71.9 &     -0.151 &    -0.050 &    0.009 \\
  RET23 &         -304.6 &      -154.5 &      -26.7 &    -0.021 &    -0.011 &   -0.002 \\
 RET46 &     -318.8 &      -158.8 &      -21.7 &    -0.021 &    -0.011 &   -0.001 \\
	   RET712 & -265.6 &     -141.1 &    -31   &  -0.018 &   -0.010 &   -0.002 \\
	    \hline 
    \multicolumn{7}{|c|}{Panel B} \\ \hline 
  \multicolumn{1}{|c|}{test statistic:}  & \multicolumn{2}{|c}{10-th quantile} & \multicolumn{2}{c}{average}  & \multicolumn{2}{c|}{90-th quantile} \\
     \multicolumn{1}{|c|}{ $F_\gamma$  } &  \multicolumn{2}{|c}{ 4930084.6}    &   \multicolumn{2}{c}{11721601.5}      &   \multicolumn{2}{c|}{20418779.1}	\\	
      \multicolumn{1}{|c|}{$ F_\beta  $} &  \multicolumn{2}{|c}{83573.5}    &   \multicolumn{2}{c}{303480.1}      &   \multicolumn{2}{c|}{ 569989.4}	\\	
      \multicolumn{1}{|c|}{$ F_{\beta , \gamma }$  } &  \multicolumn{2}{|c}{14961844.2}    &   \multicolumn{2}{c}{34707178.6}      &   \multicolumn{2}{c|}{60856307.5	}	\\
      	\hline 
 \end{tabular}
\end{center}
\end{table}
\noindent \footnotesize{{\bf Note to Table~3}:  Panel A reports t-$ratios$ and parameter estimates corresponding to the Fama and French \cite{FF5}  5-factor model, augmented with characteristics SIZE,
BM,  DVOL,  PRICE,  YLD,  RET23,  RET46 and   RET712:
\[
 R_{it}  = \tilde{\alpha}_{i} +  \tilde{\gamma}_{i1} R_{Mkt,t} +   \tilde{\gamma}_{i2} R_{SMB,t}    +  \tilde{\gamma}_{i3}  R_{HML,t} +    \tilde{\gamma}_{i4} R_{RMW,t}  
   +  \tilde{\gamma}_{i5}  R_{CMA,t} +  \boldsymbol{ \beta }_{i}'  \boldsymbol{ X}_{it}   + u_{it}, \,\,\, t=1, \cdots , T, \,\,\, i=1, \cdots , N ,
 \]
where $R_{it}$ defines the excess return on asset $i$,  $R_{Mkt,t}  $ is the $S\& P 500 $ excess return, $R_{SMB,t}$ is the size factor, $ R_{HML,t} $ is the value factor, 
$R_{RMW,t}$ is the size factor, $ R_{CMA,t} $ is the value factor   and ${\bf X}_{it}$ the the $8 \times 1 $ vector of characteristics.
  Panel B reports the $F$ test statistics corresponding to the null hypotheses $H_0: \tilde{ \boldsymbol \gamma}_i = 0 $,    $H_0:  \boldsymbol{ \beta}_{i} = 0 $ and   $H_0: \boldsymbol{ \beta}_i = 0 , \tilde{ \boldsymbol \gamma}_i = 0 $, 
given by $F_\gamma, F_\beta $ and $ F_{ \beta  , \gamma } $ respectively, setting $\tilde{ \boldsymbol \gamma}_i = (  \tilde{\gamma}_{i1}, \tilde{\gamma}_{i2},  \tilde{\gamma}_{i3},
  \tilde{\gamma}_{i4}, \tilde{\gamma}_{i5})' $.  For details refer to the notes to Table~1.


\noindent 

}

\newpage

\begin{table}  \label{MC1}
\centering
 {\footnotesize\bf $\,\,\,\,\,\,\,\,\,\,\,\,\,\,\,\,\,\,\,\,\,\,\,$  Table 4: \\ time regression with unit-specific coefficients \\ intercept term $ \alpha_{i0}  = 1, \,\,\, i=1,...,N.$ }
  \footnotesize
\begin{tabular}{|lcccccccccccc|}
\hline
 \multicolumn{1}{|c|}{ }  &
 \multicolumn{6}{|c|}{  } &  \multicolumn{6}{|c|}{  }
 \\
\multicolumn{1}{|c|}{ }  & \multicolumn{6}{|c|}{ \begin{centering}
$GLS$ \end{centering}
 } &
 \multicolumn{6}{|c|}{ \begin{centering}  $GLS$ (multi-step)   \end{centering} }
\\ \hline
\multicolumn{1}{|c|}{ }  & \multicolumn{3}{|c|}{ \begin{centering}
\texttt{mean}
\end{centering}
 } &  \multicolumn{3}{|c|}{ \begin{centering}
\texttt{rmse} \end{centering}
 } & \multicolumn{3}{|c|}{ \begin{centering} \texttt{mean}
\end{centering}
 } &  \multicolumn{3}{|c|}{ \begin{centering}
\texttt{rmse} \end{centering}
 } \\ \hline
\multicolumn{1}{|c|}{  $(N,T)$ }  & \multicolumn{1}{|r}{
\hspace{0.6pt} $30$ } & \multicolumn{1}{r}{\hspace{0.6pt} $100$ } &
\multicolumn{1}{c|}{ $300$ }   & \multicolumn{1}{|r}{ \hspace{0.6pt}
$30$ } & \multicolumn{1}{r}{\hspace{0.6pt} $100$ } &
\multicolumn{1}{c|}{ $300$ }  &  \multicolumn{1}{|r}{ \hspace{0.6pt}
$30$ } & \multicolumn{1}{r}{\hspace{0.6pt} $100$ } &
\multicolumn{1}{c|}{ $300$ }  & \multicolumn{1}{|r}{ \hspace{0.6pt}
$30$ } & \multicolumn{1}{r}{\hspace{0.6pt} $100$ } &
\multicolumn{1}{c|}{ $300$
} \\
\multicolumn{1}{|c|}{  }  & \multicolumn{3}{|c|}{ \begin{centering} \end{centering}  } &  \multicolumn{3}{|c|}{ \begin{centering}  \end{centering}  } & \multicolumn{3}{|c|}{ \begin{centering} \end{centering}  } &  \multicolumn{3}{|c|}{ \begin{centering}  \end{centering}  } \\
\multicolumn{1}{|l|}{  $60$ }  & \multicolumn{3}{|l|}{ \begin{centering} $0.944$ \hspace{0.1in} $-$ \hspace{0.16in} $-$ \end{centering}  } &  \multicolumn{3}{|l|}{ \begin{centering} $0.523$ \hspace{0.1in} $-$ \hspace{0.16in} $-$ \end{centering}  } & \multicolumn{3}{|l|}{ \begin{centering} $0.976$ \hspace{0.1in} $-$ \hspace{0.16in} $-$ \end{centering} } &  \multicolumn{3}{|l|}{ \begin{centering} $0.531$ \hspace{0.1in} $-$ \hspace{0.16in} $-$ \end{centering}  } \\
\multicolumn{1}{|l|}{  $200$ }  & \multicolumn{3}{|l}{ \begin{centering} $0.967$ $ 0.951 $ \hspace{0.1in} $-$ \end{centering} } &  \multicolumn{3}{|l|}{ \begin{centering} $0.518$ $ 0.315 $ \hspace{0.1in} $-$  \end{centering} } & \multicolumn{3}{|l|}{ \begin{centering} $0.986$ $0.987 $ \hspace{0.1in} $-$  \end{centering}  } &  \multicolumn{3}{|l|}{ \begin{centering} $0.527$ $  0.309$ \hspace{0.1in} $-$  \end{centering} } \\
\multicolumn{1}{|l|}{  $600$ }  & \multicolumn{3}{|l|}{
\begin{centering} $0.981$ $ 0.982$ $0.955$ \end{centering} } &  \multicolumn{3}{|l|}{ \begin{centering} $0.524$ $0.308$ $0.200$ \end{centering}  } & \multicolumn{3}{|l|}{ \begin{centering} $0.994$ $0.998$ $0.991$ \end{centering}  } & \multicolumn{3}{|l|}{ \begin{centering} $0.531$ $0.310$ $0.184$
\end{centering} } \\ \hline
 \multicolumn{1}{|c|}{ }  &
 \multicolumn{6}{|c|}{  } &  \multicolumn{6}{|c|}{  }
 \\
\multicolumn{1}{|c|}{ }  & \multicolumn{6}{|c|}{ \begin{centering}
$OLS$ \end{centering}
 } &
 \multicolumn{6}{|c|}{ \begin{centering}  $UGLS$    \end{centering} }
\\ \hline
\multicolumn{1}{|c|}{ }  & \multicolumn{3}{|c|}{ \begin{centering}
\texttt{mean}
\end{centering}
 } &  \multicolumn{3}{|c|}{ \begin{centering}
\texttt{rmse} \end{centering}
 } & \multicolumn{3}{|c|}{ \begin{centering} \texttt{mean}
\end{centering}
 } &  \multicolumn{3}{|c|}{ \begin{centering}
\texttt{rmse} \end{centering}
 } \\ \hline
\multicolumn{1}{|c|}{  $(N,T)$ }  & \multicolumn{1}{|r}{
\hspace{0.6pt} $30$ } & \multicolumn{1}{r}{\hspace{0.6pt} $100$ } &
\multicolumn{1}{c|}{ $300$ }   & \multicolumn{1}{|r}{ \hspace{0.6pt}
$30$ } & \multicolumn{1}{r}{\hspace{0.6pt} $100$ } &
\multicolumn{1}{c|}{ $300$ }  &  \multicolumn{1}{|r}{ \hspace{0.6pt}
$30$ } & \multicolumn{1}{r}{\hspace{0.6pt} $100$ } &
\multicolumn{1}{c|}{ $300$ }  & \multicolumn{1}{|r}{ \hspace{0.6pt}
$30$ } & \multicolumn{1}{r}{\hspace{0.6pt} $100$ } &
\multicolumn{1}{c|}{ $300$
} \\
\multicolumn{1}{|c|}{ }  & \multicolumn{3}{|c|}{
\begin{centering}
\end{centering}
 } &  \multicolumn{3}{|c|}{ \begin{centering}
 \end{centering}
 } & \multicolumn{3}{|c|}{ \begin{centering}
\end{centering}
 } &  \multicolumn{3}{|c|}{ \begin{centering}
 \end{centering}
 } \\
\multicolumn{1}{|l|}{  $60$ }  & \multicolumn{3}{|l|}{ \begin{centering} $0.897$  \hspace{0.1in} $-$ \hspace{0.16in} $-$ \end{centering}  } &  \multicolumn{3}{|l|}{ \begin{centering} $0.560$ \hspace{0.1in} $-$ \hspace{0.16in} $-$ \end{centering}  } & \multicolumn{3}{|l|}{ \begin{centering} $0.993$ \hspace{0.1in} $-$ \hspace{0.16in} $-$ \end{centering}  } &  \multicolumn{3}{|l|}{ \begin{centering} $ 0.369$ \hspace{0.1in} $-$ \hspace{0.16in} $-$ \end{centering}  } \\
 \multicolumn{1}{|l|}{  $200$ }  & \multicolumn{3}{|l|}{ \begin{centering} $0.892$ $  0.901$ \hspace{0.1in} $-$ \end{centering} } &  \multicolumn{3}{|l|}{ \begin{centering} $0.563$ $ 0.361$  \hspace{0.1in} $-$ \end{centering} } & \multicolumn{3}{|l|}{ \begin{centering} $0.993$ $0.999$ \hspace{0.1in} $-$ \end{centering}  } &  \multicolumn{3}{|l|}{ \begin{centering} $0.368$ $0.221 $ \hspace{0.1in} $-$ \end{centering}
 } \\
\multicolumn{1}{|l|}{$600$}&\multicolumn{3}{|l|}{
\begin{centering} $0.898$ $0.902$ $0.904$ \end{centering}  } & \multicolumn{3}{|l|}{
\begin{centering} $0.567$ $0.363$ $0.261$ \end{centering}  } &
\multicolumn{3}{|l|}{ \begin{centering} $0.994$ $0.998$ $0.999$
\end{centering}  } & \multicolumn{3}{|l|}{ \begin{centering} $0.369$
$0.222$ $0.134$ \end{centering}  } \\ \hline
\end{tabular}
\end{table}
\noindent \footnotesize{{\bf Note to Table~4}:  data are generated according to model
\[
 y_{it}  = \alpha_{i0} + \beta_{i0} x_{it} + b_{i10} f_{1t} + b_{i20} f_{2t}  +\varepsilon_{it}
 \]
with regressor  $ x_{it}  = \gamma_{i0} + \delta_{i10} f_{1t} + \delta_{i30} f_{3t}  + v_{it}$.
Factor loadings are normally distributed random variables, $iid$ across units and mutually independent, satisfying
$b_{i10} \sim NID (1,0.2), $ \hspace{0.1in} $ b_{i20} \sim NID (0,0.2) ,$ $ \delta_{i10} \sim NID(0.5,0.5),$ $  \delta_{i30} \sim NID(0,0.5). $
Latent common factors are
$  { f }_{j,t} = $ $  0.5  { f }_{j,t-1} $ $ + \sqrt{0.5} { \eta }_{jf,t}, $
 with $  { \eta }_{jf,t}  $ $ \sim NID(0,1) $, mutually independent for $j=1,2,3$, and
  idiosyncratic innovation are
  $
    \varepsilon_{it}  = $ $ \rho_{i\varepsilon }   { \varepsilon }_{it-1}$ $ + { \eta }_{i\varepsilon,t}$ with $
 { \eta }_{i\varepsilon,t} $ $ \sim NID( 0,   \sigma_i^2 ( 1 -  \rho_{i\varepsilon }^2) ),$ $
    v_{it}   = $ $  \rho_{iv }   { v }_{it-1} + { \eta }_{iv,t},$ with $
 { \eta }_{iv ,t} $ $ \sim NID( 0,   ( 1 -  \rho_{i v }^2) ),$
with $ \rho_{i \varepsilon } $ $ \sim UID(0.05,0.95),$ $  \rho_{i v } $ $ \sim UID(0.05,0.95),$ $ \sigma_{i \varepsilon }^2 $ $ \sim UID(0.5,1.5)$, $iid$ across $i=1,...,N$ and mutually independent.

\noindent Parameters of interest are  constant across replications and equal to $ \alpha_{i0} = 1,  \gamma_{i0} = 0.5 $ and, assuming $N$ even, $
\beta_{i0} =  1 $ for   $i=1,...,N/2$  and $\beta_{i0} =  3  $ for $ i=N /2+1,...,N.$

\noindent Panels headed by \texttt{mean} and \texttt{rmse} report, respectively,  $ N^{-1} \sum_{i=1}^N \left( MM^{-1} \sum_{m=1}^{MM} \hat{\alpha}_i^m  \right) $ and  \newline
 $  N^{-1} \sum_{i=1}^N  \left( MM^{-1} \sum_{m=1}^{MM} (   \hat{\alpha}_i^m -   1 )^2 \right)^{1 \over 2} $ with $MM=2,000$. Here $ \hat{ \alpha }_i^m $ denotes the estimate, based on either the  GLS 
 $ \tilde{\bm \alpha }_i^{GLS}  $ of Section~(\ref{common})  (top left panel), multi-step GLS  with $J=4$ steps (top right panel), OLS $  \bm{\hat{ \alpha }_i^{OLS}}  $   (bottom left panel) and UGLS $ \bm{\hat{ \alpha }_i^{UGLS}}  $ (bottom right panel) of $ \alpha_{i0} $ for the $m$ Monte Carlo iteration.
}


\newpage

\begin{table}  \label{MC2}
\vspace{-1.30in}
\centering {\footnotesize\bf
$\,\,\,\,\,\,\,\,\,\,\,\,\,\,\,\,\,\,\,\,\,\,\,$  Table 5: \\ 
 time
regression with unit-specific coefficients \\ regression coefficient
$ \beta_{i0}  = 1, \,\,\, i=1,...,N/2.$ }
  \footnotesize
\begin{tabular}{|lcccccccccccc|}
\hline
 \multicolumn{1}{|c|}{ }  &
 \multicolumn{6}{|c|}{  } &  \multicolumn{6}{|c|}{  }
 \\
\multicolumn{1}{|c|}{ }  & \multicolumn{6}{|c|}{ \begin{centering}
$GLS$ \end{centering}
 } &
 \multicolumn{6}{|c|}{ \begin{centering}  $GLS$ (multi-step)   \end{centering} }
\\ \hline
\multicolumn{1}{|c|}{ }  & \multicolumn{3}{|c|}{ \begin{centering}
\texttt{mean}
\end{centering}
 } &  \multicolumn{3}{|c|}{ \begin{centering}
\texttt{rmse} \end{centering}
 } & \multicolumn{3}{|c|}{ \begin{centering} \texttt{mean}
\end{centering}
 } &  \multicolumn{3}{|c|}{ \begin{centering}
\texttt{rmse} \end{centering}
 } \\ \hline
\multicolumn{1}{|c|}{  $(N,T)$ }  & \multicolumn{1}{|r}{
\hspace{0.6pt} $30$ } & \multicolumn{1}{r}{\hspace{0.6pt} $100$ } &
\multicolumn{1}{c|}{ $300$ }   & \multicolumn{1}{|r}{ \hspace{0.6pt}
$30$ } & \multicolumn{1}{r}{\hspace{0.6pt} $100$ } &
\multicolumn{1}{c|}{ $300$ }  &  \multicolumn{1}{|r}{ \hspace{0.6pt}
$30$ } & \multicolumn{1}{r}{\hspace{0.6pt} $100$ } &
\multicolumn{1}{c|}{ $300$ }  & \multicolumn{1}{|r}{ \hspace{0.6pt}
$30$ } & \multicolumn{1}{r}{\hspace{0.6pt} $100$ } &
\multicolumn{1}{c|}{ $300$
} \\
\multicolumn{1}{|c|}{  }  & \multicolumn{3}{|c|}{ \begin{centering} \end{centering}  } &  \multicolumn{3}{|c|}{ \begin{centering}  \end{centering}  } & \multicolumn{3}{|c|}{ \begin{centering} \end{centering}  } &  \multicolumn{3}{|c|}{ \begin{centering}  \end{centering}  } \\
\multicolumn{1}{|l|}{  $60$ }  & \multicolumn{3}{|l|}{ \begin{centering} $1.089$  \hspace{0.1in} $-$ \hspace{0.16in} $-$ \end{centering}  } &  \multicolumn{3}{|l|}{ \begin{centering} $0.266$  \hspace{0.1in} $-$ \hspace{0.16in} $-$ \end{centering}  } & \multicolumn{3}{|l|}{ \begin{centering} $1.028$  \hspace{0.1in} $-$ \hspace{0.16in} $-$ \end{centering} } &  \multicolumn{3}{|l|}{ \begin{centering} $0.198$  \hspace{0.1in} $-$ \hspace{0.16in} $-$ \end{centering}  } \\
\multicolumn{1}{|l|}{  $200$ }  & \multicolumn{3}{|l|}{ \begin{centering} $1.039$ $ 1.079 $ \hspace{0.1in}  $-$ \end{centering} } &  \multicolumn{3}{|l|}{ \begin{centering} $0.175$ $ 0.191 $ \hspace{0.1in}   $-$ \end{centering} } & \multicolumn{3}{|l|}{ \begin{centering} $1.008$ $1.015 $ \hspace{0.1in}  $-$ \end{centering}  } &  \multicolumn{3}{|l|}{ \begin{centering} $0.148$ $  0.096$\hspace{0.1in}  $-$ \end{centering} } \\
\multicolumn{1}{|l|}{  $600$ }  & \multicolumn{3}{|l|}{
\begin{centering} $1.026$ $1.028$ $1.078$ \end{centering} } &  \multicolumn{3}{|l|}{ \begin{centering} $0.155$ $0.097$ $0.167$ \end{centering}  } & \multicolumn{3}{|l|}{ \begin{centering} $1.006$ $1.002$ $1.012$ \end{centering}  } & \multicolumn{3}{|l|}{ \begin{centering} $0.138$ $0.069$ $0.055$
\end{centering} } \\ \hline
 \multicolumn{1}{|c|}{ }  &
 \multicolumn{6}{|c|}{  } &  \multicolumn{6}{|c|}{  }
 \\
\multicolumn{1}{|c|}{ }  & \multicolumn{6}{|c|}{ \begin{centering}
$OLS$ \end{centering}
 } &
 \multicolumn{6}{|c|}{ \begin{centering}  $UGLS$    \end{centering} }
\\ \hline
\multicolumn{1}{|c|}{ }  & \multicolumn{3}{|c|}{ \begin{centering}
\texttt{mean}
\end{centering}
 } &  \multicolumn{3}{|c|}{ \begin{centering}
\texttt{rmse} \end{centering}
 } & \multicolumn{3}{|c|}{ \begin{centering} \texttt{mean}
\end{centering}
 } &  \multicolumn{3}{|c|}{ \begin{centering}
\texttt{rmse} \end{centering}
 } \\ \hline
\multicolumn{1}{|c|}{  $(N,T)$ }  & \multicolumn{1}{|r}{
\hspace{0.6pt} $30$ } & \multicolumn{1}{r}{\hspace{0.6pt} $100$ } &
\multicolumn{1}{c|}{ $300$ }   & \multicolumn{1}{|r}{ \hspace{0.6pt}
$30$ } & \multicolumn{1}{r}{\hspace{0.6pt} $100$ } &
\multicolumn{1}{c|}{ $300$ }  &  \multicolumn{1}{|r}{ \hspace{0.6pt}
$30$ } & \multicolumn{1}{r}{\hspace{0.6pt} $100$ } &
\multicolumn{1}{c|}{ $300$ }  & \multicolumn{1}{|r}{ \hspace{0.6pt}
$30$ } & \multicolumn{1}{r}{\hspace{0.6pt} $100$ } &
\multicolumn{1}{c|}{ $300$
} \\
\multicolumn{1}{|c|}{ }  & \multicolumn{3}{|c|}{
\begin{centering}
\end{centering}
 } &  \multicolumn{3}{|c|}{ \begin{centering}
 \end{centering}
 } & \multicolumn{3}{|c|}{ \begin{centering}
\end{centering}
 } &  \multicolumn{3}{|c|}{ \begin{centering}
 \end{centering}
 } \\
\multicolumn{1}{|l|}{  $60$ }  & \multicolumn{3}{|l|}{ \begin{centering} $1.179$  \hspace{0.1in} $-$ \hspace{0.16in} $-$ \end{centering}  } &  \multicolumn{3}{|l|}{ \begin{centering} $0.451$  \hspace{0.1in} $-$ \hspace{0.16in} $-$ \end{centering}  } & \multicolumn{3}{|l|}{ \begin{centering} $1.013 $  \hspace{0.1in} $-$ \hspace{0.16in} $-$ \end{centering}  } &  \multicolumn{3}{|l|}{ \begin{centering} $ 0.176$  \hspace{0.1in} $-$ \hspace{0.16in} $-$ \end{centering}  } \\
 \multicolumn{1}{|l|}{  $200$ }  & \multicolumn{3}{|l|}{ \begin{centering} $1.179$ $  1.170$  \hspace{0.1in} $-$ \end{centering} } &  \multicolumn{3}{|l|}{ \begin{centering} $0.452$ $ 0.374$  \hspace{0.1in} $-$ \end{centering} } & \multicolumn{3}{|l|}{ \begin{centering} $1.014$ $1.003$  \hspace{0.1in} $-$ \end{centering}  } &  \multicolumn{3}{|l|}{ \begin{centering} $0.177$ $0.087$  \hspace{0.1in}  $-$ \end{centering}
 } \\
\multicolumn{1}{|l|}{$600$}&\multicolumn{3}{|l|}{ \begin{centering}
$1.179$ $1.171$ $1.169$ \end{centering}  } & \multicolumn{3}{|l|}{
\begin{centering} $0.451$ $0.373$ $0.347$ \end{centering}  } &
\multicolumn{3}{|l|}{ \begin{centering} $1.014$ $1.004$ $1.001$
\end{centering}  } & \multicolumn{3}{|l|}{ \begin{centering} $0.177$
$0.088$ $0.049$ \end{centering}  } \\ \hline
\multicolumn{13}{l}{
 \footnotesize{{\bf Note to Table~5}:   data are generated according to the same model described in Table~1.}} \\
 \multicolumn{13}{l}{
 \footnotesize{  Panels headed by \texttt{mean} and \texttt{rmse} report, respectively, 
  $ (N/2)^{-1} \sum_{i=1}^{N/2} \left( MM^{-1} \sum_{m=1}^{MM} \hat{\beta}_i^m  \right) $ }} \\
   \multicolumn{13}{l}{ \footnotesize{ and  $  (N/2)^{-1} \sum_{i=1}^{N/2}  \left( MM^{-1} \sum_{m=1}^{MM} (   \hat{\beta}_i^m -   1 )^2 \right)^{1 \over 2} $ with $MM=2,000$. Here $ \hat{ \beta }_i^m $ denotes the
  }} \\
  \multicolumn{13}{l}{ \footnotesize{
     estimate, based on either the  GLS  $ \bm{\hat{ \beta }_i^{GLS}} $  in equation (\ref{eq:bgls}) (top left panel), multi-step GLS $\bm{\hat{ \beta }}_i^{(J)} $  }} \\
   \multicolumn{13}{l}{ \footnotesize{  with $J=4$ steps
     (top right panel), OLS  $  \bm{\hat{ \beta }_i^{OLS}}  $ (bottom left panel) and UGLS  $ \bm{\hat{ \beta }_i^{UGLS}}  $ }} \\
  \multicolumn{13}{l}{ \footnotesize{  (bottom right panel) of $ \beta_{i0} $
 for the  $m$th Monte Carlo iteration.}}
\end{tabular}
\end{table}


\begin{table}  \label{MC3}
\centering {\footnotesize\bf
$\,\,\,\,\,\,\,\,\,\,\,\,\,\,\,\,\,\,\,\,\,\,\,$  Table 6: \\ time
regression with unit-specific coefficients \\ regression coefficient
$ \beta_{i0}  = 3, \,\,\, i=N/2+1,...,N.$ }
  \footnotesize
\begin{tabular}{|lcccccccccccc|}
\hline
 \multicolumn{1}{|c|}{ }  &
 \multicolumn{6}{|c|}{  } &  \multicolumn{6}{|c|}{  }
 \\
\multicolumn{1}{|c|}{ }  & \multicolumn{6}{|c|}{ \begin{centering}
$GLS$ \end{centering}
 } &
 \multicolumn{6}{|c|}{ \begin{centering}  $GLS$ (multi-step)   \end{centering} }
\\ \hline
\multicolumn{1}{|c|}{ }  & \multicolumn{3}{|c|}{ \begin{centering}
\texttt{mean}
\end{centering}
 } &  \multicolumn{3}{|c|}{ \begin{centering}
\texttt{rmse} \end{centering}
 } & \multicolumn{3}{|c|}{ \begin{centering} \texttt{mean}
\end{centering}
 } &  \multicolumn{3}{|c|}{ \begin{centering}
\texttt{rmse} \end{centering}
 } \\ \hline
\multicolumn{1}{|c|}{  $(N,T)$ }  & \multicolumn{1}{|r}{
\hspace{0.6pt} $30$ } & \multicolumn{1}{r}{\hspace{0.6pt} $100$ } &
\multicolumn{1}{c|}{ $300$ }   & \multicolumn{1}{|r}{ \hspace{0.6pt}
$30$ } & \multicolumn{1}{r}{\hspace{0.6pt} $100$ } &
\multicolumn{1}{c|}{ $300$ }  &  \multicolumn{1}{|r}{ \hspace{0.6pt}
$30$ } & \multicolumn{1}{r}{\hspace{0.6pt} $100$ } &
\multicolumn{1}{c|}{ $300$ }  & \multicolumn{1}{|r}{ \hspace{0.6pt}
$30$ } & \multicolumn{1}{r}{\hspace{0.6pt} $100$ } &
\multicolumn{1}{c|}{ $300$
} \\
\multicolumn{1}{|c|}{  }  & \multicolumn{3}{|c|}{ \begin{centering} \end{centering}  } &  \multicolumn{3}{|c|}{ \begin{centering}  \end{centering}  } & \multicolumn{3}{|c|}{ \begin{centering} \end{centering}  } &  \multicolumn{3}{|c|}{ \begin{centering}  \end{centering}  } \\
\multicolumn{1}{|l|}{  $60$ }  & \multicolumn{3}{|l|}{ \begin{centering} $3.105$  \hspace{0.1in} $-$ \hspace{0.16in} $-$ \end{centering}  } &  \multicolumn{3}{|l|}{ \begin{centering} $0.314$  \hspace{0.1in} $-$ \hspace{0.16in} $-$ \end{centering}  } & \multicolumn{3}{|l|}{ \begin{centering} $3.041$  \hspace{0.1in} $-$ \hspace{0.16in} $-$ \end{centering} } &  \multicolumn{3}{|l|}{ \begin{centering} $0.277$  \hspace{0.1in} $-$ \hspace{0.16in} $-$ \end{centering}  } \\
\multicolumn{1}{|l|}{  $200$ }  & \multicolumn{3}{|l|}{ \begin{centering} $3.053$ $ 3.095 $ \hspace{0.1in}  $-$ \end{centering} } &  \multicolumn{3}{|l|}{ \begin{centering} $0.228$ $ 0.227 $ \hspace{0.1in}   $-$ \end{centering} } & \multicolumn{3}{|l|}{ \begin{centering} $3.013$ $3.024 $ \hspace{0.1in}  $-$ \end{centering}  } &  \multicolumn{3}{|l|}{ \begin{centering} $0.215$ $  0.146$\hspace{0.1in}  $-$ \end{centering} } \\
\multicolumn{1}{|l|}{  $600$ }  & \multicolumn{3}{|l|}{
\begin{centering} $3.034$ $3.037$ $3.091$ \end{centering} } &  \multicolumn{3}{|l|}{ \begin{centering} $0.209$ $0.133$ $0.198$ \end{centering}  } & \multicolumn{3}{|l|}{ \begin{centering} $3.012$ $3.004$ $3.019$ \end{centering}  } & \multicolumn{3}{|l|}{ \begin{centering} $0.200$ $0.111$ $0.090$
\end{centering} } \\ \hline
 \multicolumn{1}{|c|}{ }  &
 \multicolumn{6}{|c|}{  } &  \multicolumn{6}{|c|}{  }
 \\
\multicolumn{1}{|c|}{ }  & \multicolumn{6}{|c|}{ \begin{centering}
$OLS$ \end{centering}
 } &
 \multicolumn{6}{|c|}{ \begin{centering}  $UGLS$    \end{centering} }
\\ \hline
\multicolumn{1}{|c|}{ }  & \multicolumn{3}{|c|}{ \begin{centering}
\texttt{mean}
\end{centering}
 } &  \multicolumn{3}{|c|}{ \begin{centering}
\texttt{rmse} \end{centering}
 } & \multicolumn{3}{|c|}{ \begin{centering} \texttt{mean}
\end{centering}
 } &  \multicolumn{3}{|c|}{ \begin{centering}
\texttt{rmse} \end{centering}
 } \\ \hline
\multicolumn{1}{|c|}{  $(N,T)$ }  & \multicolumn{1}{|r}{
\hspace{0.6pt} $30$ } & \multicolumn{1}{r}{\hspace{0.6pt} $100$ } &
\multicolumn{1}{c|}{ $300$ }   & \multicolumn{1}{|r}{ \hspace{0.6pt}
$30$ } & \multicolumn{1}{r}{\hspace{0.6pt} $100$ } &
\multicolumn{1}{c|}{ $300$ }  &  \multicolumn{1}{|r}{ \hspace{0.6pt}
$30$ } & \multicolumn{1}{r}{\hspace{0.6pt} $100$ } &
\multicolumn{1}{c|}{ $300$ }  & \multicolumn{1}{|r}{ \hspace{0.6pt}
$30$ } & \multicolumn{1}{r}{\hspace{0.6pt} $100$ } &
\multicolumn{1}{c|}{ $300$
} \\
\multicolumn{1}{|c|}{ }  & \multicolumn{3}{|c|}{
\begin{centering}
\end{centering}
 } &  \multicolumn{3}{|c|}{ \begin{centering}
 \end{centering}
 } & \multicolumn{3}{|c|}{ \begin{centering}
\end{centering}
 } &  \multicolumn{3}{|c|}{ \begin{centering}
 \end{centering}
 } \\
\multicolumn{1}{|l|}{  $60$ }  & \multicolumn{3}{|l|}{ \begin{centering} $3.201$  \hspace{0.1in} $-$ \hspace{0.16in} $-$ \end{centering}  } &  \multicolumn{3}{|l|}{ \begin{centering} $0.488$  \hspace{0.1in} $-$ \hspace{0.16in} $-$ \end{centering}  } & \multicolumn{3}{|l|}{ \begin{centering} $3.013 $  \hspace{0.1in} $-$ \hspace{0.16in} $-$ \end{centering}  } &  \multicolumn{3}{|l|}{ \begin{centering} $ 0.181$  \hspace{0.1in} $-$ \hspace{0.16in} $-$ \end{centering}  } \\
 \multicolumn{1}{|l|}{  $200$ }  & \multicolumn{3}{|l|}{ \begin{centering} $3.204$ $  3.196$  \hspace{0.1in} $-$ \end{centering} } &  \multicolumn{3}{|l|}{ \begin{centering} $0.490$ $ 0.415$  \hspace{0.1in} $-$ \end{centering} } & \multicolumn{3}{|l|}{ \begin{centering} $3.012$ $3.003$  \hspace{0.1in} $-$ \end{centering}  } &  \multicolumn{3}{|l|}{ \begin{centering} $0.180$ $0.090$  \hspace{0.1in}  $-$ \end{centering}
 } \\
\multicolumn{1}{|l|}{$600$}&\multicolumn{3}{|l|}{ \begin{centering}
$3.204$ $3.197$ $3.193$ \end{centering}  } & \multicolumn{3}{|l|}{
\begin{centering} $0.492$ $0.417$ $0.389$ \end{centering}  } &
\multicolumn{3}{|l|}{ \begin{centering} $3.012$ $3.003$ $3.001$
\end{centering}  } & \multicolumn{3}{|l|}{ \begin{centering} $0.180$
$0.090$ $0.051$ \end{centering}  } \\ \hline
\multicolumn{13}{l}{
 \footnotesize{{\bf Note to Table~6}:   data are generated according to the same model described in Table~1.}} \\
 \multicolumn{13}{l}{
 \footnotesize{  Panels headed by \texttt{mean} and \texttt{rmse} report, respectively, 
  $ (N/2)^{-1} \sum_{i=N/2+1}^N \left( MM^{-1} \sum_{m=1}^{MM} \hat{\beta}_i^m  \right) $ }} \\
   \multicolumn{13}{l}{ \footnotesize{ and  $  (N/2)^{-1} \sum_{i=N/2+1}^N  \left( MM^{-1} \sum_{m=1}^{MM} (   \hat{\beta}_i^m -   3 )^2 \right)^{1 \over 2} $ with $MM=2,000$. Here $ \hat{ \beta }_i^m $ denotes the
  }} \\    \multicolumn{13}{l}{ \footnotesize{
     estimate, based on either the  GLS  $ \bm{\hat{ \beta }_i^{GLS}} $  in equation (\ref{eq:bgls}) (top left panel), multi-step GLS $\bm{\hat{ \beta }}_i^{(J)} $  }} \\
   \multicolumn{13}{l}{ \footnotesize{  with $J=4$ steps
     (top right panel), OLS  $  \bm{\hat{ \beta }_i^{OLS}}  $ (bottom left panel) and UGLS  $ \bm{\hat{ \beta }_i^{UGLS}}  $ }} \\
  \multicolumn{13}{l}{ \footnotesize{  (bottom right panel) of $ \beta_{i0} $
 for the  $m$th Monte Carlo iteration.}}
\end{tabular}
\end{table}

\appendix
\normalfont
\begin{center}
\begin{Huge}
\textsc{
\textbf{Appendices}
}
\end{Huge}
\end{center}

All the proofs and technical details are reported in the three appendixes (Appendixes  A,B and C) of the manuscripts and in the seven appendixes of the Supplement
(Appendixes D,E,F,G,H,I,J).

\bigskip

In particular, Appendix A states (without proof; see Appendix F)  Lemmas A.1 to A.3 and  Appendix B contains the proofs to Theorem 3.1 and 3.2  and states
 (without proofs; see Appendix G) Propositions B.1 to B.17.
Finally, Appendix C defines the $\SS $ matrix and associated quantities that characterize the  asymptotic distribution of the GLS estimator.

\bigskip 

Regarding the Supplement, Appendix D contains needed results of linear matrix algebra (Lemmas D.1 to D.6), Appendix E  contains some results used to construct  bounds on the inverse of various covariance matrixes  (Lemmas E.1 to E.3 and Corollary E.3),  Appendix F contains the proofs to Lemmas A.1 to A.3, Appendix G 
contains the proofs to Propositions B.1 to B.17, Appendix H contains important auxiliary results for the proof of Theorem 3.2 (Lemmas H.1 to  H.19), Appendix I formalizes the 
asymptotic properties of the estimators for the common observed regressors' coefficient (Theorems I.1 to I.3) and, finally, Appendix J provides some technical details  for the case when the regressors and the residuals have different, yet correlated, factor structures.

\section{Central Lemmas}\label{centralemmas}
In the following  $m,m_1,m_3,m_3$ denote positive constants. The proofs of the lemmas stated in this section are provided in Appendix  \ref{proofAppA}.
\begin{lemma}\label{PZ}
  Let ${\bf A }(m_1
\times m_1)$, ${\bf C
}(m_2 \times m_2)$, and ${\bf B }(m_1 \times m_2) $, $m_1>m_2$ be random matrices.
Set $ {\bf E:=BCB' +   A   }$ and  assume that $\lambda_1\left(\A\right)=O_p(1)$. Assume further that $m_2<\infty$, and 
\begin{enumerate}[label=(\alph*)]
\item $\lb\frac{\bm{B'A^{-1}B}}{m_1}\rb=O_p(1)\;$ and $\quad\lb\left(\frac{\bm{B'A^{-1}B}}{m_1}\right)^{-1}\rb=O_p(1)$.
\item $\lb\left(\frac{\bm{ C^{-1} +
B'A^{-1}B}}{m_1} \right)^{-1}\rb =O_p(1)$.
\item  $\lb\frac{\bm{B'(A^{-1})'A^{-1}B}}{m_1}\rb=O_p(1)$.
\item $\lb \bm{C} \rb=O_p(1)$ and $\lb \bm{C}^{-1} \rb=O_p(1)$.
\end{enumerate}   
Then,
\begin{enumerate}[label=(\roman*)]
\item 
$
\parallel \bm{ E }^{-1} \bm{ B } \parallel_2=O_p\left( m_1^{-\nicefrac{1}{2}}\right). $
\end{enumerate}
Let $\bm{D}$ be a $m_1\times m_3$ matrix, with $m_3<\infty$.
\begin{enumerate}[label=(\roman*)]
\setcounter{enumi}{1}
\item 
If $   \lb m_1^{-1}\bm{D}' \bm {A }^{-1 
}  \bm{ B }\rb_2
 = O_p ( 1 ) $,
 then
$ \lb \bm{D}'  \bm{ E }^{-1} \bm{ B }\rb =O_p\left( 1 \right). $
\item  If  $  \lb m_1^{-\nicefrac{1}{2}}\bm{D}' \bm{ A }^{-1}
\bm{ B }\rb_2
 = O_p (1) $, then  $\lb \bm{D}'  \bm{ E }^{-1} \bm{ B } \rb =O_p\left(m_1^{-\nicefrac{1}{ 2 }} \right)$.
\end{enumerate}

\end{lemma}

\begin{remark}\label{remarkPZ}
Let $\bm{\bar{E}}:=\bm{A}^{-1}-
\bm{A}^{-1}\bm{B}
\left(
\bm {B'A^{-1}B
}
\right)^{-1}\bm{B}'\bm{A}^{-1}$.  Write $\bm{H}:=\bm{A^{-1}B}$; following  \cite[p.40]{j96} we find
\berr
\bm{\bar{E}}&=& 
\bm{A}^{-1}\left(\I_T-\bm{A}^{-1}\bm{B}\left(\bm {B'A^{-1}B}\right)^{-1}\bm{B}'\right)=\bm{H_{\bot}}\left(\bm{B_{\bot}'}\bm{H_{\bot}}\right)^{-1}\bm{B'_{\bot}},
\eerr
implying that $\bm{\bar{E}}\bm{B}=\bm{0}$. Because 
$$
\lb \bm{E}^{-1}-\bm{\bar{E}}\rb=\lb\bm{A^{-1}BC^{-1}B'A^{-1}}\rb\leq \lb\frac{\bm{A^{-1}B}}{m^{\nicefrac{1}{2}}}\rb^2\lb\frac{\bm{C^{-1}}}{m}\rb=O_p\left(\frac{1}{m^{\nicefrac{1}{2}}}\right),
$$
 it follows that the inverse of the matrix $\bm{E}$ approximates a matrix orthogonal to $\bm{B}$.
\end{remark}
%
%
%
%
\begin{lemma}\label{approx inv}
Let $\bm{A},\bm{B}$ nonsingular matrices. Then 
\berr
(i)\quad\bm{A}^{-1}&=&\bm{B}^{-1}-\bm{B}^{-1}(\bm{A}-\bm{B})\bm{B}^{-1}+\bm{B}^{-1}(\bm{A}-\bm{B})\bm{A}^{-1}(\bm{A}-\bm{B})\bm{B}^{-1}\\
(ii)\quad\bm{A}^{-1}&=& \bm{B}\sum_{j=0}^J\left[(-1)\left(\bm{A}-\bm{B}\right)\bm{B}\right]^{j}+\bm{B}(\bm{A}-\bm{B})\bm{A}^{-1}\left[(-1)\bm{B}\left(\bm{A}-\bm{B}\right)\right]^{j}
\eerr
for $J=1,2,\dots.$
\end{lemma}
\begin{lemma}\label{paolo}
Let $g( \omega) $  be a periodic (mod $2 \pi $) symmetric function defined over $ - \pi \le  \omega  \le \pi $ with bounded  $r$-th order derivative $g^{(r)}(\cdot)$,  some $ r \ge 1 $. Then the Fourier coefficients $ \varsigma_h = \int_{-\pi }^\pi  g( \omega ) \cos (  h \omega ) d \omega $   satisfy:
\[
| \varsigma_h | = O \left( {1 \over h^r  } \right)\quad \mbox{ as } \qquad h \rightarrow \infty .
\]
\end{lemma}

\section{Proof of the main theorems}\label{proofmain}
The proofs of the theorems rely on several propositions, the proofs of which are relegated in Appendix \ref{propth1}
\subsection{Proof of Theorem~\ref{Theorem_OLS} }\label{thuno}

\subsubsection*{Proof of part (i)}
Rewrite equation (\ref{eq:charlie1}) as
\be
\bm{\hat{ \beta }_i^{OLS}} - \bm{{ \beta }_i} -  ( \xx'\xx )^{-1} \xx'\ff \b  = ( \xx'\xx )^{-1} \xx '\ei.  \label{eq:ols1}
\ee
Proposition \ref{propo1} implies that
\be\label{eq:radio106}
{ \xx'\xx \over T }  = \Gi'  { \ff '  \ff \over T }  \Gi +  { \vv' \vv \over T }  + 
  { \vv' \ff  \over T }  \Gi + \Gi' { \ff'\vv \over T } \xrightarrow{p} \sx>0,
\ee
and 
\be\label{eq:radio107}
\frac{\xx'\ff}{T}=\Gi\frac{\ff' \ff}{T}+\frac{\vv' \ff}{T}.
\xrightarrow{p} \sxf \b.
\ee
Hence $( \xx'\xx )^{-1} \xx'\ff \b\b\xrightarrow{p}\bm{\tau_i^{OLS}}$,  proving $(\ref{eq:bias17})$.

Next, we derive the asymptotic distribution of  
\begin{equation}\label{eq:twoterms}
 T^{-{\nicefrac{1}{2}}} \xx'\ei = T^{-\nicefrac{1}{2}} \Gi ' \f ' \md \epi + T^{-{\nicefrac{1}{2}}} \v' \md\epi .
\end{equation}
We first show that the first term in (\ref{eq:twoterms}) satisfies
\be\label{eq:200117a}
T^{-\nicefrac{1}{2}}{ \Gi '\f' \md \epi   } 
   \xrightarrow{d}
 \norm\left(0, \Gi ' ( -    \bm{ \Sigma} _{\d'\d}^{-1} , \I_M  ) \bm{ \Sigma }_{\z' \Hi \z}  ( - \bm{ \Sigma} _{\f'\d}  \bm{ \Sigma}_{\d'\d}^{-1} , \I_M )' \Gi  \right) ,
\ee
Note that
$\Gi '\f' \md \epi=\Gi ' (  - \f '\d( \d' \d)^{-1}, \I_M ){ \z' \epi },$
 with $\z$ defined in Assumption \ref{ass factors}. Because of  (\ref{eq:maggio1}), to prove (\ref{eq:200117a}) it suffices to prove that 
\be\label{eq:radio110} 
 T^{-\frac{1}{2}}\z\bm{'\varepsilon_i}\xrightarrow{d}\norm\left(0,\bm{ \Sigma }_{\z' \Hi \z}\right).
 \ee
Adapting \cite[Proof of Theorem 1]{RH97}, and using equations (\ref{eq:11CHL}) and (\ref{eq:11}) we have
\be\label{eq:cdieci}
  T^{-\nicefrac{1}{2}}\z' \epi  =  \sum_{j=1}^N r_{ij} \left(   T^{-\nicefrac{1}{2}} \sum_{t=1}^T \zz_{t}   a_{jt}   \right) = 
  \sum_{j=1}^N r_{ij}\bm{w_{j}},
  %
\ee
 where $\bm{w_{j}}:=T^{-\nicefrac{1}{2}} \sum_{t=1}^T \zz_{t}   a_{jt} $, and $\zz'_t=(z_{t1},\dots,z_{t(M+S)})'$ is the $t-$th row of $\z$.  For $\tau_0=\tau_0(T)$ yet to be chosen, define
 \be\label{eq:cdieci3}
  \bm{w}_{j0}:=  T^{-{1 \over 2} } \sum_{u=- \tau_0 }^T \bm{s}_{ju}   \eta_{ju},
  \qquad \textrm{with} \qquad
 \bm{s}_{ju}:= \sum_{t=1}^T \zz_{t} \phi_{jt-u},\;
\phi_{jh} = 0, \mbox{ for } h < 0 ,
  \ee 
and let $\bm{w}_{j1}:=\bm{w}_{j}-\bm{w}_{j0}$, $\bm{W}_{j1}:=\bm{W}_{j}-\bm{W}_{j0}$,
where
\be\label{eq:cidieci3FEB}
\bm{W}_{j}:= \E \left( \bm{w}_{j} \bm{w}_{j}'  | \sa(\z)  \right),\qquad \mbox{and}
\qquad
 \bm{W}_{j0} :=  \E ( \bm{w}_{j0} \bm{w}_{j0}'  | \sa(\z) ).
\ee
Write
$$
\bm{W_{j}}^{-1/2}\bm{w_{j}}=\left(\bm{I_{(M+S)}}+\bm{W_{j1}}\bm{W_{j0}}^{-1}\right)^{-1/2}\bm{W_{j0}}^{-1/2}\bm{w_{j0}}+\bm{W_{j}}^{-1/2}\bm{w_{j1}}.
$$
Noting that $\E\lb\bm{W_{j1}}\rb\leq \E\lb\bm{w_{j1}}\rb^2$, Propositions \ref{prop190117a} and \ref{prop190117c} implies that, for $\tau_0$ increasing suitably with $T$, $\lb \bm{W_{j}}^{-1/2}\bm{w_{j1}}\rb=o_p(1)$ and $\lb\bm{W_{j1}}\bm{W_{j0}}^{-1}\rb=o_p(1)$. Hence, $ \bm{W_j}^{-1/2}\bm{w_j}\approx\bm{W_{j0}}^{-1/2}\bm{w_{j0}}$ as $T\to\infty$. Therefore, Bernstein's Lemma (see \cite{H70}, p. 242), Propositions \ref{prop190117b} and \ref{prop190117c} imply that, for every given $j$, $\bm{w_{j}}\xrightarrow{d}\norm(\bm{0},\bm{\mathcal{W}_{j}})$
with $\bm{\mathcal{W}_j}$ defined in Equation (\ref{eq:190117a}).
However, the terms $\bm{w_j}$ (and their limits) are uncorrelated  across $j$'s implying that 
$
 T^{-{1 \over 2}}\z' \epi=\sum_{j=1}^N  r_{ij}\bm{w_j} \xrightarrow{d} \mathcal{N}(\bm{0}, \sum_{j=1}^N r_{ij}^2 \bm{\mathcal{W}_{j}} ),
$ 
 where notice that   by easy calculations  $   \sum_{j=1}^N r_{ij}^2 {\bm{\mathcal{W} }}_{j}  = \bm{ \Sigma }_{\z' \Hi \z }$, proving (\ref{eq:radio110}).
It is worth noting that as $N$ increases, the latter distribution
 can be made arbitrarily close to $  \norm\left(\bm{0},  \sum_{j=1}^\infty r_{ij}^2  \bm{ \mathcal{W}_{j}}  \right) $ because of the absolute summability of the $r_{ij}$.

\bigskip
The second term in (\ref{eq:twoterms}) satisfies $T^{-\nicefrac{1}{2}}\v'\md\epi\approx T^{-\nicefrac{1}{2}}\v'\epi$ (see Proposition \ref{propo1}\ref{propo1iv}).  The proof of the weak convergence,
\be\label{eq:radio105}
T^{-\nicefrac{1}{2}}\v'\epi\xrightarrow{d}\norm\left(0,\bm{\Sigma}_{\v'\Hi\v}\right),
\ee
is very similar to the proof of (\ref{eq:radio110}), and hence omitted.
By assumption (\ref{independence}), the two terms on the RHS of (\ref{eq:twoterms}) are uncorrelated, so by (\ref{eq:200117a}) and (\ref{eq:radio105}) its LHS converges weakly to a  random variable  with mean zero, and variance given in Equation (\ref{eq:sigmaXHX}). 
%
\begin{prop}\label{propo1}
\begin{enumerate}[label=(\roman*)]
\item[]
\item $\frac{\ff'\ff}{T}={ \f ' \md \f  \over T } \xrightarrow{p} \bm{ \Sigma }_{\f'\md  \f} =  \bm{ \Sigma }_{\f' \f} - \bm{ \Sigma }_{\f'\d} \bm{ \Sigma }_{\d'\d}^{-1} \bm{ \Sigma }_{\d'\f}$ .\label{propo1D1}
\item $\frac{\vv'\vv}{T}={ \v '\md\v  \over T }  \xrightarrow{p} \bm{ \Sigma}_{\v'\v}.
$\label{propo1D2}
\item   $ \lb\frac{\vv'\ff}{T}\rb=\lb{ {\v ' \md \f} \over T }\rb = O_p(T^{-\nicefrac{1}{2}} )$.\label{propo1D3}
\item  $ \lb{{\v ' \epi} \over \sqrt{T}}-{{\v ' \md \epi} \over \sqrt{T} }\rb = O_p(T^{-\nicefrac{1}{2}} )$.\label{propo1iv}
\item $\SO_{\xx'\xx}>0.$\label{propo15M}
\end{enumerate}
\end{prop}
\begin{prop}\label{prop190117a}
For $\tau_0$ increasing suitably with $T$
$$
\lim_{T\to\infty}\E\lb\bm{w}_{j1}\rb^2=0.
$$
\end{prop}

\begin{prop}\label{prop190117b}
As $T\to\infty$
\be\label{eq:scott0}
\bm{W}_{j0} ^{-{1 \over 2}} \bm{w}_{j0}  \xrightarrow{d} \norm\left(\bm{0}, \bm{I_{(M+S)}} \right).
\ee
\end{prop}
\begin{prop}\label{prop190117c} 
As $T\to\infty$
\begin{equation}\label{eq:190117a}
\bm{W}_j \xrightarrow{p} \E (\eta_{j0}^2 )  \bm{\Sigma }_{\z' \bm{\Phi }_j \z }=:\bm{\mathcal{ W}_j} > 0, 
\end{equation} and 
 $\bm{\Phi }_j  $ is the $T\times T$ matrix with $(t,s)-$th element equal to $ \sum_{u=- \infty }^{\min(t,s)}  \phi_{j,t-u} \phi_{j,s-u} =  \sum_{v=0 }^{\infty}  \phi_{j,v} 
 \phi_{j,v+|t-s|}$.
\end{prop}
%
%
%
%
%
%
%
%
\subsubsection*{Proof of  part (ii)}  

All the limits below hold as $T \rightarrow \infty $.  By simple manipulation of equation (\ref{eq:UGLSdef})
\be
\sqrt{T}\left(\bm{\hat{\beta }_i^{UGLS} - { \beta }_{i}} -
  \left(\xx'  \ssthi \xx \right)^{-1}
\xx'  \ssthi  \ff \b\right) = 
\left(\frac{\xx'  \ssthi \xx}{T} \right)^{-1}
\left(\frac{\xx'  \ssthi   \ei}{\sqrt{T}}\right)  .  \label{eq:ugls11} 
\ee
We first show that estimator is asymptotically unbiased. 
Proposition \ref{lemma160117a} implies  
\be\label{eq:radio111D}
\frac{\xx'  \ssthi \xx}{ T } = \Gi'  \frac{ \ff ' \ssthi \ff}{T } \Gi + \frac{ \vv ' \ssthi \vv}{T } +  \frac{ \vv '  \ssthi \ff}{ T }  \Gi + \Gi' \frac{ \ff '  \ssthi \vv}{ T}\xrightarrow{p}\bm{ \Sigma }_{\v' \Hninv \v } ,
\ee
and
$$
\lb \xx ' \ssthi\ff \b\rb \leq   \lb\Gi\rb\lb \ff ' \ssthi\ff \rb \lb\b\rb +\lb \vv ' \ssthi\ff \b\rb=O_p(1),
$$
implying that  the the bias term is $O_p\left(T^{-\nicefrac{1}{2}}\right)$.

To complete the proof we need to derive the limiting distribution of the latter term in (\ref{eq:ugls11}). By Proposition \ref{lemma170117a}, $T^{-\nicefrac{1}{2}}\xx ' \ssthi\ei\aa T^{-\nicefrac{1}{2}}\v ' \Hninv \epi$.
Proposition \ref{prop020217a}$(i)$ shows   that,  $T^{-\nicefrac{1}{2}}\v ' \Hninv \epi\aa$ $ T^{-\nicefrac{1}{2}}$ $\v ' \Hbar  \epi$, where $\Hbar$, defined in the same proposition, is circulant and symmetric. Hence, 
 as stated by the second part of  Proposition \ref{prop020217a}, 
 $\|\Hbar\|_{row}<\infty$, allowing us to exploit again \cite[Propositons 1 and 2]{RH97}.

Similarly to (\ref{eq:cdieci})-(\ref{eq:cidieci3FEB}), we define
\be\label{eq:radio112}
  T^{-\nicefrac{1}{2}} {\v'\Hbar \epi}  = 
   \sum_{j=1}^N  r_{ij} \left(T^{-{\nicefrac{1}{2}}}  \sum_{t=1}^T\sum_{s=1}^T {\vvv}_{it}\xii_N(t-s)   \mathnormal{a_{js}}   \right)
=\sum_{j=1}^N r_{ij}\bm{w^{(\xii)}_{ij}} ,
\ee
 where $\xii_N(t-s)$, denoting the $(t,s)-$element of $\Hbar$, satisfies  $\xii_N(h)=\xii_N(T-h)$, $h=0,1,\dots,T-1$. Write
\be\label{eq:LEO}
\bm{w}^{(\xii)}_{ij0} :=    T^{-{\nicefrac{1}{2}} }  \sum_{u=- \tau_0 }^T \bm{s}^{(\xii)}_{iju}   \eta_{ju},\quad
\bm{s}^{(\xii)}_{iju}:= \sum_{s=1}^T \bm{\ell}^{(\xii)}_{is} \phi_{js-u}, \quad\bm{\ell}^{(\xii)}_{is}:=\sum_{t=1}^{T} \vvv_{it}\xii_N(t-s),
\ee
and define $ \bm{w}^{(\xii)}_{ij1}:= \bm{w}^{(\xii)}_{ij}- \bm{w}^{(\xii)}_{ij0}$, and  $ \bm{W}^{(\xii)}_{ij1}:= \bm{W}^{(\xii)}_{ij}- \bm{W}^{(\xii)}_{ij0}$, where
\be\label{eq:radio112b}
\bm{W}^{(\xii)}_{ij} := \E \left( \bm{w}^{(\xii)}_{ij} \bm{w}^{(\xii)'}_{ij}| \sa(\v)\right),\quad \mbox{and}
    \quad
  \bm{W}^{(\xii)}_{ij0} := \E \left( \bm{w}^{(\xii)}_{ij0} \bm{w}^{(\xii)'}_{ij0}| \sa(\v)\right).
\ee
%
Proceeding as in the proof of the first part of the theorem,  Propositions \ref{propCHLOE1}-\ref{propCHLOE3} allow us to establish that, for any $i,j$ 
$\bm{w^{(\xii)}_{ij}}\xrightarrow{d}\norm(\bm{0},\bm{\mathcal{W}^{(\xii)}_{ij}})$, where $\bm{\mathcal{W}^{(\xii)}_{ij}}$ defined in  (\ref{eq:06021717a})
below.
It follows that, by (\ref{eq:radio112}) 
\be\label{eq:radio115}
 T^{-{\nicefrac{1}{2}}}{\v'\Hbar \epi}  \xrightarrow{d} \norm\left(\bm{0}, \sum_{j=1}^N r_{ij}^2 \bm{{\cal {\bm W}}^{(\xii)}_{ij}} \right),
\ee 
  and  $   \sum_{j=1}^N r_{ij}^2 \bm{\mathcal{W}^{(\xii)}_{ij}}=\bm{ \Sigma }_{\v' \Hninv\Hi\Hninv \v }$, as required. The result in (\ref{eq:theorem1part2}) is proved by (\ref{eq:radio111D}) and  (\ref{eq:radio115}).

\begin{prop}\label{lemma160117a}
\begin{enumerate}[label=(\roman*)]
\item[]
\item  $\lb\ff ' \ssthi \ff\rb  = O_p(1).$
\item $\lb\vv ' \ssthi \ff\rb = O_p\left(T^{-\nicefrac{1}{2}} \right).$ 
\item $T^{-1}\left( \vv '   \ssthi\vv\right)   \xrightarrow{d} \bm{ \Sigma }_{\v' \Hninv\v } > 0.$  \label{eq:lemma160117aTRE}
\end{enumerate}
\end{prop}
\begin{prop}\label{lemma170117a}
\begin{enumerate}[label=(\roman*)]
\item[]
\item  $\lb\ff' \ssthi \ei\rb = O_p\left(T^{-\nicefrac{1}{2}}\right).$
\item $\lb {\vv'}   {\ssthi} {\ei}  -\v'{\Hninv}{\epi}\rb=O_p(1).$  
\end{enumerate}
\end{prop}
\begin{prop}\label{prop020217a}
Let 
$\Hbar:=2\pi\P\bm{G}_{\xi_N}^{-1}\P'$, with $\P$ defined in (\ref{eq:peigenvectors}) and $\bm{G}_{\xi_N}=\diag({\bm g}(\xi_N,\omega))$ is defined in (\ref{eq:gt}) with $g_{\xi_N}(\omega)=\sum_{h=-\infty}^\infty \xi_{N}(h)\cos(h\omega)$, where $\xi_N(h)=\xi_{N,ts}$ is the $(t,s)-$entry of the matrix $\Hn$ defined in (\ref{eq:deffeb}). Then,
\begin{enumerate}[label=(\roman*)]
\item $\lb \v '\Hninv \epi  -\v'\Hbar\epi\rb=O_p(1)$.
\item $\lb \Hbar\rb_{row}<\infty$.
\end{enumerate}
\end{prop}
\begin{prop}\label{propCHLOE1}
For $\tau_0$ increasing suitably with $T$
$$
\lim_{T\to\infty}\E\lb \bm{w}^{(\xii)}_{ij1}\rb^2=0.
$$
\end{prop}
\begin{prop}\label{propCHLOE2}
As $T\to\infty$
\be\label{eq:scott0B}
\left(\bm{W}^{(\xii)}_{ij0}\right) ^{-{1 \over 2}} \bm{w}^{(\xii)}_{ij0}  \xrightarrow{d} \norm(\bm{0}, \bm{I_{m+s}} ).
\ee
\end{prop}
\begin{prop}\label{propCHLOE3} Let $\bm{\mathcal{ W}^{(\xii)}_{ij}}$ as defined in (\ref{eq:06021717a}) below. Then,
\begin{equation}\label{eq:06021717a}
\bm{W}^{(\xii)}_{ij} \xrightarrow{p} \E (\eta_{j0}^2 )  \bm{\Sigma }_{\v'\Hninv \bm{\Phi }_j \Hninv\v }  =: \bm{\mathcal{W}^{(\xii)}_{ij}} > 0 ,
\end{equation}
where the matrix $\bm{\Phi_j}$ has been defined in (\ref{eq:190117a})
\end{prop}
\subsection{Proof Theorem~\ref{Theorem_GLS}}\label{proofT2}

By Proposition \ref{23feb18} below
\be\label{eq:ugls11N} 
\sqrt{T}\left(\bm{\hat{\beta}_i^{FGLS}}-\bm{\hat{\beta}_i}\right)\aa
\left(T^{-1}\xx'\MAI\xx\right)^{-1} \left(T^{-\nicefrac{1}{2}}\xx'\MAI\ei\right),
\ee
where
\be\label{eq:MMM1}
\MA=\db'\left(\f \an\f'+\bm{C_N}\right)\db,
\ee
and $\an$ and $\cn$  are  defined in equations (\ref{eq:seconmomb2}) and (\ref{eq:MMM2}), respectively.
The matrix $\MA$ can be seen as the  FGLS counterpart of $\ssthi$. The proof follows closely that of Theorem~\ref{Theorem_OLS} part (ii). However, here we consider the joint asymptotics, for (N,T) diverging simultaneously, as can be appreciated from the inspection of Proposition \ref{23feb18}.
%
For the first term on the LHS of (\ref{eq:ugls11N}),
proposition \ref{lemma160117aT2} implies that
\ber\label{eq:radio111}
{ \xx '  \MAI \xx  \over T } &=& \Gi'  { \ff ' \MAI \ff  \over T } \Gi + { \vv ' \MAI \vv  \over T } +  { \vv '  \MAI \ff \over T }  \Gi + \Gi' { \ff '  \MAI \vv \over T}\nonumber\\
&\approx & \frac{\v'\cni\v}{T}
\xrightarrow{p}\bm{ \Sigma }_{\v' \cni \v }
\eer
The matrix $\bm{ \Sigma }_{\v' \cni \v }=\E\left(\v'\cni\v\right)$ is non-stochastic, but does depend on $N$, in general. 
To complete the proof we need to derive the limiting distribution of the latter term in (\ref{eq:ugls11N}). By and 
Proposition \ref{lemma170117aT2},
\[
\lb { \xx ' \MAI  \ei- \v ' \cn^{-1}  \epi \over   \sqrt{T }} \rb \leq  
\lb{ \Gi ' \ff ' \MAI  \ei  \over \sqrt{T}}\rb+
 \lb{ \vv ' \MAI  \ei  - \v ' \MAIN  \epi  \over \sqrt{T} }\rb=O_p\left(\frac{1}{\sqrt{T}}\right),
\] 
that is,  $T^{-\nicefrac{1}{2}}\xx ' \MAI\ei\approx T^{-\nicefrac{1}{2}}\v ' \MAIN  \epi$.  
Proposition \ref{mamma}.(i) show   that, in turn, $T^{-1/2}\v ' \MAIN  \epi\approx T^{-1/2}\v ' \MAB  \epi$, where $\MAB$ is a circulant matrix  defined in the same proposition. Hence, the LHS of (\ref{eq:ugls11N}) can be further approximated as
\be\label{eq:2luglio}
\sqrt{T}\left(\bm{\hat{\beta}_i^{FGLS}}-\bm{\hat{\beta}_i}\right)
\approx \left(\frac{\v'\cni\v}{T}\right)^{-1}
\frac{\v'\MAB\epi}{\sqrt{T}}.
\ee
The second part of  Proposition \ref{mamma} states that the rows of the matrix $\MAB$ are absolutely summable,  allowing us to exploit again \cite[Propositons 1 and 2]{RH97}.

Similarly to (\ref{eq:radio112})-(\ref{eq:radio112b}), we write
\be\label{eq:radio112N}
  T^{-{1 \over 2}}\v'\MAB \epi =  \sum_{j=1}^N r_{ij} \left(   T^{-{1 \over 2}}  \sum_{t=1}^T \vvv_{it}\cb(t-s)   a_{js}   \right)=\sum_{j=1}^N r_{ij}\bm{w^{(\bar{c})}_{ij}}\nonumber,
\ee
 where $\cb(t-s)=\bar{c}_{N,ts}$ is the $(t,s)$-entry of the $(T\times T)$ matrix $\MAB$ and 
$$
\bm{w}^{(\bar{c})}_{ij} :=  \left( T^{-\nicefrac{1}{2} }  \sum_{u=- \infty  }^T \bm{s}^{(\bar{c})}_{iju}   \eta_{ju}   \right),\quad 
\bm{s}^{(\bar{c})}_{iju}:= \sum_{t=1}^T \bm{\ell}^{(\bar{c})}_{it} \phi_{jt-u}
,\quad  
\bm{\ell}^{(\bar{c})}_{it}:=\sum_{s=1}^{T} \vvv_{it}\cb(t-s),
$$
with  $\phi_{jh} = 0$ for $h < 0$.
 By Assumption \ref{ass eps} it also follows that, for any $N$
\be\label{eq:27feb}
T^{-1}\E\left(\v'\MAB \epi\epi'\MAB\v\left|\right.\left\{  \sa(\v) \right\}\right)=\sum_{i=1}^N r^2_{ij}\bm{W}^{(\bar{c})}_{ij},
\ee
with $\bm{W}^{(\bar{c})}_{ij}:= \E \left( \bm{w}^{(\bar{c})}_{ij} \bm{w}^{(\bar{c})'}_{ij}| \{   \sa(\v) \}  \right).$
Proceeding along the lines of Theorem 1, Part (ii) using Propositions \ref{propCHLOE1N}-\ref{propCHLOE2N} we establish that, 
\be\label{eq:radio115N}
\left(\sum_{i=1}^N r^2_{ij}\bm{W}^{(\bar{c})}_{ij}\right)^{-1/2} T^{-{1 \over 2}}\v'  \MAB\epi \xrightarrow{d} \norm \left(\bm{0}, \I_{K} \right).
\ee 
The result in display (\ref{eq:theorem2B}) follows from display (\ref{eq:2luglio}) and Proposition \ref{propMarch1}.
\begin{prop}\label{23feb18}
For $1/T+T^3/N^2\to 0$:
\begin{enumerate}[label=(\roman*)]
\item $\lb T^{-1}\xx'(\shi-\MAI)\xx\rb=o_p(1)$.
\item $\lb T^{-\nicefrac{1}{2}}\xx'(\shi-\MAI)\ei\rb=o_p(1)$.
\end{enumerate}
\end{prop}
%
%
\begin{prop}\label{lemma160117aT2}
For any $N$:
\begin{enumerate}[label=(\roman*)]
\item  $\lb\ff ' \MAI \ff\rb  = O_p(1).$
\item $\lb\vv ' \MAI \ff\rb = O_p\left(T^{-\nicefrac{1}{2}} \right).$ 
\item  $\lb T^{-1}\left( \vv '   \MAI \vv-\v'\cni\v\right)\rb=o_p(1).$
\item $T^{-1}\left( \vv '   \MAI \vv\right)   \xrightarrow{p} \bm{ \Sigma }_{\v' \cn^{-1} \v } > 0.$  
\end{enumerate}
\end{prop}
\begin{prop}\label{lemma170117aT2} 
For any $N$
\begin{enumerate}[label=(\roman*)]
\item  $\lb\Gi'\ff' \MAI \ei\rb = O_p\left(T^{-\nicefrac{1}{2}}\right).$
\item $\lb \vv '   \MAI\ei  -\v'\cn^{-1}\epi\rb=O_p(1).$  
\end{enumerate}
\end{prop}
\begin{prop}\label{mamma}
Let 
$\MAB:=2\pi\P\bm{G}_{c_N}^{-1}\P'$, with $\P$ defined in (\ref{eq:peigenvectors}) and $\bm{G}_{c_N}=\diag({\bm g}(c_N,\omega))$ is defined as in display (\ref{eq:gt}) with $g_{c_N}(\omega)=\sum_{h=-\infty}^\infty c_{N}(h)\cos(h\omega)$. The scalar $c_N(h)=c_{N,ts}$ is  the $(t,s)-$entry of the matrix $\cn$ defined in (\ref{eq:MMM2}). Then, $\forall N$
\begin{enumerate}[label=(\roman*)]
\item $\lb \v '\cn^{-1} \epi  -\v'\MAB\epi\rb=O_p(1)$.
\item $\lb \MAB\rb_{row}<\infty$.
\end{enumerate}
\end{prop}
\begin{prop}\label{propCHLOE1N}
For any $N$ and $T_0$ increasing suitably with $T$
$$
\lim_{T\to\infty}\E\lb \sqrt{T}\bm{w}^{(\cb)}_{ij1}\rb^2=0.
$$
\end{prop}
\begin{prop}\label{propCHLOE2N}
For any $N$, as $T\to\infty$
$$
\bm{W}_{ij0} ^{-{1 \over 2}} \bm{w}_{ij0}  \xrightarrow{d} \norm(\bm{0}, \I_K).
$$
\end{prop}
\begin{prop}\label{propMarch1} For any $N$
$$
\lb \sum_{i=1}^N r^2_{ij}\bm{W}_{ij}^{(\bar{c})}-\frac{1}{T}\v'\cni\Hi\cni\v\rb=O_p\left(\frac{1}{\sqrt{T}}\right).
$$
\end{prop}
%
%
%
%
%
\section{The matrix $\SS$} \label{auxiliary}
%
%
 By Equation (\ref{eq:ref_def}), (\ref{eq:storti}) and (\ref{eq:charlie2})
\berr
\sh &=&\frac{1}{N}\sum_{i=1}^N\mxx\uu \uu ' \mxx\\
&=&\frac{1}{N}\sum_{i=1}^N\left(\I-\ff\Gi\xp-\vv\xp\right)\left(\ff\b+\ei \right)\left(\ff\b+\ei \right)'  
\left(\I-\ff\Gi\xp-\vv\xp\right)'\\
&=:&
\ff \han \ff' + \bm{\hat{\sc}_N},
\eerr
with 
$$
\han:=\frac{1}{N}\sum_{i=1}^N \bm{\hat{A}_i},
\qquad
\bm{\hat{\sc}_N}:=\frac{1}{N}\sum_{i=1}^N\left(\bm{\hat{\sc}_{1i}}+\bm{\hat{\sc}_{2i}}+\bm{\hat{\sc}_{3i}}
+\bm{\hat{\sc}_{4i}}
\right).
$$
To define $\bm{\hat{A}}_i$, note that
\berr
&&\left(\I-\ff\Gi\xp\right)\ff\b\b'\ff'
\left(\I-\ff\Gi\xp\right)'+\left(\ff\Gi\xp\right)\ei\ei'
\left(\ff\Gi\xp\right)'\\
&=&
\ff\left(\I-\Gi\xp\ff\right)\b\b'
\left(\I-\Gi\xp\ff\right)'\ff'+\ff\left(\Gi\xp\right)\ei\ei'
\left(\Gi\xp\right)'\ff'.
\eerr
Hence,
$
\bm{\hat{A}_i}=\left(\I-\Gi\xp\ff\right)\b\b'
\left(\I-\Gi\xp\ff\right)'+\left(\Gi\xp\right)\ei\ei'
\left(\Gi\xp\right)'.
$

Likewise, 
$
\bm{\hat{\sc}_{1i}}= \left(\I-\vv\xp\right)\ei\ei'
\left(\I-\vv\xp\right)',$ and $
\bm{\hat{\sc}_{2i}}=\left(\vv\xp\right)\ff\b\b'\ff'
\left(\vv\xp\right)'.
$
The term $\bm{\hat{\sc}_{3i}}$ is defined as
$
\bm{\hat{\sc}_{3i}}=\sum_{j=1}^{13} \left(
\bm{\hat{\sc}_{3i,j}}+\bm{\hat{\sc}'_{3i,j}}\right),
$
where
$$
\begin{array}{lcl}
\bm{\hat{\sc}_{3i,1}}= -\vv\xp\ff\b\b'\ff', & \phantom{ccccc}&
%
%
\bm{\hat{\sc}_{3i,2}}= \vv\xp\ff\b\b'\ff'
\left(\ff\Gi\xp\right)' ,\\
\bm{\hat{\sc}_{3i,3}}=\ff\b\ei', & &
\bm{\hat{\sc}_{3i,4}}=\vv\xp\ff\b\ei'\left(\vv\xp\right)' ,\\
\bm{\hat{\sc}_{3i,5}}=-\vv\xp\ff\b\ei', & &
\bm{\hat{\sc}_{3i,6}}=-\ff\b\ei'\left(\vv\xp\right)' ,\\
\bm{\hat{\sc}_{3i,7}}=-\ff\Gi\xp\ff\b\ei', &&
\bm{\hat{\sc}_{3i,8}}=\vv\xp\ff\b\ei'\left(\ff\Gi\xp\right)',\\
\bm{\hat{\sc}_{3i,9}}=\ff\Gi\xp\ff'\b\ei'\left(\vv\xp\right)', && 
\bm{\hat{\sc}_{3i,10}}=-\ff\Gi\xp\ei\ei',\\ 
\bm{\hat{\sc}_{3i,11}}=\ff\Gi\xp\ei\ei'\left(\vv\xp\right)', & &
\bm{\hat{\sc}_{3i,12}}=\ff\Gi\xp\ff\b\ei'\xp\Gi'\ff',
\\
\bm{\hat{\sc}_{3i,13}}=
-\ff\b\ei'\left(\ff\Gi\xp\right)' .&& 
\end{array}
$$
 Next, define the matrices 
\begin{equation}\label{eq:silver}
\SS:=\ff'\an\ff'+\bm{\breve{\sc}_{N}},\qquad\bm{\breve{\sc}_{N}}:=\db'\bcn\db,\qquad
\bcn:= \frac{1}{N}\sum_{i=1}^N \left(\Hi+\BT\right),\quad
\end{equation}
where  $\Hi$ and $\an$ are defined in Equations (\ref{eq:uu}) and (\ref{eq:seconmomb2}),
%
%
and
\be\label{eq:defpsi}
 \BT:=\E\left[\v\psinv \Gi\frac{\ff'\ff}{T}\b\b'\frac{\ff'\ff}{T}\Gi'\psinv\v'\left|\sa(\z)\right.\right],
\ee
with $\bm{\Psi}_i$ defined as in (\ref{eq:trenta})

Some properties of the matrices defined above are established Lemmas \ref{boundinverse} and \ref{expansion} in Appendix \ref{lemmata1}. In particular, Lemma \ref{boundinverse}  verifies  that the matrix $\SS$ satisfies the assumptions in Lemma \ref{PZ}, on which rely the proof of most of the results in Lemma \ref{cor}.
%

%
\newpage
\title{
\begin{center}
\Huge
Supplementary Material to\\  
\V{Robust  Nearly-Efficient Estimation    of Large Panels   with Factor Structures}
\end{center}
}
\maketitle

\newpage

\appendix
\normalfont
\begin{center}
\begin{Huge}
\textsc{
\textbf{Appendices}
}
\end{Huge}
\end{center}

\setcounter{section}{3}

This  Supplement is made by seven appendixes:  Appendix D contains some results of linear matrix algebra (Lemmas D.1 to D.6), Appendix E  contains some results used to construct  bounds on the inverse of various covariance matrices  (Lemmas E.1 to E.3 and Corollary E.3),  Appendix F contains the proofs to Lemmas A.1 to A.3 (stated in Appendix A of the manuscript), Appendix G 
contains the proofs to Propositions B.1 to B.17 (stated in Appendix B of the manuscript), Appendix H contains  auxiliary results for the proof of Theorem 3.2 (Lemmas H.1 to  H.19), Appendix I formalises the  asymptotic properties of the estimators for the common observed regressors' coefficient (Theorems I.1 to I.3) and, finally, Appendix J provides some technical details  for the case when the regressors and the residuals have different, yet correlated, factor structures.

\section{Some results on matrix algebra}\label{someresults}
In this section, we report for reference some auxiliary results on matrix algebra. The first two Lemmas are reported without proof. In the following, all the matrices have real entries. 
\begin{lemma}[Sherman-Morrison-Woodbury  formula]\label{smw}
 For every  matrices $\bm{A},\bm{B}$ and $\bm{C}$ of  suitable dimension:
\begin{equation}\label{eq:smw}
({\bm{BCB}' +   \bm{A}   })^{-1} = {\bm{A}^{-1} - \bm{A}^{-1} \bm{B} ( \bm{C}^{-1} +
\bm{B}'\bm{A}^{-1}\bm{B} )^{-1} \bm{B}' \bm{A}^{-1} }, 
\end{equation}
if all the involved inverses exist.
\end{lemma}
\bigskip
\begin{lemma}[\cite{bern}, Fact 6.4.10]\label{fact6410}
Let $\bm{A}(n\times m)$ a matrix of rank $m$ and $\bm{B}(m\times m)$ a positive definite matrix. Then,
\begin{equation}\label{eq:fact6410}
(\bm{ABA}')^+=\bm{A}\left(\bm{A}'\bm{A}\right)^{-1}\bm{B}^{-1}\left(\bm{A}'\bm{A}\right)^{-1}\bm{A}'.
\end{equation} 
\end{lemma}
\begin{lemma}\label{MN}
Let $\A$ be defined as in Lemma \ref{fact6410}. Assume further that $\A'\A=\bm{I_m}$. Let $\bm{C}(n\times n)$ a positive definite matrix,  and define $\bm{E}:=\projp_{\bm{A}}\bm{C}\projp_{\bm{A}}$. Then,
$$
\bm{C^{-1}}-\bm{E^+}=\bm{Q}\left(\bm{I_n-\bm{L}}\right)\bm{Q'}\geq 0,
$$
where  $\bm{Q}(n\times n)$ satisfies $\bm{Q}'\bm{Q}=\bm{C^{-1}}$ and the diagonal matrix $\bm{L}$ is obtained setting equal to zero the last $n-m$ diagonal elements of the matrix $\bm{I}_n$.  
\end{lemma}
\begin{proof}
By simple manipulation and Lemma \ref{fact6410}
\berr
\bm{C^{-1}}-\bm{E^+}&=&
\bm{C^{-\frac{1}{2}}}\left(
\bm{I_n}-\bm{C^{\frac{1}{2}}}\bm{E^+}\bm{C^{\frac{1}{2}}}
\right)\bm{C^{-\frac{1}{2}}}\\
&=&
\bm{C^{-\frac{1}{2}}}\left[
\bm{I_n}-\bm{C^{\frac{1}{2}}}\A\left(\A'\bm{C}\A\right)^{-1}\bm{C^{\frac{1}{2}}}
\right]
\bm{C^{-\frac{1}{2}}}=\bm{C^{-\frac{1}{2}}}\projm_{\mathcal{A}}\bm{C^{-\frac{1}{2}}},
\eerr
with $\bm{\mathcal{A}}=\bm{C^{\frac{1}{2}}}\A$. Rewriting $\projm_{\mathcal{A}}=\bm{PLP'}$, with $\bm{P'P}=\I_n$, the proof is completed setting $\bm{Q=C^{\frac{1}{2}}P}$.
\end{proof}
\begin{lemma}\label{emme}
\begin{enumerate}[label=(\roman*)]
\item[]
\item\label{emme1} Let $\bm{A}(m\times m)$ positive semidefinite, and  $\bm{B}(n\times m)$. Then
$$
\lambda_n(\bm{BAB'})\leq \lambda_m(\bm{A})\lambda_n(\bm{BB'}),\quad \tr(\bm{BAB'})\leq \tr(\bm{A})\lambda_n\left(\bm{B}\bm{B}'\right).
$$
\item\label{emme2}Let $\projp_{B}$ a projection matrix, with $\bm{B}(n\times m)$, $n>m$, and  $\bm{A}(m\times m)$. Then,
$$
\lambda_n\left(\projp_{B}\bm{A}\projp_{B}\right)\leq \lambda_n\left(\bm{A}\right).
$$
\item\label{emme3}For any  symmetric matrix $\bm{A}(n\times n)$ and positive semidefinite  matrix $\bm{B}(n\times n)$
$$
\lambda_1(\bm{A}+\bm{B})\geq \lambda_1(\bm{A}),\quad 
\lambda_n(\bm{A}+\bm{B})\leq \lambda_n(\bm{A})+\lambda_n(\bm{B}),
$$
and 
$$
\lambda_1\left(\bm{A}\right)\tr\left(\bm{B}\right)
\leq \tr\left(\bm{AB}\right)
\leq 
\lambda_n\left(\bm{A}\right)\tr\left(\bm{B}\right).
$$
 \item\label{emme4}Let $\bm{A}(n\times n)$ be a positive semidefinite  matrix, then
 $$
 \max_{i,j}|a_{ij}|\leq \max_ia_{ii}\leq \lambda_n(\bm{A}),\quad\textrm{for}\quad i,j=1,\dots,n.
 $$ .
\end{enumerate}
\end{lemma}
\begin{proof}
The first inequality of part \ref{emme1} can be found in \cite[p. 237]{MN88}. About the second inequality,  first note that $\tr(\bm{B'}\bm{A}\bm{B})=\tr(\bm{B}\bm{B}'\bm{A})$. The proof is concluded using the quasilinear representation of the extremal eigenvalues \cite[p. 204]{MN88}. For part \ref{emme2}, by Reayleigh quotient we have
$$
\max_{\bm{x}}\frac{\bm{x}'\projp_{B}\bm{A}\projp_{B}\bm{x}}{\bm{x}'\bm{x}}\leq \max_{\bm{x}}\frac{\bm{x}'\projp_{B}\bm{A}\projp_{B}\bm{x}}{\bm{x}'\projp_{B}\bm{x}}\leq \max_{\bm{x}} \frac{\bm{x}'\bm{A}\bm{x}}{\bm{x}'\bm{x}}=\lambda_n\left(\bm{A}\right) .
$$
For part \ref{emme3} see \cite[p. 204]{MN88}. Note that for the second inequality to hold, only requires $\bm{B}$ to be symmetric. Part \ref{emme4} follows from \cite[Exercise 8.7]{AM05} and \cite[Section 9.13.4(4)]{ltk}.
\end{proof}

\begin{lemma}\label{matrix}
Let $\bm{A}(m\times m)$ a matrix satisfying $\|\bm{A}\|_{col}<\infty$ and $\|\bm{A}\|_{row}<\infty$, and $\bm{p}$ and $\bm{q}$ two conformable vectors.  Then,
$$
\sup_{\bm{p},\bm{q}}\left|\sum_{i,j=1}^m p_{i}a_{ij}q_{j}\right|\leq 
\frac{\|\bm{A}\|_{row}\|\bm{A}\|_{col}
\|\bm{p}\|\|\bm{q}\|}{\lb\bm{A}\rb_{sp}}.
$$
\end{lemma}

\begin{proof}
Follows from (20) and (14), on pages 110 and 112 in \cite{ltk}, respectively.
\end{proof}

\begin{lemma}\label{matrix2}
Let $\bm{A}(m\times m)$ positive semi-definite  matrix with $\lambda_m(\bm{A})<\infty$, and denote with $\bm{a_i}$ its $i$-th row. Then, $\|\bm{a_i}\|<\lambda_m(\bm{A})$. 
\end{lemma}

\begin{proof}
Let $\bm{B}=\bm{A}\bm{A}=[b_{ij}]_{i,j=1}^m$. By the second inequality in Lemma \ref{emme}\ref{emme4}, $\|\bm{a_i}\|^2=b_{ii}\leq \lambda_{m}(\bm{B})=\lambda_m^2(\bm{A})$.
\end{proof}

\section{The inverse of a covariance matrix}\label{invcov}
\setcounter{equation}{0}

Lemma \ref{circulant} reports a result well known in time series analysis. We prefer to report it anyway with proof, both for completeness and to introduce some further notation. 
Approximating the matrix $\mc$ with the circulant symmetric matrix $\mc^{(s)}$,  is convenient for deriving the properties of the inverse of $\mc$.  Diagonalizing $\mc^{(s)}$ we obtain the matrix $\Lambda$   of easy interpretation (see \ref{eq:lambda17}). Lemma \ref{020217b} highlights the importance of Lemma \ref{paolo}. The latter implies that the inverse  of the matrix $\mc^{(s)}$ has bounded row-norm. The same property is inherited by the matrix $\bm{P'G^{-1}_{\varsigma}\P}$, that is more convenient to work with. The main advantage is that the nonzero entries of the diagonal matrix $\bm{G^{-1}_{\varsigma}}$ do not depend on $T$. Corollary \ref{circulantN} adapt the result derived for time series analysis to a panel framework.
\begin{lemma}\label{circulant} 
Let $\varsigma(\cdot)$ a real function defined on the integers. Assume that $\varsigma(\cdot)$ is symmetric, non-negative definite and $\sum_{h=1}^{\infty}h^{\delta}\left|\varsigma(h)\right|<\infty$, with $\delta\geq 1$.
Let 
\be\label{eq:gto}
g_{\varsigma}(\omega):=\frac{1}{2\pi}\sum_{h=-\infty}^{\infty}\varsigma(h)\cos(h\omega),\qquad -\pi\leq \omega\leq\pi
\ee
and  ${\bm{G}}_{\varsigma}:=\diag\left({\bm{g}}(\varsigma,\omega)\right)$ be the $T\times T$ matrix, with
\be\label{eq:gt}
{\bm{g}}(\varsigma,\omega):=\left\{
\begin{array}{l}
\left(g_{\varsigma}(0),g_{\varsigma}(\omega_1),g_{\varsigma}(\omega_1),\dots,g_{\varsigma}\left(\omega_{[T/2]}\right),g_{\varsigma}\left(\omega_{[T/2]}\right)\right)' \qquad\qquad\qquad\textrm{if T is odd,}\\
\left(g_{\varsigma}(0),g_{\varsigma}(\omega_1),g_{\varsigma}(\omega_1),\dots,g_{\varsigma}\left(\omega_{(T-2)/2}\right),g_{\varsigma}\left(\omega_{(T-2)/2}\right),g_{\varsigma}\left(\omega_{T/2}\right)\right)' \;\textrm{otherwise,}
\end{array}
\right.
\ee
where $\omega_j=2\pi j/T$, $j=0,\dots,[T/2]$ and $[x]$ denotes the integer part of $x$.
Define the real orthogonal  $T \times T$ matrix   ${\bm P } $ by
\begin{equation}\label{eq:peigenvectors}
\begin{array}{lll}
{\bm P } &:= \left[ {\bm q }_0,\; {\bm q }_1,\; {\bm s }_1,\;  \cdots, \;{\bm q }_{[T/2]},\;  {\bm s }_{[T/2]} \right]' &\qquad \textrm{if $T$ is odd},\\
{\bm P } &:= \left[ {\bm q }_0,\;  {\bm q }_1,\;{\bm s }_1,\;  \cdots, \; 2^{-1/2}{\bm q }_{T/2} \right]'&\qquad \textrm{otherwise},
\end{array} 
\end{equation}
 where  
\begin{eqnarray*}
\bm{q_j}&=&\sqrt{2/T}\left[
1\quad \cos\omega_j\quad \cos2\omega_j \quad \cdots \quad \cos(T-1)\omega_j
\right]',\\
\bm{s_j}&=&\sqrt{2/T}\left[
0\quad \sin\omega_j\quad \sin 2\omega_j \quad \cdots \quad \sin(T-1)\omega_j
\right]',
\end{eqnarray*}
and the matrix $\mc=\left[\varsigma(i-j)\right]_{i,j=1}^T$. Denote by $\Delta_{pq}^{(T)} $ is the $p,q$ component of the matrix 
$
 {\bm P } \mc {\bm P }' - 2\pi\bm{G}_{\varsigma},
$
then  
\be\label{eq:300117a}
\sup_{1 \le i,j \le T } | \Delta_{ij}^{(T)}|=O\left(\frac{1}{T}\right) \qquad\textrm{for}\quad i,j=1,\dots,T.
\ee
\end{lemma}

\begin{proof} The proof follows closely \cite[Proposition 4.5.2]{BD91}. Define the $T\times T$ circulant symmetric matrix 
\begin{equation}\label{eq:CS}
\mc^{(s)}=\left[
\begin{array}{cccccc}
\varsigma(0) & \varsigma(1) & \varsigma(2) & \cdots & \varsigma(2) & \varsigma(1)\\
\varsigma(1) & \varsigma(0) & \varsigma(1) & \cdots & \varsigma(3) & \varsigma(2)\\
\varsigma(2) & c(1) & \varsigma(0) & \cdots & \varsigma(4) & \varsigma(3)\\
\vdots & \vdots & \vdots & & \vdots & \vdots\\
\varsigma(1) & \varsigma(2) & \varsigma(3) & \cdots & \varsigma(1) & \varsigma(0)
\end{array}
\right].
\end{equation}
For brevity's sake we only consider the case of $T$ odd\footnote{For the case when $T$ is even the reader is referred to \cite{BD91}, p. 135}. Following \cite[Section 4.5]{BD91}, the above matrix can be diagonalized as $\P\mc^{(s)}\P'=\bm{\Lambda}$,  
\be\label{eq:lambda17}
\bm{\Lambda}=\textrm{diag}\left(\lambda_0,\;\lambda_1,\;\lambda_1,\;\dots,\;\lambda_{[T/2]},\;\lambda_{[T/2]}\right), 
\ee
with
\begin{equation}\label{eq:settembre}
\lambda_0=\sum_{|h|\leq [T/2]}\varsigma(h),\qquad
\lambda_j=
\sum_{|h|\leq [T/2]}\varsigma(h)\exp(ih\omega_j),\quad j=1,2,\dots [T/2].
\end{equation}
Let 
$
\bm{p_i}=\left[p_{i1},p_{i2},\dots,p_{iT}\right]
$
and $\bm{e_i}$
denote the $i^{th}$ row of the matrices $\P$, $\bm{I_T}$, respectively. Hence we have to show that
\be\label{eq:310117a}
\left|\bm{p_i} \mc\bm{p_j}'-2\pi\bm{e_i}\bm{G}_{\varsigma}\bm{e_j}'\right|\leq 
\left|\bm{p_i} \mc^{(s)}\bm{p_j}'-2\pi\bm{e_i}\bm{G}_{\varsigma}\bm{e_j}'\right|+
\left|\bm{p_i} \mc^{(s)}\bm{p_j}'-\bm{p_i} \mc\bm{p_j}'\right|
\ee
The first term on the right hand side of the above inequality
is bounded in absolute value by 
\be\label{eq:310117b}
\sum_{|h|>[T/2]}|\varsigma(h)|\leq \frac{2}{T^{\delta}}\sum_{|h|>[T/2]}h^{\delta}|\varsigma(h)|=O\left(\frac{1}{T^{\delta}}\right).
\ee
For the second term we have
\be\label{eq:310117c}
\left|
\bm{p_i\left(\mc^{(s)}-\mc\right)p_j'}\right|\leq \frac{4}{T}\left(2\sum_{h=1}^{[(T-1)/2]}h|\varsigma(h)|
+2\sum_{h=1}^{[(T-1)/2]}h|\varsigma(T-h)|\right)=O\left(\frac{1}{T}\right).
\ee
Since both (\ref{eq:310117b}),(\ref{eq:310117c}) terms are independent of $i$ and $j$, the proof is completed.
 \end{proof}

\begin{lemma}\label{020217b}
Let $\P$, $\bm{G_{\varsigma}}$ and $\bm{\Lambda}$ be as  in (\ref{eq:gt}), (\ref{eq:peigenvectors}) and (\ref{eq:lambda17}), respectively. 
Let $\lambda(\omega)=\sum_{|h|\leq [T/2]}\varsigma(h)\exp(ih\omega)$, with $\lambda(\omega_j)=\lambda_j$. Suppose that $\inf_\omega \lambda(\omega)>0$,  $\inf_\omega g_{\varsigma}(\omega)>0$, and the conditions
assumptions of Lemma \ref{circulant} are satisfied with $\delta\geq2$. 
Then, $\lb\P'\bm{G_{\varsigma}}^{-1}\P\rb_{row}<\infty$,  
\end{lemma}
\begin{proof}
For sake of brevity, we consider only the case of $T$ odd, as in the proof of Lemma \ref{circulant}. Write
$$
\lb\P'\bm{G_{\varsigma}}^{-1}\P\rb_{row}\leq \lb\left(\mc^{(s)}\right)^{-1}\rb_{row}+\lb\left(\mc^{(s)}\right)^{-1}-\bm{P}'\bm{G^{-1}_{\varsigma}}\bm{P}\rb_{row},
$$
where $\mc^{(s)}$ has been  defined in (\ref{eq:CS}).

The inverse of $\mc^{(s)}$ is circulant and symmetric, with eigenvalues $1/\lambda_j$, $j=0, \dots, \left[T/2\right]$, where $\lambda_j$ is defined as (\ref{eq:settembre}). 
Because
\be\label{eq:derivative}
\frac{d^2}{d\omega^2}\frac{1}{\lambda(\omega)}=-\frac{1}{\pi \lambda(\omega)^3}\sum_{|h|\leq [T/2]}h\varsigma(h)\sin(h\omega)-
\frac{1}{2\pi \lambda(\omega)^2}\sum_{|h|\leq [T/2]}h^2\varsigma(h)\cos(h\omega)
\ee
is continuous in $\omega$, Lemma \ref{paolo} implies that $\varsigma(h)$ is absolutely summable, and $\lb\left(\mc^{(s)}\right)^{-1}\rb_{row}<\infty$.

By the norm inequalities (6) and (12) in \cite[8.5.12, (12)]{ltk}
$$
\lb\left(\mc^{(s)}\right)^{-1}-\bm{P}'\bm{G^{-1}_{\varsigma}}\bm{P}\rb_{row}
=
\lb\bm{P}\left(\bm{\Lambda}^{-1}-\bm{G^{-1}_{\varsigma}}\right)\bm{P}'\rb_{row}
\leq \sqrt{T}\lb\bm{\Lambda}^{-1}-\bm{G^{-1}_{\varsigma}}\rb =\left(\frac{1}{T^{\delta-\nicefrac{1}{2}}}\right),
$$
where the latter equality follows from(\ref{eq:310117b}). 
\end{proof}

\begin{corollary}\label{circulantN} 
Let $\varsigma_i(\cdot)$, $i=1,\dots,N$ a set of  real functions defined on the integers. Assume that, 
\begin{enumerate}[label=(\alph*)]
\item $\varsigma_i(\cdot)$ is symmetric, non-negative definite and $\sup_i\sum_{h=1}^{\infty}h^{\delta}\left|\varsigma_i(h)\right|<\infty$, with $\delta\geq 2$ , $\forall \;i$;
\item $m_g:=\inf_i\inf_{\omega}g_i(\omega)>0$, and $M_g:=\sup_i\sup_{\omega}g_i(\omega)<\infty$;
\end{enumerate} 
where $g_{\varsigma_i}(\omega)=(2\pi)^{-1}\sum_{h=-\infty}^{\infty}\varsigma_i(h)\cos(h\omega)$, $-\pi\leq \omega\leq\pi$ and $m_g,M_g$ are positive constants. 
Let $\bm{\mathcal{C}_{N}}=\left[\varsigma_N(i-j)\right]_{\ell,j=1}^T$, with $\varsigma_N(h)=N^{-1}\sum_{i=1}^N\varsigma_i(h)$. Define the function
\be\label{eq:gtoN}
g_{\varsigma_N}(\omega)=\frac{1}{2\pi}\sum_{h=-\infty}^{\infty}\varsigma_N(h)\cos(h\omega),\qquad -\pi\leq \omega\leq\pi,
\ee
and  $T\times T$ diagonal matrix $G_{\varsigma_N}$ as in display(\ref{eq:gt}).
 Denote by $\Delta_{pq}^{(NT)} $ is the $p,q$ component of the matrix 
$
 {\bm P } \bm{\mathcal{C}_{N} } {\bm P }' - 2\pi\bm{G_{\varsigma_N}},
$
with $\bm{P}$ defined in (\ref{eq:peigenvectors}).  Then,
\begin{enumerate}[label=(\roman)]
\item[(i)]
$
\sup_{1 \le \ell,j \le T } | \Delta_{\ell,j}^{(NT)}|=O(1/T), \quad\textrm{as}\quad N,T\to\infty , \qquad \ell,j=1,\dots,T.
$
\item[(ii)] $\lb\bm{P'}\bm{G}^{-1}_{\varsigma_N}\bm{P}\rb_{row}<\infty$.
\end{enumerate}

\end{corollary}

\begin{proof}
Tonelli's lemma imply that, under Assumption $(a)$, $\sum_h\varsigma_N(h)=N^{-1}\sum_i\sum_h\varsigma_i(h)$. The proof of part $(i)$ then follows very closely equations (\ref{eq:310117a})-(\ref{eq:310117c}). 

Part $(ii)$ of the lemma is proved following the proof of Lemma \ref{020217b}.
A slight modification of the calculations in (\ref{eq:derivative}) show that  $d^2\lambda^{-1}_{N}(\omega)/d\omega^2$ is continuous in $\omega$, with 
$\lambda_N(\omega)=\sum_{|h|\leq [T/2]}\varsigma_N(h)\exp(ih\omega)$ .
\end{proof}

%
%
%
%
%
%
%
%
%
%
%
\section{Proof of the lemmas in Appendix  \ref{centralemmas}}\label{proofAppA}
\begin{proof}{\textbf{Lemma \ref{PZ}}}
\emph{Part (i)}
By the identity in equation (\ref{eq:inv1}) we have
$$
\left(
\bm{C^{-1}+B'A^{-1}B}\right)^{-1}
=\left(\bm {B'A^{-1}B}\right)^{-1}-\left(\bm {B'A^{-1}B}\right)^{-1}\bm{C^{-1}}\left(\bm {C^{-1}+B'A^{-1}B}\right)^{-1}.
$$
Combining the above equality with Lemma \ref{smw},  we get
\berr
\bm{E^{-1}}&=&
\bm{A}^{-1}-
\bm{A}^{-1}\bm{B}
\left(
\bm {C^{-1}+B'A^{-1}B
}
\right)^{-1}\bm{B}\bm{A}^{-1}\nonumber\\
&=&
\bm{A}^{-1}-
\bm{A}^{-1}\bm{B}
\left(
\bm {B'A^{-1}B
}
\right)^{-1}\bm{B}'\bm{A}^{-1}\label{eq:veccev}\\
&&+
\bm{A^{-1}B}\left[
\left(
\bm{C^{-1}+B'A^{-1}B}
\right)^{-1}
\bm{C^{-1}}
\left(
\bm{B'A^{-1}B}
\right)^{-1}
\right]\bm{B'A^{-1}}.\nonumber
\eerr
 Hence, by Assumptions $(a)-(d)$ 
\berr
\lb\bm{E}^{-1}\bm{B}\rb &\leq&
\lb\frac{\bm{A^{-1}B}}{m_1^{1/2}}\rb\lb
\left(\frac{
\bm{C^{-1}+B'A^{-1}B}}{m_1}
\right)^{-1}\rb
\lb\frac{\bm{C^{-1}}}{m_1^{1/2}}\rb=O_p\left(\frac{1}{m^{\nicefrac{1}{2}}_1}\right).\\
\eerr
\\
\emph{Parts (ii) and (iii)}: The proof follows straightforwardly from the inequality

 \[
\lb\bm{ D }'   \bm{ E }^{-1} \bm{ B }\rb  \leq \lb \frac{\bm{ D }'
 \bm{ A }^{-1} \bm{  B }}{m_1}\rb  \lb\left(  \frac{\bm{ C }^{-1}  +  \bm{ B } ' \bm{ A
}^{-1} \bm{ B } }{m_1}    \right)^{-1}\rb   \lb\bm{ C }^{-1}\rb .
 \]
\end{proof}

\begin{proof}\textbf{Lemma \ref{approx inv}}
By 3.5.2(5) in \cite{ltk}
\ber
\bm{A}^{-1}&=& \bm{B}^{-1}-\bm{B}^{-1}(\bm{A}-\bm{B})\bm{A}^{-1},\label{eq:inv1}\\
\bm{A}^{-1}&=& \bm{B}^{-1}-\bm{A}^{-1}(\bm{A}-\bm{B})\bm{B}^{-1}.\label{eq:inv2}
\eer
Plugging-in the RHS of (\ref{eq:inv2}) in $\bm{A}^{-1}$ in RHS of (\ref{eq:inv1}) proves  Part $(ii)$ of the Lemma. Similarly, Part $(ii)$ of the Lemma is proved replacing $r$ times the RHS of (\ref{eq:inv2}) for $\bm{A}^{-1}$ in the RHS of (\ref{eq:inv1}).
\end{proof}

%
%
%
%
%
%
\begin{proof}\textbf{Lemma \ref{paolo}} For a proof of case  $r=1$ see \cite[ Chapter II, Theorem 4.7, p.48]{Z93}.
 Consider case $r=2 $ and, without loss of generality, set $h> 0$.  Replace $\omega$ with $\omega-\pi/h$;  using integration by substitution, the Fourier coefficient $\varsigma_h $ can  be re-written as
\ber
\varsigma_h  & =&        \int_{- \pi + \pi/h  }^{\pi + \pi /h }    g\left(\omega - {\pi \over h } \right) \cos \left(  h \left( \omega -  \frac{\pi}{h} \right) \right) d \omega \nonumber\\
& = & \int_{- \pi + \pi/h  }^{\pi + \pi /h}    g\left(\omega - {\pi \over h } \right) \cos \left(  h  \omega  \right)  \cos \left(   -  \pi \right)    d \omega  \nonumber\\
    & = & 
   - \int_{- \pi   }^{\pi }    g\left(\omega - {\pi \over h } \right) \cos \left(  h \omega  \right)     d \omega . \label{eq:310117f}
\eer 
Similarly, substituting $\omega$ with $\omega+\pi/h$, we get
\be\label{eq:310117g}
\varsigma_h=
- \int_{- \pi   }^{\pi }    g\left(\omega + {\pi \over h } \right) \cos \left(  h \omega  \right)     d \omega .
\ee
 Hence, for some  for some $ 0 \le s_1=s_1(\omega) , s_2=s_2(\omega) , s_3=s_3(\omega) \le 1  $
\berr
\varsigma_h & =& -{1 \over 4 } \int_{- \pi }^\pi  \Big[    g\left(\omega + {\pi \over h } \right) + g\left(\omega - {\pi \over h } \right)  -2 g\left( \omega \right)   \Big] \cos (  h \omega ) d \omega \\
&  = & -{\pi \over 4  h } \int_{- \pi }^\pi   \Big[    g^{(1)}\left(\omega + s_1 {\pi \over h } \right) - g^{(1)}\left(\omega - s_2 {\pi \over h} \right)  
  \Big]  \cos (  h \omega ) d \omega \\
&  = & -{\pi^2  \over 4  h^2 } \int_{- \pi }^\pi ( s_1 + s_2  )   g^{(2)}\left(\omega + \left(s_1s_3 - (1-s_3)s_2     \right)\frac{\pi}{h} \right) \cos (  h \omega ) d \omega .
\eerr
The first equality follows from expressions (\ref{eq:310117f}) and(\ref{eq:310117g}). The 
 second and third equalities follows by applying twice the mean-value theorem.  

Since $| s_1 + s_2 | \le 2 $ and, by assumption  $| g^{(2)}( \omega ) | \le \ka $, we have  
 \[
 | \varsigma_h | \le    {\pi^2  \ka \over 4  \nu^2 } \int_{- \pi }^\pi  \left|( s_1 + (1-s_2)  )\right|  d\omega \le   {\pi^3  \ka \over   h^2 } = O \left(  {1 \over  h^2 } \right),
  \]
proving the theorem for $r=2$.  The general case follows along the same lines. 
\end{proof}

%
%
%

\section{Proof of the propositions in Appendix \ref{proofmain}}\label{propth1}
\setcounter{equation}{0}
%
%
%
%

%
%
\subsection{Proof of the propositions in Appendix \ref{thuno}}
\begin{proof}\textbf{of Proposition \ref{propo1}}
Part $(i)$ follows from Remark \ref{fpd}. Combining Remark \ref{inverseV}, Assumptions \ref{ass V} and   \ref{independence} and Lemma \ref{emme}\ref{emme1}
\berr
\E\lb\frac{\v'\pd}{T}\rb^2&=& \E\tr\left(\frac{\v'\pd\v}{T^2}\right)\leq
 \frac{\E\tr(\pd\E\v\v')}{T}\\
 &\leq & \frac{\tr(\pd)\lambda_K\left(\E\v\v'\right)}{T}=\frac{S}{T}\lambda_K\left(\E\v\v'\right)=O\left(\frac{1}{T}\right).
\eerr
The proof of Part $(ii)$ is completed noting that 
$$
\lb\frac{\v'\v}{T}-\frac{\v '\md\v}{T}\rb=\lb\frac{\v'\pd\v}{T}\rb\leq \lb \frac{\v'\pd}{T}\rb^2.
$$
Using similar arguments, for Part $(iii)$ we have
\berr
\E\lb\frac{\v'\md\f}{\sqrt{T}}\rb^2&=& \E\tr\left(\frac{\f'\md\f\v\v'}{T}\right)\leq
 \frac{\E\tr(\f'\md\f\E\v\v')}{T^2}\\
 &\leq & \frac{1}{T}\frac{\E\tr(\f'\md\f)}{T}\lambda_K\left(\E\v\v'\right)=O\left(\frac{1}{T}\right).
\eerr
noting that $\E\tr(\f'\md\f/T)\leq \E(\f'\f/T)<\infty$ by Assumption \ref{ass factors}. The result in Part $(iv)$ follows along the same lines, showing 
$$
\lb{{\v ' \epi} \over \sqrt{T}}-{{\v ' \md \epi} \over \sqrt{T} }\rb^2=
\lb{{\v ' \pd \epi} \over \sqrt{T} }\rb^2\leq \frac{\tr(\pd)}{T}\lambda_K(\E{\v\v'})\lambda_T(\E{\epi\epi'})=O\left(\frac{1}{T}\right).
$$
Finally, Part $(v)$ follows from Lemma \ref{emme}\ref{emme3} and Remark \ref{inverseV}.
\end{proof}

\begin{proof}\textbf{Proposition \ref{prop190117a}}
\ber
\E\lb \bm{w}_{j1}\rb^2&=&\frac{\E\left(\eta_{j0}^2\right)}{T}\sum_{u=-\infty}^{-\tau_0-1}
\sum_{t,s=1}^T\E(\zz_t'\zz_s)\phi_{jt-u}\phi_{js-u}\\
&\leq& \frac{\ka}{T}\sum_{t=1}^{\infty}|\phi_{jt}|
\sum_{u=-\infty}^{-\tau_0-1}\sum_{s=1}^T|\phi_{js-u}|\leq \ka\sum_{u=\tau_0}^{\infty}|\phi_{ju}|,
\eer
which tends to zero by Assumption \ref{ass eps} as $T\to\infty$, for a suitable choice $\tau_0$ increasing with $T$.
\end{proof}

\begin{proof}\textbf{Proposition \ref{prop190117b}}
Convergence in distribution in (\ref{eq:scott0}) is implied if the convergence holds conditional on $\sa(\z)$, and we establish the latter.
For any vector $\bm{h}$ satisfying $ \lb   \bm{h} \rb = 1 $ and any $T,\tau_0$, define 
${d}_{ju} := T^{-{\nicefrac{1}{2}} } \bm{h}'  \bm{W}_{j0} ^{-{1 \over 2}}    \bm{s}_{ju}  $.
Then, following \cite{RH97}, Proposition 2 (see also \cite{Scott73}, Theorem 1, Conditions (B)), it is sufficient to show that
\begin{equation}
\sum_{u=-T_0}^T  {d}_{ju}^2 \E \eta_{ju}^2 \xrightarrow{p}1 ,   \label{eq:scott1}
\end{equation}
and
\begin{equation}
\E \left[  \sum_{u=-T_0}^T   {d}_{ju}^2 \E\left(  \eta_{ju}^2 \bm{1}(|  d_{ju}  \eta_{ju} | > \delta_1) |  \sa(\z)   \right)        \right]  \rightarrow 0  \mbox{ for all } \delta_1 >  0 .  \label{eq:scott2}
  \end{equation}
  where $\bm{1}(\cdot)$ denotes the indicator function.
By construction $\sum_{-T_0}^T d^2_{ju}=1/(\E\eta^2_{j0})$, hence (\ref{eq:scott1}) follows trivially.  For some $\delta_2>0$, (\ref{eq:scott2}) is bounded by 
\be
\E \left[  \sum_{u=-T_0}^T   {d}_{ju}^2 \E\left(  \eta_{ju}^2 \bm{1}(|   \eta_{ju} | > \delta_1 / \delta_2  )   \right)        \right]  + \prob(  \max_u | d_{ju} | >  \delta_2 ).\label{eq:secondterm}
\ee
The first term is made arbitrarily small by choosing $ \delta_2$ suitably small and noting that $\eta_{ju}^2$ is uniformly integrable by Assumption (\ref{ass eps}) and   in \cite[Theorem 12.10]{Dav94}. Concerning the second term in (\ref{eq:secondterm}), Propositions (\ref{prop190117a}) and (\ref{prop190117c}) below  imply that $\|\bm{W}_{j0}^{-1}\|=O_p(1)$, so  we can  consider the set $ \parallel  \bm{W}_{j0}^{-{1 \over 2}}  \parallel \le \ka $ on which, for any $\epsilon>0$
\berr
\max_u  | d_{ju} | &\le& 
\max_u \ka T^{-\nicefrac{1}{2}}\sum_{t=1}^T\lb\zz_t\rb |\phi_{jt-u}|\\
&\le&
\ka \left(  \frac{\epsilon}{T} \sum_{t=1}^T  \parallel   \zz_t \parallel^2 \right)^{\nicefrac{1}{2}}  + \ka L T^{-{\nicefrac{1}{2}}} \max_{ 1 \le t \le T }  \lb  \zz_t  \rb  \max_u | \phi_{ju} |  ,
\eerr
for $L=L(\epsilon)$ satisfying $ \sum_{u \ge L }  \phi_{ju}^2  < \epsilon $. By Assumption \ref{ass factors},  $ \E \lb \bm{\mathrm{z}}_t \rb^2 < \infty $ and $ \max_{1 \le t \le T } \lb \bm{\mathrm{z}}_t \rb \le \ka \left( \sum_{t=1}^T \lb \bm{\mathrm{z}}_t \rb^{4} \right)^{\nicefrac{1}{4} }  = O_p( T^{\nicefrac{1}{4}  } ) $. The proof of (\ref{eq:scott2}) is concluded by Markov inequality. 
\end{proof}

\begin{proof}\textbf{Proposition \ref{prop190117c}}
Note that
 \begin{align}
 & \bm{W}_j  = { \E (\eta_{j0}^2 ) \over T } \sum_{u=- \infty }^T \Big(  \sum_{t=1}^T \sum_{s=1}^T  \zz_t \zz_s'    \phi_{jt-u} \phi_{js-u}  \Big) =  \E (\eta_{j0}^2 )  { \bm{Z}'   \bm{\Phi }_j  \bm{Z}   \over T } .
  \end{align}
  Remark \ref{unif lower bound}
implies that the matrix $ \bm{\Phi }_j $ is positive definite for a sufficiently large $T$. 
The proof follows letting $T$ diverge to infinity.
\end{proof}

\begin{proof}\textbf{of proposition \ref{lemma160117a}}
Part $(i)$ follows from Lemma \ref{PZ}(ii). 
To verify that
$\lb T^{-1}\ff'\left(\db'\Hninv\db\right)^{-1}\ff\rb=O_p(1)$, by Lemma \ref{emme}\ref{emme1}  it is enough to prove that $\lambda_1\left(\db'\Hn\db\right)>\ka$. Define $\bar{\d}:=\d\left(\d'\d\right)^{-\nicefrac{1}{2}}$, then
$
\lambda_i\left(\left[\db\;\bar{\d}\right]'\Hn\left[\db\;\bar{\d}\right]\right)
$
$
=\lambda_i\left(\Hn\right),
$
for $i=1,\dots,T$. The inclusion principle \cite[p. 160]{ltk} implies that 
\be\label{eq:incl}
\lambda_1(\Hn)\leq\lambda_1\left(\db'\Hn\db\right)
\leq\lambda_T\left(\db'\Hn\db\right)
\leq\lambda_T\left(\Hn\right).
\ee
Remark~\ref{unif lower bound}, and Lemma~\ref{emme}\ref{emme3} imply that
\be\label{eq:28febsnow}
\lambda_1\left(\Hn\right)\geq
\frac{1}{N}\sum_{i=1}^N \lambda_1\left(\Hi\right)\geq \inf_i \lambda_1\left(\Hi\right)>\ka,
\ee
proving the result.
Assumption $(a)-(d)$  of the Lemma are satisfied because of Assumptions \ref{ass factors} and \ref{ass loading}, completing the proof. 
Similarly,  Part $(ii)$ follows from Lemma \ref{PZ}, part $(iii)$ if we show that $\lb T^{-\nicefrac{1}{2}}\vv'\left(\db'\Hn\db\right)^{-1}\ff\rb=O_p(1)$. The latter is proved noting that, by Assumptions \ref{independence}, Lemma \ \ref{fact6410} and Lemma \ref{emme}.\ref{emme1} 
\berr
&&\E\lb T^{-\frac{1}{2}}\vv'\left(\db'\Hn\db\right)^{-1}\ff\rb^2\nonumber\\
 &=&
\E\left[
\frac{1}{T}\tr\db\left(\db'\Hn\db\right)^{-1}\db'\f\f'\db\left(\db'\Hn\db\right)^{-1}\db'\E\left(\v\v'\right)
\right]\nonumber\\
&\leq &
\lambda_T\left(\E\left(\v\v'\right)\right)
\E\left[
\lb(\db\left(\db'\Hn\db'\right)^{-1}\db\rb^2\tr\left(\frac{\f\f'}{T}\right)
\right]
\\
&\leq& 
 \lambda_T\left(\E\left(\v\v'\right)\right)\lambda^2_T\left(\Hninv\right)\tr\E\left(\frac{\f'\f}{T}\right)=O(1).\label{eq:130117a}
\eerr
The last bound follows noting that $\lb\db\left(\db'\Hn\db\right)^{-1}\db'\rb=\lb\left(\db'\Hn\db\right)^{-1}\rb$,  display (\ref{eq:incl}), Remarks \ref{unif lower bound} and \ref{inverseV}, Assumption \ref{ass factors}.

  To prove part $(iii)$, we first show that $T^{-1}\vv'\ssthi\vv\approx$ $T^{-1}$ $\vv'\left(\db'\Hn\db\right)^{-1}\vv$. By Lemma \ref{smw},
\ber\label{eq:130117b}
&&\lb T^{-1}\left(\vv '   \ssthi \vv - \vv '   \left(\db'\Hn\db\right)^{-1}  \vv\right)\rb\\
 &\leq & \lb \frac{\vv ' \left(\db'\Hn\db\right)^{-1}  \ff}{\sqrt{T}}\rb^2
  \lb\left(  \frac{\B^{-1} +  \ff ' \left(\db'\Hn\db\right)^{-1} \ff}{T} \right)^{-1}\rb=O_p(1).\nonumber
\eer
where the bound follow from  part $(ii)$, Assumption \ref{ass factors}, and Remark \ref{unif lower bound}.  Next we show that $T^{-1}\vv'\left(\db'\Hn\db\right)^{-1}\vv\approx$ $ T^{-1}$ $\v'\Hninv\v$. Lemma \ref{MN}  and Lemma \ref{emme}\ref{emme1},\ref{emme3} entail that
\berr
&&\E\lb T^{-1}\v'\Hninv\v-\v'\db\left(\db'\Hn\db\right)^{-1}\db'\v\rb_{sp}\\
&\leq &
\E\tr\left[ T^{-1}\v'\left(
\Hninv-\left(\md\Hn\md\right)^{+}\right)\v
\right]\\
&\leq &T^{-1}\E\tr\left(
\Hninv-\left(\md\Hn\md\right)^{+}\right)
\lambda_T\left(\E\v\v'\right)\\
&\leq& \frac{S}{T}\lambda_T\left(\Hninv\right)\lambda_T \left(\E\v\v'\right)=
O\left(\frac{1}{T}\right). 
\eerr
The desired result follows from Remark \ref{inverseV}. 
\end{proof}

\begin{proof}\textbf{of proposition \ref{lemma170117a}}
The proof of parts  $(i)$ and $(ii)$ are analogous to the proofs of Proposition \ref{lemma160117a}), parts $(i)$ and $(iii)$, respectively, and hence omitted. 
\end{proof}

\begin{proof}\textbf{of proposition \ref{prop020217a}}
Note that $\xi_i(h)=\sum_{j=1}^N r^2_{ij}\cov(a_{jt},a_{jt+h})$, therefore by Assumption \ref{ass eps} satisfies the condition of Corollary \ref{circulantN}, implying that $\tr\left(\P\Hn\P'-2\pi\bm{G}_{\xi_N}\right)^2=O(1)$.
Hence, by Lemma \ref{emme}\ref{emme1},
\ber
&&\E\lb\v'\left(\Hninv-\Hbar\right)\epi \rb^2\leq  
\E\lb\v'\Hninv\left(\Hn-\Hbari\right)\Hbar\epi \rb^2\label{eq:070217}\\
&=&\E\lb\v'\Hninv\P\left(\P'\Hn\P-2\pi\bm{G}_{\xi_N}\right)\P'\Hbar\epi \rb^2\nonumber\\
&\leq&\lambda_T\left(\P'\Hbar\E(\epi\epi'\right)\Hbar\P)
\lambda_T\left(\P'\Hninv\E(\v\v'\right)\Hninv\P)\tr\left(\P'\Hn\P-2\pi\bm{G}_{\xi_N}\right)^2\nonumber\\
&=&O(1).
\nonumber
\eer
The proof of Part (i) follows  by Markov's inequality, whereas Part (ii) follows along the lines of Lemma \ref{020217b}'s proof.
\end{proof}
%
%
\begin{proof}\textbf{of proposition \ref{propCHLOE1}}
Following \cite[Propositon 1]{RH97} and Proposition \ref{prop190117b}
\berr
\E\lb \bm{w}^{(\xii)}_{ij1}\rb^2&=&\frac{\E\left(\eta_{j0}^2\right)}{T}\sum_{u=-\infty}^{-T_0-1}
\sum_{t,s,r,q=1}^T\E(\vvv_{it}'\vvv_{is})\xii_N(t-r)\xii_N(s-q)\phi_{j,t-u}\phi_{j,s-u}\\
&=& \frac{\ka}{T}\left(
\sum_{t,s=1}^T\left|\xii_N(t-s)\right|
\right)^2\sum_{u=T_0}^{\infty}\phi_u^2\leq \ka T\sum_{u=\tau_0}^{\infty}\phi_u^2,
\eerr
which tends to zero as $T\to\infty$ for a suitable choice $\tau_0$ increasing with $T$, by Assumption \ref{ass eps}.
\end{proof}
\begin{proof}\textbf{of proposition \ref{propCHLOE2}}
Following \cite[Propositon 2]{RH97} and Proposition \ref{prop190117c} convergence in distribution in (\ref{eq:scott0B}) is implied if the convergence holds conditional on $\sa(\v)$, the sigma algebra generated by $\{\bm{v_{it}}\}_{t=1}^{T}$, and we establish the latter.
For any vector $\bm{h}$ satisfying $ \parallel   \bm{h} \parallel = 1 $ and any $T,\tau_0$, define 
${d}_{iju} = T^{-{\nicefrac{1}{2}} } \bm{h}' \left(\bm{W}^{(\xii)}_{ij0}\right) ^{-{1 \over 2}}   \bm{s}_{iju}  $.
The proof of the following results
\begin{equation}
\sum_{u=-T_0}^T  {d}_{iju}^2 \E \eta_{ju}^2 \xrightarrow{p} 1 ,   \label{eq:scott1B}
\end{equation}
and
\begin{equation}
\E \left[  \sum_{u=-T_0}^T   {d}_{iju}^2 \E\left(  \eta_{ju}^2 \bm{1}(|  d_{ju}  \eta_{ju} | > \delta_1 ) | \{  \sa(\v)  \}  \right)        \right]  \rightarrow 0  \mbox{ for all } \delta_1 >  0 .  \label{eq:scott2B}
  \end{equation}
follow closely the proof of (\ref{eq:scott1}) and (\ref{eq:scott2}). For  conciseness  we will limit ourselves to prove that the second term of (\ref{eq:secondtermR})
\be
\E \left[  \sum_{u=-T_0}^T   {d}_{iju}^2 \E\left(  \eta_{ju}^2 \bm{1}(|   \eta_{ju} | > \kappa / \delta_1  )   \right)        \right]  + \prob(  \max_u | d_{ju} | >  \delta_2 ).\label{eq:secondtermR}
\ee
converges in probability to zero, for some $\delta_2>0$.
Consider the set $\lb\left(\bm{W}^{(\xii)}_{ij0}\right) ^{-{1 \over 2}}\rb\le \ka $ on which, for any $ \kappa_2 > 0 $,
\[
\max_u  | d_{iju} | \le \ka \left(  \frac{\kappa_2 }{T} \sum_{t=1}^T  \lb \bm{\ell^{(\xii)}_{it}} \rb^2 \right)^{\frac{1}{2}}  + \ka L T^{-{\nicefrac{1}{2}}} \max_{ 1 \le t \le T }  \lb  \bm{\ell^{(\xii)}_{it}}  \rb  \max_u | \phi_{ju} |  ,
\]
for $L=L(\kappa_2)$ satisfying $ \sum_{u \ge L }  \phi_{ju}^2  < \epsilon $. The proof of (\ref{eq:scott2B}) is completed noting that, by Assumption \ref{ass V} and Proposition \ref{propCHLOE1}.(ii)
$$
\E\lb\bm{\ell^{(\bar{\xi})}_{it}}\rb^2\leq \sum_{s,u=1}^T\E\lb\vvv_{it}\rb^2\left|\hbar_N(t-s)\hbar_N(t-u)\right|\leq \ka<\infty ,
$$
and
$$
\max_{1\leq t\leq T}\lb\bm{\ell^{(\xii)}_{it}}\rb\leq \max_{1\leq t\leq T}\lb\vvv_{it}\rb\sum_{u=1}^T \left|\xii_N(h)\right|\leq \ka \left[\sum_{t=1}^T\lb \vvv_{it}\rb^4\right]^{1/4}=O_p(T^{1/4}).
$$
\end{proof}

\begin{proof}\textbf{\ref{propCHLOE3}}
 Proceeding as for the proof of Proposition \ref{prop190117c}, define
$$
  \bm{W}_j^{(\bar{\xi})}  = { \E (\eta_{j0}^2 ) \over T } \sum_{u=- \infty }^T \left(  \sum_{p,q=1}^T \sum_{t,s=1}^T  \vvv_{ip} \vvv_{iq}'\xii_N(t-p)\hbar_N(s-q)    \phi_{jt-u} \phi_{js-u}  \right) =  \E (\eta_{j0}^2 )  { \v'\Hbar   \bm{\Phi }_j  \Hbar\v   \over T } .
$$
with $\bm{\Phi_j}$ is positive definite for large $T$. 
The proof is concluded if we prove that
\berr
&& \lb T^{-1}\v'\left(\Hbar\bm{\Phi_j}\Hbar-\Hninv\bm{\Phi_j}\Hninv\right)\v\rb\\
&\leq & \|\bm{\Phi_j}\|_{sp} \lb T^{-\nicefrac{1}{2}}\left(\Hbar-\Hninv\right)\v\rb\left(\lb T^{-\nicefrac{1}{2}}\left(\Hbar-\Hninv\right)\v\rb+\lb T^{-\nicefrac{1}{2}}\Hninv\v\rb\right)=o_p(1).
\eerr
The bound follows noting that $\E\lb T^{-\nicefrac{1}{2}}\Hninv\v\rb=O(1)$ and, analogously to the Lemma \ref{prop020217a}' proof,
$$
\E\lb T^{-\nicefrac{1}{2}}\left(\Hbar-\Hninv\right)\v\rb^2\leq T^{-\nicefrac{1}{2}}\lambda_T\left(\E\bm{P}\v\v'\bm{P}'\right)\tr\left(\Hbar-\Hninv\right)^2\leq O\left(T^{-\nicefrac{1}{2}}\right).
$$

\end{proof}

\subsection{Proof of the propositions in Appendix \ref{proofT2}}\label{propth2}

%
%
%
%
%
%
%
%
%
%

%
%

\begin{proof}\textbf{of proposition \ref{23feb18}}
The results follows from Lemmata 2 (Section~\ref{lemmata2}) and the triangular inequality.
\end{proof}
\begin{proof}\textbf{of proposition \ref{lemma160117aT2}}
The proof is similar to that of Proposition \ref{lemma160117a}, and hence details are omitted. We will only check that $\MA$ satisfies the assumption of Lemma~\ref{PZ}. It will be enough to establish that the eigenvalues of $\cb$ are bounded, and uniformly bounded away from zero.
Denote by $\theta_{its}$ the $t,s$-element of $\T$.  
\be\label{eq:28febsnow2}
\theta_{its}=\tr\left(\bm{\mathscr{Q}}_i\E\left(\vvv_{it}\vvv_{is}'\right)\right)=\sum_{k=1}^K\sum_{\ell=k}^K\mathcal{q}_{ik\ell}\left[\cov(v_{ikt},v_{i\ell s})+\cov(v_{i\ell t},v_{ik s})\right]+\sum_{k=1}^K \mathcal{q}_{ikk}\cov(v_{ikt},v_{iks}),
\ee
where $\mathcal{q}_{ik\ell}$ is the $k,\ell-$element of the symmetric $(K\times K)$ matrix $\bm{\mathscr{Q}}_i$
\be\label{eq:qnormale}
\bm{\mathscr{Q}}_i:=\sxi \Gi\sf\b\b'\sf\Gi'\sxi.
\ee
Assumptions \ref{ass factors} and \ref{ass loading}, Lemma \ref{emme}\ref{emme3}\ref{emme4} and Proposition \ref{propo1}\ref{propo15M} imply that $\sup_{i,k,\ell}|\mathcal{q}_{ik\ell}|<\infty$. Assumption \ref{ass V} implies that $\sup_{it}\sum_{s=1}^T|\theta_{its}|<\infty$, and hence $\lambda_T(\T)<\infty$. Using again Lemma \ref{emme}\ref{emme4} and display (\ref{eq:28febsnow}) we find that
\ber
&& \lambda_T(\cn)\leq  \sup_i\lambda_T(\Hi)+\sup_i\lambda_T(\T)<\infty,\label{eq:luca1}\\
&& \lambda_1(\db'\cn\db)\geq  \lambda_1(\cn)\geq \lambda_1(\Hn)\geq \ka\label{eq:luca2},
\eer
proving the result.
\end{proof}
\begin{proof}\textbf{of proposition \ref{lemma170117aT2} }
Analogous to  Proposition \ref{lemma170117a}, and hence omitted.
\end{proof}
\begin{proof}\textbf{of proposition \ref{mamma} }
Recall that $\T=[\theta_{its}]$ and $\Hi=[\xi_{its}]$. For $h=t-s$, Assumptions \ref{ass eps} and \ref{ass V}, and the uniform boundness of $\mathcal{q}_{ik\ell}$ (see the proof of Lemma \ref{lemma160117aT2}) entail that $\sup_i\sum_{h=1}^{\infty}h^2\left(|\xi_i(h)|+|\theta_i(h)|\right)<\infty$.
 Hence,
\be
c_N(h)=N^{-1}\sum_{i=1}^N\left(\xi_i(h)+\theta_i(h)\right),
\ee
is even and positive definite and the conditions in Lemmas~\ref{circulant} and~\ref{paolo} and Corollary~\ref{circulantN} are satisfied.
The proof is completed along the line of Proposition \ref{prop020217a}.
\end{proof}
%
%
\begin{proof}\textbf{of proposition \ref{propCHLOE1N} }
The proof makes use of Corollary \ref{circulantN} and proceed as for Lemma \ref{propCHLOE1}
\end{proof}
\begin{proof}\textbf{of proposition \ref{propCHLOE2N} }
See the Proof of Proposition (\ref{propCHLOE2})
\end{proof}

\begin{proof}\textbf{of proposition \ref{propMarch1} } Equation (\ref{eq:27feb}) implies that 
 \berr
 &&\lb \sum_{i=1}^N r^2_{ij}\bm{W}_{ij}^{(\bar{c})}-\frac{1}{T}\v'\cni\Hi\cni\v\rb
 \leq \lb\frac{1}{T}\v'\left(\MAB-\cni\right)\Hi\cni\v\rb\\
 &&+\lb\frac{1}{T}\v'\cni\Hi\left(\MAB-\cni\right)\v\rb+
 \lb\frac{1}{T}\v'\left(\MAB-\cni\right)\Hi\left(\MAB-\cni\right)\v\rb .
 \eerr
 For sake of brevity, we only consider the first norm on the RHS of the inequality, that is bounded by
 $$
 \lb \frac{1}{\sqrt{T}} \v'\left(\MAB-\cni\right) \rb^2 \lambda^2_T(\Hi)\lambda^2_K(\cni)\lb\frac{\v}{\sqrt{T}}\rb^2=O_p\left(\frac{1}{T}\right).
 $$
The results follows from Remark~\ref{unif lower bound}, Assumption \ref{ass V}, inequality (\ref{eq:luca2}) and  noting that, using similar arguments as in the proof of Proposition
 \ref{prop020217a}
$$
\E  \lb \v'\left(\MAB-\cni\right) \rb^2=\E\lb\v'\cni\P(\P'\cn\P'-2\pi\bm{G}_{c_N}) \P'\MAB\rb^2=O(1).
$$
The remaining terms can be bounded similarly.
\end{proof}

%
%
%
%
%
%
%
%
%
%
%
\section{Auxiliary  results for proof of  Theorem~\ref{Theorem_GLS}} \label{auxiliary1}
%
%
\setcounter{equation}{0}
As mentioned at the end of Appendix \ref{auxiliary}, Lemmata I derive some properties of the matrix $\SS$ and results frequently used in the following proofs. Lemmata II and Lemmata III develop key results for the proof of Theorem \ref{Theorem_GLS} (see Proposition \ref{23feb18}). Following the discussion in Section \ref{LEMMI}, we consider the double weighted convergence of the difference $\shi-\MAI$, that by the identity in (\ref{eq:formula}), can be rewritten as
\berr
&&\shi-\MAI\\
&=& \SSI(\MA-\SS)\MAI  - \SSI   ( \sh  - \SS) \SSI+ 
\SSI  (   \sh  - \SS ) \shi    ( \sh   - \SS ) \SSI.
\eerr

\subsection{Lemmata I}\label{lemmata1}
\begin{lemma}\label{boundinverse}Under Assumptions of Theorem \ref{Theorem_GLS}
\begin{enumerate}[label=(\roman*)]
\item $\lb\bm{\Psi_i}\rb_{sp}=O_p\left(1 \right)$ and  $\lb\psinv\rb_{sp}<\infty$, $\forall i$.\label{boundinv1}
\item $\lambda_T\left(\BT\right)=O_p(1)$, $\forall i$.\label{boundinv2}
\item $\lambda_{T-K}\left(\SSI\right)<\infty . $\label{boundinv3}
\item $0<\lambda_i\left(\I_M-\Gi\psinv\Gi'\frac{\ff'\ff}{T}\right)\leq 1$ for $i=1,\dots,M$.\label{zerone}
\end{enumerate}
\end{lemma}
\begin{lemma}\label{cor}
Under Assumptions of Theorem \ref{Theorem_GLS}
\begin{enumerate}[label=(\roman*)]
\item $\lb \SSI\ff \rb=O_p\left(\frac{1}{\sqrt{T}} \right).$\label{corD1}
\item $\lb \ff'\SSI\ff \rb=O_p\left(1\right).$\label{corD2}
\item $\lb \vvj'\SSI\ff \rb=O_p\left(\frac{1}{\sqrt{T}} \right).$\label{corD3}
\item $\lb \ff\SSI\ei\rb=O_p\left(\frac{1}{\sqrt{T}} \right).$\label{corD4}
\item $\lb \vv'\SSI\ei\rb=O_p(\sqrt{T}).$\label{corD5}
\item $\lb \vvj'\SSI\vv\rb=O_p(T)$\quad \textrm{for any} $i,j$.\label{corD6}
\end{enumerate}
\end{lemma}
\begin{lemma}\label{extra1}Under Assumptions of Theorem \ref{Theorem_GLS}
\begin{enumerate}[label=(\roman*)]
\item $\lb \epi'\d \rb=O_p\left(\sqrt{T} \right).$\label{extra1a}
\item $\lb \v'\d \rb=O_p\left(\sqrt{T} \right).$\label{extra1c}
\item $\lb\ei'\xx\rb=\lb\epi'\md\x\rb=O_p\left(\sqrt{T} \right).$\label{extra1b}
\item $\lb\xx'\ff\rb=\lb\x'\md\f\rb=O_p\left(T\right).$\label{extra1d}
\end{enumerate}
\end{lemma}
\begin{lemma}\label{expansion}
Under Assumptions \ref{ass V}-\ref{independence}, for $j=1,2,\dots$ $J$,
\be\label{eq:exp1}
\left(\frac{\xx'\xx}{T}\right)^{-1}=\psinv+\sum_{j=1}^J\psinv\left[(-1)\dif\psinv\right]^j+\rem^{(J+1)},
\ee
with
\be\label{eq:exp2}
\rem_i^{(J+1)}=\left[(-1)\dif\psinv\right]^{J+1}\left(\frac{\xx'\xx}{T}\right)^{-1},
\ee
and
\be\label{eq:exp3}
\lb\rem^{(J+1)}\rb=O_p\left(T^{-\frac{(J+1)}{2}}\right). 
\ee
\end{lemma}

\begin{remark}\label{reamrk_expansion}
Lemma \ref{expansion} implies that, for $J=0$
$$
\lb\left(\frac{\xx'\xx}{T}\right)^{-1}-\psinv\rb=O_p\left(T^{-\nicefrac{1}{2}}\right).
$$
\end{remark}
\subsection{Proof Lemmata I}
\begin{proof}{\textbf{Lemma~\ref{boundinverse}}}
\begin{enumerate}[label=(\roman*)]
\item\label{boundpsi}The proof of the first part follows from Remark~\ref{inverseV} and Assumption~\ref{ass factors}. For the second part, it is enough to show that $\lambda_1\left(\ps\right)>0$ as $T\to\infty$. By Lemma \ref{emme}\ref{emme3} and Remark \ref{inverseV}
$$
\lambda_1\left(\ps\right)=
\lambda_1\left(\frac{\Gi'\ff'\ff\Gi}{T}+\sv\right)
\geq \inf_i\lambda_1\left(\sv\right)>0.
$$
\item 
%
The proof is similar that of Proposition \ref{lemma160117aT2}, and hence some details are omitted. Similarly to Equation (\ref{eq:28febsnow2}), let $\breve{\theta}_{its}=\tr\left(\bm{\breve{\mathscr{Q}}}_i\E\left(\vvv_{it}\vvv_{is}'\right)\right)$
where  
\be\label{eq:qbreve}
\bm{\breve{\mathscr{Q}}}_i:=\psinv \Gi\left(\frac{\ff'\ff}{T}\right)\b\b'\left(\frac{\ff'\ff}{T}\right)\Gi'\psinv ,
\ee
with $k,\ell-$ entry equal to $\breve{\mathcal{q}}_{ik\ell}$.
Because $\sup_{i,k,\ell}|\breve{q}_{ik\ell}|=O_p(1)$,  $\sup_{it}\sum_{s=1}^T|\breve{\theta}_{its}|=O_p(1)$, and hence $\lambda_T(\BT)=O_p(1)$. 
%
\item From Lemma \ref{emme}, part \ref{emme1} $\lambda_{T-k}(\SSI)\leq \lambda_T(\cni)\lambda_T(\md)$. Part \ref{emme3} of the same lemma implies that $\lambda_T(\cni)\leq \lambda_1^{-1}(\Hn)<\infty$, where the last inequality follows from Remark \ref{unif lower bound}.
\item To simplify the exposition, we assume  that $K\geq M$ (the case $K<M$ follows similarly).  Because
\be\label{eq:eig01}
\lambda_i\left(T^{-1}\Gi\psinv\Gi'\ff'\ff\right), \qquad i=1,\dots, M,
\ee
 are equal to $\lambda_i\left(T^{-1}\bm{\Psi_i}^{-\nicefrac{1}{2}}\Gi'\ff'\ff\Gi\bm{\Psi_i}^{-\nicefrac{1}{2}}\right)$, $i=K-M+1,\dots,K$, and hence the eigenvalues in (\ref{eq:eig01}) are non-negative. Moreover
 $$
\lambda_K\left(T^{-1}\bm{\Psi_i}^{-\nicefrac{1}{2}}\Gi'\ff'\ff\Gi\bm{\Psi}^{-\nicefrac{1}{2}}\right)<\lambda_K\left(T^{-1}\bm{\Psi_i}^{-\nicefrac{1}{2}}\Gi'\ff'\ff\Gi\bm{\Psi}^{-\nicefrac{1}{2}}+
\bm{\Psi}^{-\nicefrac{1}{2}}
\SO_{\v'\v}
\bm{\Psi}^{-\nicefrac{1}{2}}
\right)=1.
 $$
It follows that the eigenvalues in (\ref{eq:eig01}) are non-negative and smaller than one, proving the theorem.
\end{enumerate}
\end{proof}

\begin{proof}{\textbf{Lemma~\ref{cor}}}
In view of Lemma\ref{boundinverse}, Part (i) follows from Lemma \ref{PZ}(a).
The proofs of Part(ii) is similar to \ref{lemma160117a}(i), whereas that of  Part(iii) and (iv) are similar to \ref{lemma160117a}(ii), and hence omitted.
Finally, for Part (v) we have
\berr
\E\|\vv'\SSI\ei\|^2&=&\tr\left(\db\SSI\db'\E(\epi\epi')\db\SSI\db'\E(\v\v')\right)\\
&\leq&\lambda_T(\Hi)\lambda_T\left(\E\v\v'\right)\lambda_T(\SSI)\tr(\SSI)=O(T).
\eerr
\end{proof}
\begin{proof}{\textbf{Lemma~\ref{extra1}}}
For part \ref{extra1c}, by Lemma \ref{emme}, Assumptions \ref{independence} and \ref{ass factors} and Remark \ref{unif lower bound}, we have 
$$
\E\lb\ei'\xx\rb^2\leq
\left[\tr\left(\E\left(\x'\md\x\right)\E\left(\epi\epi'\right)\right)\right]
\leq
\tr\left(\E\x'\x\right)\lambda_T\left(\Hi\right)=O(T).
$$
The proofs of parts \ref{extra1a} and \ref{extra1c} use similar arguments, and hence are omitted. For the latter part,  $\lb\x'\md\f\rb\leq \lb\Gi\rb\lb\f'\md\f\rb+
\lb\v'\md\f\rb=O_p(T)+O_p(\sqrt{T})$ in view of Assumption \ref{ass factors} and noting that $\E\lb\v'\md\f\rb^2\leq \E\|\f\|^2\lambda_T(\E\v\v')=O_p(T)$.
\end{proof}
\begin{proof}{\textbf{Lemma~\ref{expansion}}}
The equality in displays (\ref{eq:exp1})-(\ref{eq:exp2}) follows from Lemma~\ref{approx inv}. Concerning the bound in (\ref{eq:exp3}), by  triangular inequality
\berr
\lb\frac{\xx'\xx}{T}-\bm{\Psi_i}\rb&\leq&\lb\frac{\v'\v}{T}-\sv\rb+\lb\frac{\v'\pd\v}{T}\rb+\lb\frac{\Gi'\ff'\vv}{T}\rb+
\lb\frac{\vv'\ff\Gi}{T}\rb .
\eerr
For the first term on the RHS of the inequality, Assumption \ref{ass V} 
imply that 
\berr
&&\E\lb\frac{\v'\v}{T}-\sv\rb^2  =  
\frac{1}{T^2}\sum_{k,h=1}^K\sum_{t,s=1}^T\cum(v_{ikt}v_{iht},v_{iks}v_{ihs})\nonumber\\
&=& \frac{1}{T^2}\sum_{k,h}\sum_{t,s} \left[\cum(v_{ikt},v_{iks}) \cum(v_{iht},v_{ihs})+ \cum(v_{ikt},v_{ihs}) \cum(v_{ikt},v_{ihs})\right]\\
&&+\frac{1}{T^2}\sum_{k,h}\sum_{t,s}\cum(v_{ikt},v_{iht},v_{iks},v_{ihs})=O\left(\frac{1}{T}\right).
\eerr
By Lemma \ref{emme}\ref{emme1}
$$
\E\lb\v'\pd\rb^2=\frac{}{}\E\tr(\pd\E\v\v')\leq \tr(\pd)\lambda_{T}\left(\E\v\v'\right)=O(1),
$$
implying that $\lb\frac{\v'\pd\v}{T}\rb\leq \frac{1}{T}\lb\v'\pd\rb^2=O_p\left(\frac{1}{T}\right) $. Similarly, for the last two terms we have
$$
\E\lb\frac{\Gi'\ff'\vv}{T}\rb^2\leq\frac{1}{T}\lb\Gi\rb^2\tr\left(\frac{\E\f'\f}{T}\right)
\lambda_T\left(\E\v\v'\right)\lambda_T(\md)=O\left(\frac{1}{T}\right),
$$
concluding the proof.
\end{proof}

\subsection{Lemmata II}\label{lemmata2}
Under the Assumption of Theorem~\ref{Theorem_GLS}, the following results holds:
\begin{lemma}\label{mmtt0} 
\begin{enumerate}[label=(\roman*)]
\item[]
\item $
\lb \frac{1}{\sqrt{T} }\xxj'\left(\shi-\SSI \right)\uuj\rb =o_p(1).
$
\item
$
\lb \frac{1}{T }\xxj'\left(\shi-\SSI \right)\xxj\rb =o_p(1).
$
\end{enumerate}

\end{lemma}
\begin{lemma}\label{mmttt}
\begin{enumerate}[label=(\roman*)]
\item[]
\item
$
\lb \frac{1}{\sqrt{T} }\xxj'\left(\SSI-\MAI \right)\uuj\rb =o_p(1).
$
\item 
$
\lb \frac{1}{T }\xxj'\left(\SSI-\MAI \right)\xxj\rb =o_p(1).
$
\end{enumerate}
\end{lemma}
%
%
\subsection{Proofs Lemmata II}
\begin{proof}{\textbf{Lemma \ref{mmtt0}}}
The proof rely on the results stated in the next section, Lemmata III.
\begin{Partn}[font=\emph]
\item  By Lemma~\ref{approx inv}(i) and the triangular inequality
\ber
&&\lb\xxj'\left(\shi-\SSI\right)\uuj\rb\nonumber\\
&\leq & \lb \xxj'\SSI(\sh-\SS)\SSI\uuj\rb\label{eq:EMME}\\
&&+
\lb \xxj'\SSI(\sh-\SS)\shi(\sh-\SS)\SSI\uuj\rb . \label{eq:EMMEa}
\eer
For the term  (\ref{eq:EMME}) we have
\ber
&&\lb \xxj'\SSI(\sh-\SS)\SSI\uuj\rb\nonumber\\
&\leq & \lb \vvj'{\SSI}(\sh-\SS)\SSI\ej\rb\label{eq:EMS2}\\
&&+\lb\vvj'\SSI(\sh-\SS)\SSI\ff\bj\rb\label{eq:EMS3}\\
&&+
 \lb \Gj'\ff'\SSI(\sh-\SS)\SSI\ej\rb\label{eq:EMS4}\\
&&+
 \lb \Gj'\ff'\SSI(\sh-\SS)\SSI\ff\bj\rb  .\label{eq:EMS5}
\eer
The term (\ref{eq:EMS2}) is bounded in Lemma \ref{mmtt1}.  Lemma \ref{mmtt2}, Lemma \ref{cor}(i), and Assumption \ref{ass loading} entail that  the term (\ref{eq:EMS3}) can be bound as 
\berr
&&\lb\vvj'\SSI(\sh-\SS)\SSI\ff\bj\rb
\leq 
\lb\vvj'\SSI(\sh-\SS)\rb\lb\SSI\ff\rb\lb\bj\rb\\
&=&\max\left\{O_p\left(1\right),O_p\left(\frac{T}{\sqrt{N}}\right)\right\}O_p\left(\frac{1}{\sqrt{T}}\right)O(1)=o_p(1).
\eerr
Similarly, but using Lemma \ref{mmtt3} instead of Lemma \ref{mmtt2}, we find that the term (\ref{eq:EMS4}) is $o_p(1)$.

Next, for the term (\ref{eq:EMS5}), Assumption \ref{ass loading}, Lemma \ref{cor}(i) and Lemma \ref{mmttF} imply
\berr
&&\lb\Gj'\ff\SSI(\sh-\SS)\SSI\ff\bj\rb
\leq 
\lb\Gj\rb \lb \ff\SSI\rb 
\lb(\sh-\SS)\SSI\ff\rb\lb\bj\rb=o_p(1).
\eerr
Hence, we conclude that the term (\ref{eq:EMME}) is $o_p(1)$.

As noted in Remark \ref{remark ass X}, Assumption \ref{ass X} implies that $\|\shi\|<\infty$. Hence, the term (\ref{eq:EMME}) is bounded by
$$
\lb \xxj'\SSI(\sh-\SS)\rb \lb\shi\rb \lb(\sh-\SS)\SSI\uuj\rb. 
$$
Taking similar steps as for the term in (\ref{eq:EMME}), Lemmas \ref{mmtt2}, \ref{mmtt3} and \ref{mmttF} entail  that (\ref{eq:EMMEa}) is of order $\max\left\{O_p(1), O_p(T/\sqrt{N}),O_p(T^2/N)\right\}$, concluding the proof.
\item The proof follows closely that of \emph{Part (i)}, and hence is omitted.
\end{Partn}
\end{proof}
\begin{proof}{\textbf{Lemma \ref{mmttt}}}
Similar to the proof of Lemma~\ref{mmtt0}
\end{proof}
%
\subsection{Lemmata III}\label{lemmata3}
%
%
\begin{lemma}\label{mmtt1}
\ber
\lb \vvj'\SSI(\sh-\SS)\SSI\ej\rb=\max\left\{O_p\left(\frac{1}{\sqrt{T}}\right),O_p\left(\frac{T}{\sqrt{N}}\right)\right\}.
\eer
\end{lemma}
%
%

%
%
%
%
\begin{lemma}\label{mmtt2}
\ber
\lb (\sh-\SS)\SSI\vvj \rb=\max\left\{O_p(1),O_p\left(\frac{T}{\sqrt{N}}\right)\right\}.
\eer
\end{lemma}
%
%
%
%
%
%
%
\begin{lemma}\label{mmtt3}
\berr
\lb (\sh-\SS)\SSI\ej \rb=\max\left\{O_p(1),O_p\left(\frac{T}{\sqrt{N}}\right)\right\}.
\eerr
\end{lemma}
\begin{lemma}\label{mmttF}
\berr
\lb (\sh-\SS)\SSI\ff \rb=\max\left\{O_p(1),O_p\left(\frac{T}{\sqrt{N}}\right)\right\}. 
\eerr
\end{lemma}

\begin{lemma}\label{mmtt4}
\berr
\lb \vv'\MAI(\SS-\MA)\SSI\ej \rb=o_p(T^{1/2}).
\eerr
\end{lemma}
\begin{lemma}\label{mmtt5}
\berr
\lb \vv'\MAI(\SS-\MA)\SSI\vvj \rb=o_p(T).
\eerr
\end{lemma}

%
%
%
%
\subsection{Proof Lemmata III}
The proofs exploits results on moments and cumulants. A summary of the main results on cumulants used in the proof are provided in Section \ref{diagram formula}, that also serves to introduce further notation.

\subsubsection{Combinatorial expressions of cumulants and moments}\label{diagram formula}
In the presentation we follows \cite{PT08}. Let $a$ denotes a finite non-empty set, and $\pp(a)$ the set of  partitions of $a$. By definition, an element $\pi\in\pp(a)$ is a collection of non-empty and disjoint subset of $a$ (called \emph{blocks}), such that their union equals $a$. To clarify the notation used in the following, if $\pi=\{a_1,\dots,a_{j}\}\in\pp(a)$, then 
$
\cap_{i=1}^j a_i= \emptyset$, $\cup_{i=1}^j a_i=a$, where $\emptyset$ is the empty set.
The symbol $|\pi|$ indicates the number of blocks (or the \emph{size}) of the partition $\pi$.

For every $\sigma,\pi\in\pp(a)$, we write $\sigma\leq\pi$ if, and only if, each block of $\sigma$ is contained in a block of $\pi$. Whenever $\sigma\leq \pi$, one has $|\sigma|\geq|\pi|$. The partition $\sigma\wedge \pi$, \emph{meet} of $\sigma,\pi\in\pp(a)$, is the element of $\pp(a)$ whose blocks are constructed by taking the non-empty  intersection of the blocks of $\sigma$ and $\pi$.
The partition $\sigma\vee \pi$, \emph{join} of $\sigma,\pi\in\pp(a)$, is the element of $\pp(a)$ whose blocks are constructed by taking the non-disjoint unions of the blocks of $\sigma$ and $\pi$, that is, by taking the union of those blocks that have at least one element in common. In particular, $|\sigma\wedge\pi|\geq |\sigma\vee\pi|$. The partitions $\sigma\wedge\pi$ and $\sigma\vee\pi$ are the greatest lower bound and the least upper bounds associated with the pair $(\sigma,\delta)$.
The \emph{maximal element} of $\pp(a)$ is the trivial partition $\hat{1}=\left\{a\right\}$. 
The \emph{minimal element} of $\pp(a)$ is the  partition $\hat{0}$, such that each block of $\hat{0}$ contains exactly one element of $a$. Observe that $|\hat{1}|=1$ and $|\hat{0}|=| a|$.
\begin{example}\label{exampcum1} Take $a=\{1,2,3,4,5\}$. If $\pi=\{\{1,2,3\},\{4,5\}\}$, $\sigma=\{\{1,2\},\{3\},\{4,5\}\}$. Then, $\sigma\leq\pi$ and $\delta\wedge \pi=\delta$,  $\delta\vee \pi=\pi$.
\end{example}
\bigskip
\noindent Let $[n]=\left\{1,2,\dots,n\right\}$. For $n\geq 1$ we consider a vector-values random variables $\bm{X^{,[n]}}=(X_1,\dots,X_n)$ such that $\E|X_j|^n\leq \infty$. For every subset $a=\left\{j_1,\dots,j_k\right\}\subseteq [n]$, we write
$$
\bm{X^{,a}}=(X_{j_1},\dots,X_{j_k}),\qquad \bm{X^a}=\prod_{i=1}^k X_{j_i}.
$$
For instance, $\forall m\leq n$,
$$
\bm{X^{,[m]}}=(X_{1},\dots,X_{m}),\qquad \bm{X^{[m]}}=\prod_{i=1}^m X_i. 
$$
The next result contains three crucial relations, linking the cumulants and the moments associated with a random vector $\bf{X^{,[m]}}$.
\begin{prop}[Proposition 3.1. \cite{PT08}]\label{pt08} For every $a\subset [n]$,
\begin{enumerate}
\item 
\be\label{eq:leo1}
\E\bm{X}^{a}=\sum_{\pi=\{a_1,\dots,a_k\}\in\pp(a)}\cum(\bm{X}^{,a_1})\cdots \cum(\bm{X}^{,a_k}),
\ee
\item 
\be\label{eq:leo2}
\cum(\bm{X}^{,a})=\sum_{\tau=\{t_1,\dots,t_r\}\in\pp(a)}(-1)^{r-1}(r-1)!\E\bm{X}^{t_1}\dots
\E\bm{X}^{t_r},
\ee
\item $\forall \pi=\{a_1,\dots,a_k\}\in \pp(a)$,
\be\label{eq:mal}
\cum(\bm{X}^{a_1},\dots,\bm{X}^{a_k})=\sum_{
\substack{ \sigma=\{p_1,\dots,p_s\}\in\pp(a)\\ \pi \vee \sigma=\hat{1}}}
\cum(\bm{X}^{,p_1})\times\cdots\times \cum(\bm{X}^{,p_s}).
\ee
\end{enumerate}
\end{prop}
Part 1 and 2 are due to \cite{LS59}, part 3 to \cite{MAL80}. 
\begin{example}\label{exampcum2} As an illustration, consider the cumulant $\cum(X_1X_2,X_3)$, in which case one has $\pi={a_1,a_2}$, with $a_1=\{1,2\}$ and $a_2=\{3\}$. There are three partition $\sigma\in\pp\left([3]\right)$ such that $\sigma\vee \pi=\hat{1}=\{1,2,3\}$, namely 
$$
\sigma_1=\{1,2,3\}=\hat{1},\quad \sigma_2=\{\{1,3\},\{2\}\}, \quad \sigma_3=\{\{2,3\},\{1\}\}.
$$
It follows that
\be\label{eq:cot}
\cum(X_1X_2,X_3)=\cum(X_1,X_2,X_3)+\cum(X_1,X_3)\cum(X_2)+\cum(X_2,X_3)\cum(X_1).
\ee
Similarly
\be\label{eq:cotV}
\E(X_1X_2X_3)=\cum(X_1X_2,X_3)+\cum(X_1,X_2)\E(X_3)+\E(X_1)\E(X_2)\E(X_3).
\ee
\end{example}

\bigskip
%
%
%
%
\noindent Part of the notions presented in the above sections can be translated into the language of diagrams, which are often used to compute higher order cumulants.

Consider a finite set $a$. A \emph{diagram} is a graphical representation of a pair of partitions $(\pi,\sigma)\subseteq\pp(a)$ such that $\pi=\{a_1,\dots,a_k\}$ and $\sigma=\{p_1,\dots,p_s\}$. It is obtained as follows
\begin{enumerate}[wide, labelwidth=!, labelindent=0pt]
\item Order the elements of each block $a_i$, for $i=1,\dots,k;$
\item Associate with each block $a_i\in\pi$ a row of each $| a_i |$ \emph{vertices} (represented as dots), in such a way that the $j$th vertex of the $i$th row corresponds to the $j$th element of the block $a_i$; 
\item For every $\ell=1,\dots,s$, draw a closed curve around the vertices corresponding to the elements of the block $p_l\in\sigma$. We will denote by $\di(\pi,\sigma)$ the diagram of a pair of partitions $(\pi,\sigma).$ 
\end{enumerate}

The diagram $\di(\pi,\sigma)$ associated with two partitions $(\pi,\sigma)$ is said to be $\emph{connected}$ is $\delta\vee\pi=\hat{1}$, that is if the only partition $\rho$ such that $\pi\leq \rho$ and $\sigma\leq \rho$ is the maximal partition $\hat{1}$. In other words,  $\di(\pi,\sigma)$ is connected if and only if the rows of the diagram (the blocks of $\pi$), cannot be divided into two subsets, each defined a separate diagram.

The Malyshev formula (\ref{eq:mal}) can be reformulated as follows. For every finite set $a$ and every $\pi=\{a_1,\dots,a_k\}\in \pp(a)$
\be\label{eq:df}
\cum(\bm{X^{a_1}},\dots,\bm{X^{a_k}})=\sum_{
\sigma=\{p_1,\dots,p_s\}\in\dic(\pi,\sigma)}
\cum(\bm{X^{,p_1}})\times\cdots\times \cum(\bm{X^{,p_s}}),
\ee
where $\dic(\pi,\sigma)$ denotes the class of connected diagrams for the pair of partitions $(\pi,\omega)\subset \pp(a)$, that is
$$
\dic(\pi,\sigma)=\left\{\sigma=\{p_1,\dots,p_s\}\in\pp(a)\; : \;  \di(\pi,\sigma)\;\textrm{ connected}\right\}.
$$
\begin{example}\label{Fdiagramma} Consider again the Example \ref{exampcum1}. Recall that in this case $a=[3]$ and that the relevant partition is $\pi=\{\{1,2\},\{3\}\}$. There are only three partitions $\sigma_1,\sigma_2,\sigma_3\in\pp([3])$ such that $\di(\pi,\sigma_1)$,
$\di(\pi,\sigma_2)$, and $\di(\pi,\sigma_3)$ are connected, namely 
$$
\sigma_1=\hat{1},\qquad \sigma_2=\{\{1,3\},\{2\}\},\qquad \sigma_3=\{\{1\},\{2,3\}\}.
$$ 
The cumulants $\cum(X_1,X_2,X_3)$,$\cum(X_1,X_3)\cum_1(X_2)$, and $\cum(X_1)\cum(X_2,X_3)$ are associated with the diagrams 
$\di(\pi,\sigma_1)$, 
$\di(\pi,\sigma_2)$ and  $\di(\pi,\sigma_3)$, respectively. Taking the sums over the partitions associated with the connected diagram, the formula (\ref{eq:cot}) is obtained.
\end{example}

%
%
%
%
%

%
%
%
%
\subsubsection{Proof Lemma \ref{mmtt1}}
Recalling the definitions in Appendix \ref{auxiliary}; for $j=1,2,3$ let 
$
\bm{\hat{C}_{jN}}:=N^{-1}\sum_{i=1}^N\bm{\hat{C}_{ji}}$, with $\db'\bm{\hat{C}}_{ji}\db=\bm{\hat{\mathcal{C}}}_{ji}$, and $\BTN=N^{-1}\sum_{i=1}^N\BT$.

To save notation, define the $T\times T$ matrix $\SSIb=[\oi_{ts}]$, with
\be
\SSIb:=\db\SSI\db'.
\ee
From the bounds in displays (\ref{eq:boundMARCH6B}), (\ref{eq:boundMARCH6A}), (\ref{eq:BOUNDT1c}) and (\ref{eq:BOUNDT1d}) we have
\berr
&&\lb \vvj'\SSI(\sh-\SS)\SSI\ej \rb\\
&\leq & \overbrace{\lb \vvj'\SSI\ff\left(\han-\an\right)\ff'\SSI\ej \rb}^{Term \;I}+    
\overbrace{\lb \vj'\SSIb\left(\bm{\hat{C}_{1N}}-\Hn\right)\SSIb\epj \rb }^{Term \;II}\\ 
&&+\underbrace{\lb \vj'\SSIb\left(\bm{\hat{C}_{2N}}-\BTN\right)\SSIb\epj \rb}_{Term \;III} 
+\underbrace{\lb \vj'\SSIb\left(\bm{\hat{C}_{3N}}+\bm{\hat{C}'_{3N}}\right)\SSIb\epj \rb }_{Term \;IV} \\
&=& O_p\left(\frac{1}{T}\right)+\max\left(O_p(1),O_p\left(\frac{T}{\sqrt{N}}\right)\right)+ \max\left(O_p(1),O_p\left(\frac{T^{\nicefrac{3}{2}}}{N}\right)\right).
\eerr
%
%
%
%
The proof of most of the terms is quite tedious. For the terms that cannot be easily bounded using matrix norm inequality, we apply Markov inequality and bound the expectation. It is useful to note that the conditional expectation  (with respect $\sa(\z)$) of $\SSIb$ and $\psinv$ have uniformly bounded entries, so they can be treated as bounded constants. The save space, the conditioning will be not made explicit, because of Assumption~\ref{independence}. For illustrative purposed, a quite detailed proof for $Part (a)$ of Term II will be provided. More concise proof will be provided for the remaining parts and terms. Below, we report as remarks few facts, following from the combination of the Assumptions and previous lemmas, that will be repeatedly used later. 
\begin{remark}\label{FHILL}
Property (4) in \cite[Section 5.2]{ltk} and Lemma~\ref{boundinverse}\ref{boundinv3} imply that $\lambda_{T-K}(\SSIb)=\lambda_{T-K}(\SSI)$. By Lemma \ref{emme}\ref{emme4}, $\sup_{t,s}|\oi_{ts}|<\infty$ and Lemma~\ref{matrix2} implies that $\sup_t\sum_s|\oi_{ts}|^2<\infty$. Hence, by Lemma \ref{matrix},  Remark \ref{remcum} and Assumptions \ref{ass V}  imply that, 
$$
\sup_{p,q}\left|\sum_{t,s}\oi_{pt}\cov\left(v_{ikt},v_{jhs}\right)\oi_{sq}\right|<\infty,\;\textrm{and}\;
\sup_j\sup_{p,q}\left|\sum_{i}\sum_{t,s}\oi_{pt}\cov\left(v_{ikt},v_{jhs}\right)\oi_{sq}\right|<\infty.
$$
Similarly, $\sup_p\left|\sum_{t,s}\oi_{pt}\cov\left(v_{ikt},v_{jhs}\right)\right|$ and 
$\sup_{p}\sup_j\left|\sum_i\sum_{t,s}|\oi_{pt}\cov\left(\varepsilon_{it},\varepsilon_{js}\right)\right|$ are both 
$O(\sqrt{T})$.

The same bounds can be obtained replacing $\cov(v_{ikt},v_{jhs})$ with $\cov(\varepsilon_{it},\varepsilon_{js})$.
\end{remark}
\begin{remark}\label{FHILL2}
Let $\si_{ikh}$ denote the $k,h$-entry of the $K\times K$ matrix $\psinv$. Lemma~\ref{boundinverse}\ref{boundinv1}, Lemma \ref{emme}\ref{emme4} and Remark~\ref{inverseV} imply that $\sup_{i}|\si_{ikh}|<\infty$. 
\end{remark}
\subsubsection*{Bound Term I}
Displays (\ref{eq:radio106}) and (\ref{eq:radio107}), Assumption \ref{ass loading} and Remark \ref{reamrk_expansion}, entail that 
\berr
&&\lb\left(\I_M-\Gi\xx^+\ff\b\right)-\left(\I_M-\Gi\psinv\Gi'\frac{\ff'\ff}{T}\b\right)\rb\\
&=&\lb 
\Gi\psinv\Gi'\frac{\ff'\ff}{T}\b-
\Gi\left(\frac{\xx'\xx}{T}\right)^{-1}\left(\Gi'\frac{\ff'\ff}{T}+\frac{\vv'\ff}{T}\right)\b 
\rb\\
&\leq &
\lb 
\Gi\left[\psinv-\left(\frac{\xx'\xx}{T}\right)^{-1}\right]\Gi'\frac{\ff'\ff}{T}\b
\rb
+ \lb
\Gi\left(\frac{\xx'\xx}{T}\right)^{-1}\frac{\vv'\ff}{T}\b 
 \rb=O_p\left(\frac{1}{\sqrt{T}}\right).
\eerr
Hence,
$
\lb\han-\an\rb
=O_p(T^{-\nicefrac{1}{2}})$. Using Lemma \ref{cor}\ref{corD3},\ref{corD4} we conclude that 
\ber
&&\lb \vvj'\SSI\ff\left(\han-\an\right)\ff\SSI\ej \rb\nonumber\\
&\leq &  \lb \vvj'\SSI\ff\rb \lb\han-\an\rb
\lb\ff\SSI\ej \rb
=O_p(T^{-\nicefrac{3}{2}}). \label{eq:boundMARCH6B}
\eer
%
%
%
%
%
%
%
%
%
%
\subsubsection*{Bound Term II}
We split the proof of the latter equality in four parts (Parts (a)-(d)), one for each term on the RHS of the  inequality below.
\ber
&& \lb \vj'\SSIb\left(\bm{\hat{C}_{1N}}-\Hn\right)\SSIb\epj\rb\nonumber\\
 &\leq &
\lb\frac{1}{N}\sum_{i=1}^N\vj'\SSIb\left[\epi\epi'-\Hi\right]\SSIb\epj\rb\label{eq:boundIIa}\\
&+&
\lb\frac{1}{TN}\sum_{i=1}^N\vj'\SSIb\epi\epi'\md\x\left(\frac{\xx'\xx}{T}\right)^{-1}\v'\SSIb\epj\rb\label{eq:boundIIb}\\
&+& \lb\frac{1}{TN}\sum_{i=1}^N\vj'\SSIb\v\left(\frac{\xx'\xx}{T}\right)^{-1}\x'\md\epi\epi'\SSIb\epj\rb\label{eq:boundIIc}\\
&+& 
 \lb\frac{1}{T^2N}\sum_{i=1}^N\vj'\SSIb
\v\left(\frac{\xx'\xx}{T}\right)^{-1}\x'\md\epi\epi'
\md\x\right.\\
&&\times\left.\left(\frac{\xx'\xx}{T}\right)^{-1}\v'\SSIb\epj\rb\nonumber\\
&=&O_p\left(\frac{T}{\sqrt{N}}\right)+\max\left(O_p(1),O_p\left(\frac{T}{\sqrt{N}}\right)\right)+\max\left(O_p(1),O_p\left(\frac{T^{\nicefrac{3}{2}}}{N}\right)\right)\nonumber\\
&&+\max\left(O_p(1),O_p\left(\frac{T}{N}\right)\right)\\
&=&\max\left(O_p(1),O_p\left(\frac{T}{\sqrt{N}}\right)\right). \label{eq:boundMARCH6A}
\nonumber
\eer

\begin{Part}[font=\emph]
%
%
\item  To bound the term (\ref{eq:boundIIa}) show that
\ber
&&\E\lb \vj'\SSIb\left(\frac{1}{N}\sum_{i=1}^N\bm{\varepsilon_i\varepsilon_i'} -\Hi\right)
\SSIb\bm{\varepsilon_j}
\rb^2\nonumber\\
&=& \frac{1}{N^2}\E\left[\sum_k\sum_{t_1,\dots,t_4}\sum_{s_1,\dots,s_4}\sum_{i,\ell}\oi_{t_1t_2}\oi_{s_1s_2}\oi_{t_3t_4}\oi_{s_3s_4}\E(v_{kjt_1}v_{kjs_1})\right.\nonumber\\
&&
\left.\E\left(\varepsilon_{it_2}\varepsilon_{it_3}-\xi_{it_2t_3})
(\varepsilon_{\ell s_2}\varepsilon_{\ell s_3}-\xi_{\ell s_2s_3}\right)\left(\varepsilon_{jt_4}\varepsilon_{js_4}\right)\right]\nonumber\\
&=& \E\left[\frac{1}{N^2}\sum_k\sum_{t_1,\dots,t_4}\sum_{s_1,\dots,s_4}\sum_{i,\ell}\oi_{t_1t_2}\oi_{s_1s_2}\oi_{t_3t_4}\oi_{s_3s_4}\cum(v_{kjt_1},v_{kjs_1})\right.\label{eq:LA1unoa}\\
&&
\times\cum\left(\varepsilon_{it_2}\varepsilon_{it_3},
\varepsilon_{\ell s_2}\varepsilon_{\ell s_3}\right)
\cum\left(\varepsilon_{jt_4},\varepsilon_{js_4}\right)\label{eq:LA1unoab}
\\
&&+ \frac{1}{N^2}\sum_k\sum_{t_1,\dots,t_4}\sum_{s_1,\dots,s_4}\sum_{i,\ell}\oi_{t_1t_2}\oi_{s_1s_2}\oi_{t_3t_4}\oi_{s_3s_4}\cov(v_{kjt_1},v_{kjs_1}).\label{eq:LA1unob}\\
&&\left.
\times \cum\left(\varepsilon_{it_2}\varepsilon_{it_3},
\varepsilon_{\ell s_2}\varepsilon_{\ell s_3},
\varepsilon_{jt_4}\varepsilon_{js_4}\right)\right]. \label{eq:LA1unobb}\\
&=& O\left(\frac{T^2}{N}\right)\label{eq:DUKE790a}
\eer
The bound follows from (\ref{eq:dcd}) and (\ref{eq:LA1unobBOUND}) below. 

Consider the first cumulant in (\ref{eq:LA1unoab}). Using the notation introduced in Section \ref{diagram formula}
$$
\pi=\left\{
\left\{it_2,it_3 \right\}
\left\{\ell s_2,\ell s_3 \right\}
\right\}
$$
be the partition defining the vertices of the diagram. Then,
\berr
\sigma_1&=&\left\{\left\{it_2,it_3\right\},\left\{\ell s_2,\ell s_3\right\}\right\} ,\qquad
\sigma_2=\left\{\left\{it_2,\ell s_2  \right\},\left\{it_3,\ell s_3  \right\}\right\},\\
\sigma_3&=&\left\{\left\{it_2,\ell s_3\right\},\left\{i t_3,\ell s_2\right\}\right\} ,\qquad
\sigma_4=\left\{it_2,it_3,\ell s_2,\ell s_3  \right\}=\hat{1},
\eerr
and 
$$
\dic(\pi,\sigma)=\{\di(\pi,\sigma_2), \di(\pi,\sigma_3), \di(\pi,\sigma_4)\}
$$
For $l=1,2$, we have
\berr
&&\cum\left(\varepsilon_{it_2}\varepsilon_{it_3},
\varepsilon_{\ell s_2}\varepsilon_{\ell s_3}\right)=\sum_{\sigma=\{p_1,\dots,p_l\}\in\dic(\pi,\sigma)}\cum(\bm{\varepsilon^{,p_1}})\times\cdots\times \cum\left(\bm{\varepsilon^{,p_l}}\right)\nonumber\\
&=& \cov\left(
\varepsilon_{it_2},\varepsilon_{\ell s_2}
\right)
\cov\left(
\varepsilon_{it_3},\varepsilon_{\ell s_3}
\right)+
\cov\left(\varepsilon_{it_2},
\varepsilon_{\ell s_3}\right)
\cov\left(
\varepsilon_{it_3},\varepsilon_{\ell s_2}
\right)
+
\cum(\varepsilon_{it_2},\varepsilon_{it_3},\varepsilon_{js_2},\varepsilon_{js_3}).
\eerr
By  Remark \ref{FHILL}, the term in  (\ref{eq:LA1unoa})-(\ref{eq:LA1unoab}) can be  bounded as follows
\ber
&=&\frac{1}{N^2}\sum_k\sum_{t_1,\dots,t_4}\sum_{s_1,\dots,s_4}\sum_{i,\ell}
\oi_{t_1t_2}\oi_{s_1s_2}\oi_{t_3t_4}\oi_{s_3s_4}\nonumber\\
&&\times\cov(v_{k_1jt_1},v_{k_1js_1})\cov\left(\varepsilon_{jt_4},\varepsilon_{js_4}\right)
\cov\left(\varepsilon_{it_2}\varepsilon_{it_3},
\varepsilon_{\ell s_2}\varepsilon_{\ell s_3}\right)
\nonumber
\\
&\leq &
\frac{C}{N^2}
\left(\sup_{t_2,s_2}\left|\sum_{t_1,s_1} \oi_{t_1,t_2}\cov(v_{kjt_1},v_{kjs_1})\oi_{s_1s_2}\right|\right)\nonumber\\
%
&&\times\left(\sup_{i,\ell}\sup_{t_3,s_3}\left|\sum_{t_4,s_4} \oi_{t_3,t_4}\cov(\varepsilon_{j t_4},\varepsilon_{\ell j_4})\oi_{s_3,s_4}\right|\right)\nonumber\\
%
%
%
&&\times \left[\left(
\sup_{i,\ell}\sum_{t_2,s_2}|\cov(\varepsilon_{it_2},\varepsilon_{\ell s_2})|\right)
\left(\sum_{i,\ell}\sum_{t_3,s_3}|\cov(\varepsilon_{it_3},\varepsilon_{\ell s_3})|\right)\right.\nonumber\\
&&\phantom{c}
+\left(\sup_{i,\ell}\sum_{t_2,s_3}|\cov(\varepsilon_{it_2},\varepsilon_{\ell s_3})|\right)
\left(\sum_{i,\ell}\sum_{t_3,s_2}|\cov(\varepsilon_{it_3},\varepsilon_{\ell s_2})|\right)\nonumber\\
&&\phantom{c}
\left.+\left(\sum_{i,\ell}\sum_{t_2t_3s_2s_3}|\cum(\varepsilon_{it_2},\varepsilon_{it_3},\varepsilon_{js_2},\varepsilon_{js_3})|\right)\right]\nonumber\\
&=&O\left(\frac{C}{N^2}\left(T^2N+T^2N+TN\right)\right)=O\left(\frac{T^2}{N}\right)\label{eq:dcd}.
\eer

Next, we consider the cumulant in (\ref{eq:LA1unobb}):
\be\label{eq:diagcumR}
\cum\left(\varepsilon_{it_2}\varepsilon_{it_3},
\varepsilon_{\ell s_2}\varepsilon_{\ell s_3},\varepsilon_{j t_3}\varepsilon_{j s_3}\right)=\sum_{\sigma=\{p_1,\dots,p_l\}\in\dic(\pi,\sigma)}\cum(\bm{\varepsilon^{,p_1}})\times\cdots\times \bm{\varepsilon^{,p_l}}),
\ee
for $l=1,2,3.$
We can distinguish four different sets of diagrams, defined by the number  of blocks in the partition ($l$) and the size of the blocks. We label these sets as $\sigma^{(a)},\dots,\sigma^{(d)}$. For example, the first partition $\sigma^{(a)}=\left\{p_1,p_2,p_3: |p_1|=|p_2|=|p_3|=2\right\}$ includes all the partitions made of three blocks, $p_1,p_2$ and $p_3$, and each block is of size two ($|p_i|=2$, for $i=1,2,3$). In a diagram, we would have three rows, each with two vertices.
\begin{itemize}
\item[$\sigma^{(a)}$] $:=\left\{p_1,p_2,p_3: |p_1|=|p_2|=|p_3|=2\right\}$
\berr
&&\sum_{\sigma=\{p_1,p_2,p_3\}\in\dic(\pi,\sigma^{(a)})}\cum\left(\bm{\varepsilon^{,p_1}}\right)\cum\left(\bm{\varepsilon^{,p_2}}\right)\cum\left(\bm{\varepsilon^{,p_3}}\right)
\nonumber
\\
&=&\cov(\varepsilon_{it_2},\varepsilon_{\ell s_2})
\cov(\varepsilon_{it_3},\varepsilon_{jt_4})
\cov(\varepsilon_{\ell s_3},\varepsilon_{js_4})
+\cov(\varepsilon_{it_2},\varepsilon_{\ell s_2})
\cov(\varepsilon_{\ell s_3},\varepsilon_{jt_4})
\cov(\varepsilon_{it_3},\varepsilon_{js_4})\\
&&
+\cov(\varepsilon_{it_2},\varepsilon_{\ell s_3})
\cov(\varepsilon_{\ell s_2},\varepsilon_{js_2})
\cov(\varepsilon_{it_3},\varepsilon_{jt_4})
+\cov(\varepsilon_{it_2},\varepsilon_{\ell s_3})
\cov(\varepsilon_{\ell s_2},\varepsilon_{jt_4})
\cov(\varepsilon_{it_3},\varepsilon_{js_4})\\
&&
+\cov(\varepsilon_{it_2},\varepsilon_{jt_4})
\cov(\varepsilon_{it_3},\varepsilon_{\ell s_2})
\cov(\varepsilon_{\ell s_3},\varepsilon_{js_4})
+\cov(\varepsilon_{it_2},\varepsilon_{jt_4})
\cov(\varepsilon_{\ell s_2},\varepsilon_{js_4})
\cov(\varepsilon_{it_3},\varepsilon_{\ell s_3})\\
&&+\cov(\varepsilon_{it_2},\varepsilon_{js_4})
\cov(\varepsilon_{it_3},\varepsilon_{\ell s_2})
\cov(\varepsilon_{\ell s_3},\varepsilon_{jt_4})
+\cov(\varepsilon_{it_2},\varepsilon_{js_4})
\cov(\varepsilon_{\ell s_2},\varepsilon_{jt_4})
\cov(\varepsilon_{it_3},\varepsilon_{\ell s_3})\\
&&
+\cov(\varepsilon_{it_3},\varepsilon_{\ell s_3})
\cov(\varepsilon_{\ell t_2},\varepsilon_{jt_4})
\cov(\varepsilon_{\ell s_2},\varepsilon_{j s_4})
+\cov(\varepsilon_{it_3},\varepsilon_{\ell s_3})
\cov(\varepsilon_{it_2},\varepsilon_{js_4})
\cov(\varepsilon_{\ell s_2},\varepsilon_{jt_4}).
\eerr
%
%
\item[ $\sigma^{(b)}$] $:=\left\{p_1,p_2: |p_1|=2, |p_2|=4\right\}$
\berr
&&\sum_{\sigma=\{p_1,p_2\}\in\dic(\pi,\sigma^{(b)})}\cum\left(\bm{\varepsilon^{,p_1}}\right)\cum\left(\bm{\varepsilon^{,p_2}}\right)
\nonumber\\
&
=&\cov(\varepsilon_{it_2},\varepsilon_{js_4})
\cum(\varepsilon_{\ell s_2},\varepsilon_{jt_4},
\varepsilon_{it_3},\varepsilon_{\ell s_3})
+\cov(\varepsilon_{it_3},\varepsilon_{\ell s_3})
\cum(\varepsilon_{it_2},\varepsilon_{\ell s_2},\varepsilon_{jt_4},\varepsilon_{js_4})\\
&&
+\cov(\varepsilon_{it_3},\varepsilon_{\ell s_2})
\cum(\varepsilon_{it_2},\varepsilon_{jt_4},\varepsilon_{\ell s_3},\varepsilon_{js_4})
+\cov(\varepsilon_{it_2},\varepsilon_{\ell s_3})
\cum(\varepsilon_{it_3},\varepsilon_{\ell s_2},\varepsilon_{jt_4},\varepsilon_{js_4})\\
&&
+\cov(\varepsilon_{it_2},\varepsilon_{\ell s_2})
\cum(\varepsilon_{it_3},\varepsilon_{\ell s_3},\varepsilon_{jt_4},\varepsilon_{js_4})
+\cov(\varepsilon_{it_3},\varepsilon_{js_4})
\cum(\varepsilon_{it_2},\varepsilon_{\ell s_2},\varepsilon_{\ell s_3},\varepsilon_{js_4})\\
&&
+\cov(\varepsilon_{it_3},\varepsilon_{jt_4})
\cum(\varepsilon_{it_2},\varepsilon_{\ell s_2},\varepsilon_{\ell s_3},\varepsilon_{js_4})
+\cov(\varepsilon_{\ell s_2},\varepsilon_{jt_4})
\cum(\varepsilon_{it_2},\varepsilon_{it_3},\varepsilon_{\ell s_3},\varepsilon_{js_4})\\
&&
+\cov(\varepsilon_{\ell s_2},\varepsilon_{js_4})
\cum(\varepsilon_{it_2},\varepsilon_{it_3},\varepsilon_{\ell s_3},\varepsilon_{jt_4})
+\cov(\varepsilon_{\ell s_3},\varepsilon_{jt_4})
\cum(\varepsilon_{it_2},\varepsilon_{it_3},\varepsilon_{\ell s_2},\varepsilon_{js_4})\\
&&
+
\cov(\varepsilon_{\ell s_3},\varepsilon_{js_4})
\cum(\varepsilon_{it_2},\varepsilon_{it_3},\varepsilon_{\ell s_2},\varepsilon_{jt_4})
+\cov(\varepsilon_{it_2},\varepsilon_{jt_4})
\cum(\varepsilon_{it_3},\varepsilon_{\ell s_2},\varepsilon_{\ell s_3},\varepsilon_{js_4}). 
\eerr
%
%
%
%
%
\item[ $\sigma^{(c)}$] $:=\left\{p_1,p_2: |p_1|=3, |p_2|=3\right\}$
\berr
&&\sum_{\sigma=\{p_1,p_2\}\in\dic(\pi,\sigma^{(c)})}\cum\left(\bm{\varepsilon^{,p_1}}\right)\cum\left(\bm{\varepsilon^{,p_2}}\right)
\nonumber\\
&=&\cum(\varepsilon_{it_2},\varepsilon_{\ell s_2},\varepsilon_{jt_4})
\cum(\varepsilon_{it_3},\varepsilon_{\ell s_3},\varepsilon_{js_4})
+ \cum(\varepsilon_{it_2},\varepsilon_{it_3},\varepsilon_{\ell s_2})
\cum(\varepsilon_{\ell s_3},\varepsilon_{jt_4},\varepsilon_{j s_4})\\
&& +\cum(\varepsilon_{it_2},\varepsilon_{\ell s_2},\varepsilon_{\ell s_3})
\cum(\varepsilon_{it_3},\varepsilon_{jt_4},\varepsilon_{js_4})
+\cum(\varepsilon_{it_2},\varepsilon_{\ell s_2},\varepsilon_{js_4})
\cum(\varepsilon_{it_3},\varepsilon_{\ell s_3},\varepsilon_{jt_4})\\
&& +\cum(\varepsilon_{it_2},\varepsilon_{it_3},\varepsilon_{\ell s_3})
\cum(\varepsilon_{\ell s_2},\varepsilon_{jt_4},\varepsilon_{js_4})
+\cum(\varepsilon_{it_2},\varepsilon_{it_3},\varepsilon_{js_4})
\cum(\varepsilon_{\ell s_2},\varepsilon_{\ell s_3},\varepsilon_{jt_4})\\
&& +\cum(\varepsilon_{it_2},\varepsilon_{\ell s_3},\varepsilon_{j s_4})
\cum(\varepsilon_{i t_3},\varepsilon_{\ell s_2},\varepsilon_{jt_4})
+\cum(\varepsilon_{it_2},\varepsilon_{jt_4},\varepsilon_{js_4})
\cum(\varepsilon_{it_3},\varepsilon_{\ell s_2},\varepsilon_{\ell s_3})\\
&& +\cum(\varepsilon_{it_2},\varepsilon_{it_3},\varepsilon_{jt_4})
\cum(\varepsilon_{\ell s_2},\varepsilon_{\ell s_3},\varepsilon_{js_4})
 +
\cum(\varepsilon_{it_2},\varepsilon_{\ell s_3},\varepsilon_{jt_4})
\cum(\varepsilon_{it_3},\varepsilon_{\ell s_2},\varepsilon_{js_4}).
\eerr
\item[ $\sigma^{(d)}$] $:=\left\{p_1, |p_1|=6\right\}$
\berr
&&\sum_{\sigma=\{p_1\}\in\dic(\pi,\sigma^{(d)})}\cum\left(\bm{\varepsilon^{,p_1}}\right)=\cum\left(\varepsilon_{it_2},\varepsilon_{it_3},\varepsilon_{it_3},
\varepsilon_{\ell s_2},\varepsilon_{\ell s_3},\varepsilon_{jt_3},\varepsilon_{js_3}\right).
\eerr
\end{itemize}
Using the equality in (\ref{eq:diagcumR}), for the terms in (\ref{eq:LA1unob})-(\ref{eq:LA1unobb}) we find
\ber
&&\frac{1}{N^2}\sum_k\sum_{t_1,\dots,t_4}\sum_{s_1,\dots,s_4}\sum_{i,\ell}\oi_{t_1t_2}\oi_{s_1s_2}\oi_{t_3t_4}\oi_{s_3s_4}\cov(v_{kjt_1},v_{kjs_1})\nonumber\\
&&\times
\cum\left(\varepsilon_{it_2}\varepsilon_{it_3},
\varepsilon_{\ell s_2}\varepsilon_{\ell s_3},
\varepsilon_{jt_4}\varepsilon_{js_4}\right)\nonumber\\
&=&\frac{1}{N^2}\sum_k\sum_{t_1,\dots,t_4}\sum_{s_1,\dots,s_4}\sum_{i,\ell}\oi_{t_1t_2}\oi_{s_1s_2}\oi_{t_3t_4}\oi_{s_3s_4}\cov(v_{kjt_1},v_{kjs_1})\nonumber\\
&&\times\left[
\sum_{\sigma=\{p_1,p_2,p_3\}\in\dic(\pi,\sigma^{(a)})}\cum\left(\bm{\varepsilon^{,p_1}}\right)\cum\left(\bm{\varepsilon^{,p_2}}\right)\cum\left(\bm{\varepsilon^{,p_3}}\right)
\right.\nonumber\\
&&+
\sum_{\sigma=\{p_1,p_2\}\in\dic(\pi,\sigma^{(b)})}\cum\left(\bm{\varepsilon^{,p_1}}\right)\cum\left(\bm{\varepsilon^{,p_2}}\right)
\nonumber\\
&&+
\sum_{\sigma=\{p_1,p_2\}\in\dic(\pi,\sigma^{(c)})}\cum\left(\bm{\varepsilon^{,p_1}}\right)\cum\left(\bm{\varepsilon^{,p_2}}\right)
\nonumber\\
&&\left.+\cum\left(\varepsilon_{it_2},\varepsilon_{it_3},\varepsilon_{it_3},
\varepsilon_{\ell s_2},\varepsilon_{\ell s_3},\varepsilon_{jt_3},\varepsilon_{js_3}\right)\right]\nonumber\\
&=& O\left(\frac{T^3}{N^2}\right)+O\left(\frac{T^2}{N^2}\right)+O\left(\frac{T}{N^2}\right). \label{eq:LA1unobBOUND}
\eer
The latter bound results from the bounds in displays (\ref{eq:partAA})-(\ref{eq:partAD}), which are in turn derived for  each of the four sets of diagrams, $\sigma^{(a)},\dots,\sigma^{(d)},$. To save space, for each set, computations are shown for only one diagram. The remaining terms can be shown to be at most of the same order, using similar arguments.

Consider first $\sigma_1^{(a)}\in\sigma^{(a)}$
$$
\sigma^{(a)}_1=\left\{p_1,p_2,p_3\right\}=\left\{
\{it_2,\ell s_2\}, \{it_3,j t_4\}, \{\ell s_3,j s_4\}
\right\}.
$$
Then, by Remark \ref{FHILL}
\ber
&&N^{-2}\sum_{i,j}\sum_{t_1,\dots,t_4}\sum_{s_1,\dots,s_4}\oi_{t_1t_2}\oi_{s_1s_2}\oi_{t_3t_4}\oi_{s_3s_4}\cov(v_{kjt_1},v_{kjs_1})\nonumber\\
&&\times
\cov(\varepsilon_{it_2},\varepsilon_{\ell s_2})
\cov(\varepsilon_{it_3},\varepsilon_{jt_4})
\cov(\varepsilon_{\ell s_3},\varepsilon_{js_4})\nonumber\\
&\leq &
N^{-2}\sum_k\sup_{t_2,s_2}\left|\sum_{t_1,s_1}|\oi_{t_1t_2}\cov(v_{kjt_1},v_{kjs_1})\oi_{s_1s_2}\right|\sup_{t_3,t_4}|\oi_{t_3t_4}|\sup_{s_3,s_4}|\oi_{s_3s_4}|\nonumber\\
&&\times\left(\sup_{i,\ell}\sum_{t_2,s_2}|\cov(\varepsilon_{it_2},\varepsilon_{\ell s_2})|\right)
\left(\sum_i\sum_{t_3,t_4}|\cov(\varepsilon_{it_3},\varepsilon_{jt_4})|\right)\nonumber\\
&&\times\left(\sum_{\ell}\sum_{s_3,s_4}|\cov(\varepsilon_{\ell s_3},\varepsilon_{js_4})|\right)\nonumber\\
&=&O\left(T^3/N^2\right). \label{eq:partAA}
\eer
Similarly, for 
$$
\sigma^{(b)}_1=\left\{p_1,p_2\right\}=\left\{
\{it_1,\ell t_5\},\{it_3,\ell s_3, jt_4,js_4  \}
\right\}
$$ 
\ber
&&N^{-2}\sum_{i,j}\sum_{t_2,t_3,t_4}\sum_{s_2,s_3,s_4}\oi_{t_1t_2}\oi_{s_1s_2}\oi_{t_3t_4}\oi_{s_3s_4}\cov(v_{kjt_1},v_{kjs_1})
\cov(\varepsilon_{it_2},\varepsilon_{\ell s_2})
\cum(\varepsilon_{it_3},\varepsilon_{\ell s_3},\varepsilon_{jt_4},\varepsilon_{js_4})
\nonumber\\
&\leq & N^{-2}
\sum_k\sup_{t_2,s_2}\left|\sum_{t_1,s_1}\oi_{t_1t_2}\cov(v_{kjt_1},v_{kjs_1})\oi_{s_1s_2}\right|\sup_{t_3,t_4}|\oi_{t_3t_4}|\sup_{s_3,s_4}|\oi_{s_3s_4}|\nonumber\\
&&\times
\left(\sup_{i,\ell}\sum_{t_2,s_2}|\cov(\varepsilon_{it_2},\varepsilon_{\ell s_2})|\right)
\left(\sum_{i,\ell}\sum_{t_3,t_4,s_3,s_4}|\cum(\varepsilon_{it_3},\varepsilon_{\ell s_3},\varepsilon_{jt_4},\varepsilon_{js_4})|\right)\nonumber\\
&=&O\left(T^2/N^2\right). \label{eq:partAB}
\eer
For
$$
\sigma^{(c)}_1=\left\{p_1,p_2\right\}=\left\{\left\{it_2,it_3,jt_4\right\},\left\{\ell s_2, \ell s_3, js_4\right\}\right\}
$$
we have
\ber
&&N^{-2}\sum_{i,j}\sum_{t_2,t_3,t_4}\sum_{s_2,s_3,s_4}\oi_{t_1t_2}\oi_{s_1s_2}\oi_{t_3t_4}\oi_{s_3s_4}\cov(v_{kjt_1},v_{kjs_1})
\cum(\varepsilon_{it_2},\varepsilon_{\ell s_2},\varepsilon_{jt_4})
\cum(\varepsilon_{it_3},\varepsilon_{\ell s_3},\varepsilon_{js_4})\nonumber
\\
&\leq & N^{-2}
\sum_k\sup_{t_2,s_2}\left|\sum_{t_1,s_1}\oi_{t_1t_2}\cov(v_{kjt_1},v_{kjs_1})\oi_{s_1s_2}\right|\sup_{t_3,t_4}|\oi_{t_3t_4}|\sup_{s_3,s_4}|\oi_{s_3s_4}|\nonumber\\
&&\times
\left(
\sum_{i,\ell}\sum_{t_2,s_2,t_4}
|\cum(\varepsilon_{it_2},\varepsilon_{\ell s_2},\varepsilon_{jt_4})|\right)
\left(
\sum_{i,\ell}\sum_{t_3,s_3,s_4}
|\cum(\varepsilon_{it_3},\varepsilon_{\ell s_3},\varepsilon_{js_4})|\right)\nonumber\\
&=&O\left(T^2/N^2\right).\label{eq:partAC}
\eer
Finally, for $\sigma^{(d)}$ we have
\ber
&&N^{-2}\sum_{i,j}\sum_{t_2,t_3,t_4}\sum_{s_2,s_3,s_4}\oi_{t_1t_2}\oi_{s_1s_2}\oi_{t_3t_4}\oi_{s_3s_4}\cov(v_{kjt_1},v_{kjs_1})
\cum\left(\varepsilon_{it_2},\varepsilon_{it_3},\varepsilon_{it_3},
\varepsilon_{\ell s_2},\varepsilon_{\ell s_3},\varepsilon_{jt_3},\varepsilon_{js_3}\right)\nonumber
\\
&\leq &N^{-2}
\sum_k\sup_{t_2,s_2}\left|\sum_{t_1,s_1}\oi_{t_1t_2}\cov(v_{kjt_1},v_{kjs_1})\oi_{s_1s_2}\right|\sup_{t_3,t_4}|\oi_{t_3t_4}|\sup_{s_3,s_4}|\oi_{s_3s_4}|\nonumber\\
&&\times \sum_{t_2,t_3,t_4}\sum_{s_2,s_3,s_4}\sum_{i,\ell}|\cum\left(\varepsilon_{it_2},\varepsilon_{it_3},\varepsilon_{it_3},
\varepsilon_{\ell s_2},\varepsilon_{\ell s_3},\varepsilon_{jt_3},\varepsilon_{js_3}\right)|\nonumber\\
&=&O(T/N^2)\label{eq:partAD}.
\eer

%
%
%
\item[]
\item
Setting $J=0$ in  Lemma~\ref{expansion}, the term (\ref{eq:boundIIb}) can be rewritten as
\ber
&&
\frac{1}{TN}\sum_{i=1}^N\vj'\SSIb\epi\epi'\md \x\psinv\v'\SSIb\epj\label{eq:boundIIb1}\\
&&+\frac{1}{TN}\sum_{i=1}^N\vj'\SSIb\epi\epi'\md \x\psinv\left(\ps-\frac{\xx'\xx}{T}\right)\ixx\v'\SSIb\epj\label{eq:boundIIb1K}
\\
&=& O_p(1). \label{eq:boundM6ALB}
\eer
The latter bound in (\ref{eq:boundM6ALB}) follows from the results in Lemmata I:
\berr
&& \lb  
\frac{1}{TN}\sum_{i=1}^N\vj'\SSIb\epi\epi'\md \x\psinv\left(\ps-\frac{\xx'\xx}{T}\right)\ixx\v'\SSIb\epj
\rb\nonumber\\
&&\leq\frac{1}{N}\sum_{i=1}^N\lb\frac{\vj'\SSIb\epi}{\sqrt{T}}\rb
\lb \frac{\epi'\md \x}{\sqrt{T}}\rb \lb\psinv\rb 
\lb\dif\rb \lb \ixx \rb\lb\v'\SSIb\epj\rb
\\
&=&O_p(1).
\eerr

Recalling that $\md\x=\md\f\Gi+\v-\pd\v$, for the term in  (\ref{eq:boundIIb1}) we have
\ber
&&\lb\frac{1}{TN}\sum_{i=1}^N\vj'\SSIb\epi\epi'\md \x\psinv\v'\SSIb\epj\rb\nonumber\\
%
&\leq & \lb\vj'\SSIb\frac{1}{TN}\sum_{i=1}^N\epi\epi'\v\psinv\v'\SSIb\epj\rb\label{eq:bcinque}\\
&&+\lb\vj'\SSIb\frac{1}{TN}\sum_{i=1}^N\epi\epi'\pd\v\psinv\v'\SSIb\epj\rb\label{eq:strano}\\
&&+  \lb\vj'\SSIb\frac{1}{TN}\sum_{i=1}^N\epi\epi'\md\f\Gi\psinv\v'\SSIb\epj\rb . \label{eq:sei}
\eer
The term (\ref{eq:strano}) is bounded using again the results in Lemmata I,
\be\label{alber1}
 \frac{1}{N}\sum_{i=1}^N\lb\frac{\vj'\SSIb\epi}{\sqrt{T}}\rb
\lb \frac{\epi'\d}{\sqrt{T}}\rb \lb \left(\frac{\d'\d}{T}\right)^{-1}\rb \lb \frac{\d'\v}{\sqrt{T}}\rb \lb\psinv\rb\lb\frac{\v'\SSIb\epj}{\sqrt{T}}\rb=O_p(1). 
\ee
For the term(\ref{eq:bcinque}) we have
\ber
&&\E\lb\vj'\SSIb\frac{1}{TN}\sum_{i=1}^N\epi\epi'\v\psinv\v'\SSIb\epj\rb^2\nonumber\\
&=&
\frac{1}{T^2N^2}\E\left[\sum_{i,\ell}\sum_{t_1,\dots,t_5}\sum_{s_1,\dots,s_5}
\sum_{k_1,k_2,k_3}\sum_{h_2,h_3}\oi_{t_1t_2}\si_{ik_2k_3}\oi_{t_4t_5}\oi_{s_1s_2}\si_{lh_2h_3}\oi_{s_4s_5}\right.\nonumber\\
&\times&\left.\E\left(v_{k_1t_1j}v_{k_1s_1j}
v_{it_3k_2}v_{ik_3t_4}
v_{\ell s_3h_2}v_{\ell h_3s_4}v_{\ell s_3h_2}\right)
\E\left(\varepsilon_{it_2}\varepsilon_{it_3}\varepsilon_{jt_5}\varepsilon_{\ell s_2}\varepsilon_{\ell s_3}\varepsilon_{js_5}\right)
\nonumber\right]
\\
&\leq&
\frac{1}{T^2N^2}
\sum_{k_1,k_2,k_3}\sum_{h_1,h_2,h_3}\E\left[\sup_{i,\ell}\left|\si_{ik_2k_3}\si_{\ell h_2h_3}\right|\right.\nonumber\\
&\times&
\sup_{i,\ell}\sup_{t_3,t_4,s_3,s_4}
\sum_{t_2,t_5}\sum_{s_2,s_5}|\oi_{t_4t_5}\oi_{s_4s_5}\E\left(
\varepsilon_{it_2}\varepsilon_{it_3}\varepsilon_{\ell s_2}\varepsilon_{\ell s_3}\varepsilon_{js_5}\varepsilon_{jt_5}\right)|\label{eq:25janA}\\
&\times& \!\!\!\!
\left.
\sum_{i,\ell}\sup_{t_2,s_2}\sum_{t_1,t_3,t_4}\sum_{s_1,s_3,s_4}|\oi_{t_1t_2}\oi_{s_1s_2}\E\left(v_{jk_1t_1}v_{jh_1s_1}v_{it_3k_2}v_{ik_3t_4}v_{\ell s_3h_2}v_{\ell h_3s_4}\right)|\right]
\label{eq:25janB}\\
&&= O(1)+O\left(\frac{T}{N}\right). \label{eq:boudT1B}
\eer
The  upper bound follows noting  that 
$
\sum_{k,h}\sup_{i}|\si_{ikh}|\leq  \infty
$  by Remark~\ref{FHILL2}
and the bounds in(\ref{eq:EXT1}) and (\ref{eq:EXT2}) below. 

The term (\ref{eq:25janB}) can be written as
\berr
&&
\sum_{i,\ell}\sup_{t_2,s_2}\sum_{t_1,t_3,t_4}\sum_{s_1,s_3,s_4}|\oi_{t_1t_2}\oi_{s_1s_2}\E\left(v_{jk_1t_1}v_{jh_1s_1}v_{it_3k_2}v_{ik_3t_4}v_{\ell h_2s_3}v_{\ell h_3s_4}\right)|\\
&=& \sum_{i,\ell}\sup_{t_2,s_2}\sum_{t_1,t_3,t_4}\sum_{s_1,s_3,s_4}\left|\oi_{t_1t_2}\oi_{s_1s_2}
\sum_{\pi=\{a_1,\dots,a_l\}\in\pp(a)}\cum\left(\bm{v^{,a_1}}\right)\cdots \cum\left(\bm{v^{,a_l}}\right)\right|,
\eerr
for $l=1,2,3$, and 
$$
\bm{a}=\left\{
jk_1t_1,jh_1s_1,ik_2t_3,ik_5t_4,\ell h_2 s_3,\ell h_5 s_4,
\right\}.
$$
First we consider the partition $\pi_1\in\pi$ where each block has the same cross-sectional index:
$$
\pi_1=\left\{
\{jk_1t_1,jh_1s_1\},\{ik_2t_3,ik_5t_4\},\{\ell h_2 s_3,\ell h_5 s_4\}
\right\},
$$
leading to
\ber
&&
\sum_{i,\ell}\sup_{t_2,s_2}\sum_{t_1,t_3,t_4}\sum_{s_1,s_3,s_4}|\oi_{t_1t_2}\oi_{s_1s_2}\cum\left(v_{jk_1t_1},v_{jh_1s_1}\right)\cum\left(v_{it_3k_2},v_{ik_3t_4}\right)\cum\left(v_{\ell s_3h_2},v_{\ell h_3s_4}\right)|\nonumber\\
&\leq &\sup_{i,\ell}\sup_{t_2,s_2}\sum_{t_1,s_1}|\oi_{t_1t_2}\cum(v_{jk_1t_1},v_{jh_1s_1})\oi_{s_1s_2}|
\sum_{t_3,t_4}|\cum\left(v_{it_3k_2},v_{ik_3t_4}\right)|
\sum_{s_3,s_4}|\cum\left(v_{\ell s_3h_2},v_{\ell h_3s_4}\right)|\nonumber\\
&=& O\left(T^2N^2\right). \label{eq:EXT1}
\eer
Using similar arguments, it can be shown that the remaining partitions the bound is $O(T^2N)$. For illustrative purposes, we consider the partition
$$
\pi_2=\left\{
\{jk_1t_1,ik_2t_3\},\{jh_1s_1,ik_5t_4\},\{\ell h_2 s_3,\ell h_5 s_4\}
\right\},
$$
for which, again by Remark~\ref{FHILL}
\berr
&&\sum_{i,\ell}\sup_{t_2,s_2}\sum_{t_1,t_3}|\oi_{t_1t_2}\cum(v_{jk_1t_1},v_{ik_2t_3})|
\sum_{s_1,t_4}|\oi_{s_1s_2}\cum\left(v_{jh_1s_1},v_{ik_3t_4}\right)|
 \sum_{s_3,s_4}|\cum\left(v_{\ell h_2 s_3},v_{\ell h_3s_4}\right)|\\
 %
 &=&O_p(T^2N).
\eerr

The proof of the result below,
\be\label{eq:EXT2}
\sup_{t_3,t_4,s_3,s_4}
\sum_{t_2,t_5}\sum_{s_2,s_5}|\oi_{t_4t_5}\oi_{s_4s_5}\E\left(
\varepsilon_{it_2}\varepsilon_{it_3}\varepsilon_{\ell s_2}\varepsilon_{\ell s_3}\varepsilon_{js_5}\varepsilon_{jt_5}\right)|=O\left(T^2\right),
\ee
follows along the same lines, and hence details are omitted.

\bigskip
Next we consider  the second term (\ref{eq:sei}). Define the $T\times K$ matrix
$
\bm{\chi_i}:=T^{-\nicefrac{1}{2}}\md\f\Gi\psinv
$
with entry $\ft_{itk}$. Noting that 
\ber
\lb\fft\rb^2& \leq & \lambda_K\left(\bm{\Psi_i^{-1}}\right)
\tr\left(\bm{\Psi_i^{-\nicefrac{1}{2}}}
\frac{\Gi'\ff'\ff\Gi}{T}
\bm{\Psi_i^{-\nicefrac{1}{2}}}
\right)\nonumber\\
&=&\lambda_K\left(\bm{\Psi_i^{-1}}\right)
\tr\left(
\bm{I_K}-\bm{\Psi_i^{-\nicefrac{1}{2}}}\SO_{\v'\v}
\bm{\Psi_i^{-\nicefrac{1}{2}}}
\right)\leq K\lambda_K\left(\bm{\Psi_i^{-1}}\right)<\infty . \label{eq:chibound}
\eer
Hence, by Lemma \ref{boundinverse}\ref{boundinv1}, $\sup_i\lb\fft\rb$ is uniformly bounded. Then,
\ber
&&\phantom{=}\E\lb\vj'\SSIb\frac{1}{TN}\sum_{i=1}^N\epi\epi'\md\f\Gi\psinv\v'\SSIb\epj\rb^2\nonumber\\
&&=\E\lb\vj'\SSIb\frac{1}{\sqrt{T}N}\sum_{i=1}^N\epi\epi'\fft\v'\SSIb\epj\rb^2\nonumber\\
&&=\E\left[
\frac{1}{TN^2}\sum_{i,\ell}\sum_{t_1,\dots,t_5}\sum_{s_1,\dots,s_5}
\sum_{k_1,k_2}\sum_{h_1,h_2}\oi_{t_1t_2}\oi_{s_1s_2}\oi_{t_4t_5}\oi_{s_4s_5}\ft_{it_3k_2}\ft_{is_3h_2}
\right.\nonumber\\
&&\phantom{=}\times\left.
\E\left(
v_{jk_1t_1}v_{jk_1s_1}v_{it_4k_2}v_{\ell s_4h_2}
\right)
\E\left(
\varepsilon_{it_2}\varepsilon_{it_3}
\varepsilon_{\ell s_2}\varepsilon_{\ell s_3}\varepsilon_{jt_5}\varepsilon_{js_5}
\right)
\right]\nonumber\\
&&\leq\E\left[
\frac{1}{TN^2}
\sup_{i,\ell}\sum_{t_2,t_3,t_5}\sum_{s_2,s_3,s_5}
|\E\left(
\varepsilon_{it_2}\varepsilon_{it_3}
\varepsilon_{\ell s_2}\varepsilon_{\ell s_3}\varepsilon_{jt_5}\varepsilon_{js_5}
\right)\ft_{it_3k_2}\ft_{is_3h_2}|\right.\label{eq:NINO2}\\
&&\phantom{=}\times \sup_{t_2,t_5,s_2,s_5}\sum_{i,\ell}\sum_{t_1,t_4}\sum_{s_1,s_4}|\oi_{t_1t_2}\oi_{s_1s_2}\oi_{t_4t_5}\oi_{s_4s_5}\E\left(
v_{jk_1t_1}v_{jk_1s_1}v_{it_4k_2}v_{\ell s_4h_2}\right)|\label{eq:NINO}\\
&&= \frac{1}{TN^2}O(N)O(T^2)=O\left(\frac{T}{N}\right). 
\label{eq:BOUNDT1b}
\eer
To prove that latter bound, we first consider the sum in Equation (\ref{eq:NINO}). By Remark \ref{FHILL},
\berr
%
%
&&\sup_{t_2,t_5,s_2,s_5}\sum_{i,\ell}\sum_{t_1,t_4}\sum_{s_1,s_4}\left|\oi_{t_1t_2}\oi_{s_1s_2}\oi_{t_4t_5}\oi_{s_4s_5}\E\left(
v_{jk_1t_1}v_{jk_1s_1}v_{it_4k_2}v_{\ell s_4h_2}\right)\right|\\
&\leq& \sup_{t_2,s_2}\sum_{t_1,s_1}\left|\oi_{t_1t_2}\cum\left(v_{jk_1t_1},v_{jk_1s_1}\right)\oi_{s_1s_2}\right|
\sup_{t_5,s_5}\sum_{i,\ell}\sum_{t_4,s_4}\left|\oi_{t_4t_5}\cum\left(v_{it_4k_2},v_{\ell s_4h_2}\right)\oi_{s_4s_5}\right|\\
&+&\sup_{t_2,t_5}\sum_{t_1,t_4}\sum_i\left|\oi_{t_1t_2}\cum\left(v_{jk_1t_1},v_{it_4k_2}\right)\oi_{t_4t_5}\right|
\sup_{s_2,s_5}\sum_{s_1,s_2}\sum_{\ell}\left|\oi_{s_1s_2}\cum\left(v_{jk_1s_1},v_{\ell s_4h_2}\right)\oi_{s_4s_5}\right|\\
&+&\sup_{t_2,s_5}\sum_{t_1,s_4}\sum_{\ell}\left|\oi_{t_1t_2}
\cum\left(v_{jk_1t_1},v_{\ell s_4h_2}\right)\oi_{s_4s_5}\right|
\sup_{s_2,t_5}\sum_{s_1,t_4}\sum_{i}\left|
\oi_{s_1s_2}\cum\left(v_{jk_1s_1},v_{it_4k_2}\right)\oi_{t_4t_5}\right|\\
&+& \sup_{t_1,t_2}\left|\oi_{t_1t_2}\right|\sup_{s_1,s_2}\left|\oi_{s_1s_2}\right|
\sup_{t_4,t_5}
\sum_i\left|\oi_{t_4t_5}\right|\sup_{s_4,s_5}\left|\oi_{s_4s_5}\right|\sum_{t_1,t_4,s_1,s_4}
\sum_{i,\ell}\cum\left(v_{jk_1t_1}v_{jk_1s_1}v_{it_4k_2}v_{\ell s_4h_2}\right)\\
&=& O(N)+O(1)+O(1)+O(1).
\eerr
Next, for the term in  (\ref{eq:NINO2}), we show that,  for $l=1,2,3$,
\berr
&&\sup_{i,\ell}\sum_{t_2,t_3,t_5}\sum_{s_2,s_3,s_5}
|\E\left(
\varepsilon_{it_2}\varepsilon_{it_3}
\varepsilon_{\ell s_2}\varepsilon_{\ell s_3}\varepsilon_{jt_5}\varepsilon_{js_5}
\right)\ft_{it_3k_2}\ft_{is_3h_2}|\\
&=&\sup_{i,\ell}\sum_{t_2,t_3,t_5}\sum_{s_2,s_3,s_5}
\left|\sum_{\pi=\{a_1,\dots,a_l\}\in\pp(a)}\cum(\bm{\varepsilon^{,a_1}})\cdots\cum(\bm{\varepsilon^{,a_l}})\ft_{it_3k_2}\ft_{is_3h_2}\right|=O(T^2).
\eerr
For the summation over $\pi$, we consider only two partitions for illustrative purposes, namely
$$
\pi_1 =\left\{
\{it_2,it_3\},\{\ell s_2,\ell s_3\},\{jt_5,js_5\}
\right\},\qquad
\pi_2 = \left\{
\{it_2,\ell s_2\},\{i t_3,\ell s_3\},\{jt_5,js_5\}
\right\}
$$
The partition $\pi_1$ is formed by blocks with the same cross-sectional indices. By Lemma~\ref{matrix} and the result in display (\ref{eq:chibound}), we have
\berr
&&\sup_{i,\ell}\sum_{t_2,t_3,t_5}\sum_{s_2,s_3,s_5}
|\cum\left(
\varepsilon_{it_2},\varepsilon_{it_3}\right)
\cum\left(\varepsilon_{\ell s_2},\varepsilon_{\ell s_3}\right)\cum\left(\varepsilon_{jt_5},\varepsilon_{js_5}
\right)\ft_{it_3k_2}\ft_{is_3h_2}|\\
&=&\sup_{i,\ell}
\left(\sum_{t_2,t_3}\left|\cov(\varepsilon_{it_2},\varepsilon_{it_3})\ft_{it_3k_2}\right|\right)
\left(\sum_{s_2,s_3}\left|\cov(\varepsilon_{\ell s_2}\varepsilon_{\ell s_3})\ft_{\ell s_3h_2}\right|\right)
\left(\sum_{t_5,s_5}\left|\cov(\varepsilon_{j t_5},\varepsilon_{j s_5})\right|\right)\\
&=&O\left(\st\right)O\left(\st\right)O(T),
\eerr
whereas, for $\pi_2$
\berr
&&\sup_{i,\ell}\sum_{t_2,t_3,t_5}\sum_{s_2,s_3,s_5}
|\cum\left(
\varepsilon_{it_2}\varepsilon_{\ell s_2}\right)
\cum\left(\varepsilon_{it_3}\varepsilon_{\ell s_3}\right)\cum\left(\varepsilon_{jt_5}\varepsilon_{js_5}
\right)\ft_{it_3k_2}\ft_{is_3h_2}|\\
&=&\sup_{i,\ell}
\left(\sum_{t_2,s_2}|\cov(\varepsilon_{it_2},\varepsilon_{\ell s_2})|\right)
\left(\sum_{t_3,s_3}|\ft_{i t_3 k_2}\cov(\varepsilon_{i t_3}\varepsilon_{\ell s_3})\ft_{\ell s_3h_2}|\right)
\left(\sum_{t_5,s_5}|\cum\left(\varepsilon_{jt_5}\varepsilon_{js_5}\right)|\right)\\
&=&O\left(T\right)O\left(1\right)O(T).
\eerr

%
%
%
%
%
%
%
\item[]
\item
Similarly to \emph{Part (b)} we first consider the inequality 
\ber
&& \lb\frac{1}{TN}\sum_{i=1}^N\vj'\SSI\v\left(\frac{\xx'\xx}{T}\right)^{-1}\x'\md\epi\epi'\SSIb\epj\rb\nonumber\\
&\leq &
\lb\frac{1}{TN}\sum_{i=1}^N\vj'\SSI\v\left(\frac{\xx'\xx}{T}\right)^{-1}\v'\md\epi\epi'\SSIb\epj\rb\label{eq:cuno5a}\\
&+& \lb\frac{1}{TN}\sum_{i=1}^N\vj'\SSI\v\left(\frac{\xx'\xx}{T}\right)^{-1}\Gi'\f'\md\epi\epi'\SSIb\epj\rb\label{eq:cuno5b}\\
&=&\max\left\{
O_p\left(1\right),O_p\left(\frac{T^{\nicefrac{3}{2}}}{N}\right)
\right\}.  \label{eq:BOUNDT1c}
\eer
Using Lemma~\ref{expansion} with $J=2$, the squared norm in  (\ref{eq:cuno5a}) is bounded by
\ber 
&&\lb\frac{1}{TN}\sum_{i=1}^N\vj'\SSI\v\left(\frac{\xx'\xx}{T}\right)^{-1}\v'\md\epi\epi'\SSIb\epj\rb^2\nonumber\\
&\leq & \lb\frac{1}{TN}\sum_{i=1}^N\vj'\SSI\v\psinv\v'\md\epi\epi'\SSIb\epj\rb^2
\label{eq:ctre1}
\\
&+&\lb\frac{1}{TN}\sum_{i=1}^N\vj'\SSI\v\psinv\dif\psinv\v'\md\epi\epi'\SSIb\epj\rb^2\label{eq:ctre2}\\
&+&\lb\frac{1}{TN}\sum_{i=1}^N\vj'\SSI\v\left[\psinv\dif\right]^2\psinv\v'\md\epi\epi'\SSIb\epj\rb^2\label{eq:ctre3}\\
&+&\lb\frac{1}{TN}\sum_{i=1}^N\vj'\SSI\v\rem^{(3)}\v'\md\epi\epi'\SSIb\epj\rb^2\label{eq:ctre4}\\
&=& \max\left(O_p(1),O_p\left(\frac{T^{\nicefrac{3}{2}}}{N}\right)\right).
\label{eq:BoundCA}
\eer
To show how the latter bound is obtained, we analyse each term separately.
Let $\m_{ts}$ denote $t,s$-entry of $\md$. We first consider the expectation of the term (\ref{eq:ctre1}):
\berr
&&\E\left[\frac{1}{N^2T^2}\sum_{t_1,\dots,t_6}\sum_{s_1,\dots,s_6}\sum_{k_1,k_2,k_3}\sum_{h_1,h_2,h_3}\sum_{i,\ell}\oi_{t_1t_2}\oi_{t_5t_6}\oi_{s_1s_2}\oi_{s_5s_6}\m_{t_3t_4}\m_{s_3s_4}\si_{k_2k_3}\si_{h_2h_3}\right.\\
&&\left.\times \E\left(
v_{jt_1k_1}v_{js_1k_1}
v_{it_2k_2}v_{it_3k_3}
v_{\ell s_2h_2}v_{\ell t_3h_3}
\right)
\E\left(
\varepsilon_{it_4}\varepsilon_{it_5}
\varepsilon_{\ell s_4}\varepsilon_{\ell s_5}
\varepsilon_{j t_6}\varepsilon_{j s_6}
\right)
\right]\\
&=&\E\left[\frac{1}{N^2T^2}\sum_{t_1,\dots,t_6}\sum_{s_1,\dots,s_6}\sum_{k_1,k_2,k_3}\sum_{h_1,h_2,h_3}\sum_{i,\ell}\oi_{t_1t_2}\oi_{t_5t_6}\oi_{s_1s_2}\oi_{s_5s_6}\m_{t_3t_4}\m_{s_3s_4}\si_{k_2k_3}\si_{h_2h_3}\right.\\
&&\times\sum_{\pi=\{a_1,\dots,a_l\}\in\pp(\bm{a})}\cum(\vvv^{,a_1})\cdots\cum(\vvv^{,a_l})
\sum_{\sigma=\{p_1,\dots,p_k\}\in\pp(\bm{p})}\cum(\bm{\varepsilon}^{,p_1})\cdots\cum(\bm{\varepsilon}^{,p_k}),
\eerr
where $k,l=1,2,3$.

Consider first the partitions  where all the blocks have the same cross-sectional indexes 
\ber
&&\E\left[\frac{1}{N^2T^2}\sum_{t_1,\dots,t_6}\sum_{s_1,\dots,s_6}\sum_{k_1,k_2,k_3}\sum_{h_1,h_2,h_3}\sum_{i,\ell}
\oi_{t_1t_2}\oi_{t_5t_6}\oi_{s_1s_2}\oi_{s_5s_6}\m_{t_3t_4}\m_{s_3s_4}\si_{k_2k_3}\si_{h_2h_3}\right.\nonumber\\
&&\times  \cov(v_{jt_1k_1},v_{js_1k_1})\cov(v_{i t_2 k_2},v_{i t_3 k_3})
\cov(v_{\ell  s_2 h_2},v_{\ell  s_3 h_3})\nonumber\\
&&\times\left.
\cov(\varepsilon_{it_4},\varepsilon_{it_5})\cov(\varepsilon_{\ell s_4},\varepsilon_{\ell s_5})
\cov(\varepsilon_{jt_6},\varepsilon_{js_6})\right]\nonumber\\
&\leq&\frac{\ka}{T^2}\E\sup_{i,\ell}\left[
\sup_{t_3,s_3}\Big|\sum_{t_2,s_2}\cov(v_{it_2k_2},v_{it_3k_3})\left[\sum_{t_1,s_1}\oi_{t_1t_2}\cov(v_{jt_1k_1},v_{js_1k_1})\oi_{s_1s_2}\right]
\cov(v_{\ell  s_2 h_2},v_{\ell  s_3 h_3})\Big|\right.\nonumber
\\
&\times& \left.\sum_{t_3,s_3}\Big|\sum_{t_4,t_5}\sum_{s_4,s_5}
\m_{t_3t4}\cov(\varepsilon_{it_4},\varepsilon_{it_5})\left[\sum_{t_6s_6}\oi_{t_5t_6}\cov(\varepsilon_{jt_6},\varepsilon_{js_6})\oi_{s_5s_6}\right]\cov(\varepsilon_{\ell s_4},\varepsilon_{\ell s_5})\m_{s_3s_4}\label{eq:rinogaeteano}
\Big|\right]\\
&=&\frac{1}{T^2}O(1)O(T^2)=O\left(1\right). \label{eq:BoundC1}
\eer
We only show how to derive the bound for the term in (\ref{eq:rinogaeteano}). 
By inequality 4.(e) p. 111 in \cite{ltk}, is it bounded by $T^{\nicefrac{3}{2}}
\lb \md\Hi\left[\SSIb\bm{\Xi_j}\SSIb\right]\bm{\Xi_{\ell}}\md\rb_{sp}=O_p\left(
T^{\nicefrac{3}{2}}\right)$. 

We consider two further sets of partitions, characterized by different upper bounds.
\berr
\pi_1&=&\left\{
\{jt_1k_1,js_1k_1\},\{it_2k_2,\ell s_2 h_2\},\{it_3k_3,\ell s_3 h_3\}\right\},\quad \sigma_1=\left\{\{it_4,\ell s_4\}, \{it_5,\ell s_5\},\{jt_6,js_6\}
\right\},\\
\pi_2 &=& \left\{
\{ jt_1k_1, it_2 k_2\}, \{js_1k_1, \ell s_2 h_2\}, \{it_3k_3, \ell s_3 h_3\}\right\},\quad\sigma_2=\left\{
\{it_5,jt_6\}, \{\ell s_5, j s_6\}, \{it_4,\ell s_4\}
\right\}.
\eerr
For the partitions $\pi_1,\sigma_1$ we have 

\ber
&& \E\left[\frac{1}{(NT)^2}\sum_{t_1,\dots,t_6}\sum_{s_1,\dots,s_6}\sum_{k_1,k_2,k_3}\sum_{h_1,h_2,h_3}\sum_{i,\ell}\oi_{t_1t_2}\oi_{t_5t_6}\oi_{s_1s_2}\oi_{s_5s_6}\m_{t_3t_4}\m_{s_3s_4}\si_{k_2k_3}\pi_{h_2h_3}\right.\nonumber\\
&&\times \left.
\cov(v_{ j t_1 k_1},v_{j s_1 k_1})\cov(v_{ i t_2 k_2},v_{\ell s_2 h_2})\cov(v_{it_3k_3},v_{\ell s_3 h_3})\right.\nonumber\\
&&\times\left.\cov(\varepsilon_{it_4},\varepsilon_{\ell s_4})\cov(
\varepsilon_{it_5},\varepsilon_{\ell s_5})
\cov(\varepsilon_{j t_6},\varepsilon_{j s_6})
\right]\nonumber\\
&\leq & \frac{\ka}{(NT)^2}\E\left[\sup_{t_2,s_2}\left|\sum_{t_1,s_1}\oi_{t_1t_2}\cov(v_{ j t_1 k_1},v_{j s_1 k_1})\oi_{s1s_2}\right|\sup_{i,\ell}\left[\sup_{t_4s_4}\left|\sum_{t_3,s_3}\m_{t_3t_4}\cov(v_{it_3k_3},v_{\ell s_3 h_3}) \m_{s_3s_4}\right|\right.\right.\nonumber\\
&&\times\left.\sup_{t_5,s_5}\left|\sum_{t_6,s_6}\oi_{t_5t_6}\cov(
\varepsilon_{j t_6},\varepsilon_{js_6})\oi_{s_5s_6}\right|
\sum_{t_4,s_4}|\cov(\varepsilon_{it_4},\varepsilon_{\ell s_4})|\sum_{t_5,s_5}|\cov(\varepsilon_{i t_5},\varepsilon_{\ell s_5})|\right]\nonumber\\
&&\times
\left.\sum_{i,\ell}\sum_{t_2,s_2}\left|\cov(v_{ i t_2 k_2},v_{\ell s_2 h_2})\right|\right]\nonumber\\
&=&\frac{1}{(NT)^2}O(1)O(1)O(1)O(T)O(T)O(TN)=O\left(\frac{T}{N}\right). \label{eq:BoundC2}
\eer
Finally, for the partitions $\pi_2,\sigma_2$
\berr
&& \E\left[\frac{1}{(NT)^2}\sum_{t_1,\dots,t_6}\sum_{s_1,\dots,s_6}\sum_{k_1,k_2,k_3}\sum_{h_1,h_2,h_3}\sum_{i,\ell}\oi_{t_1t_2}\oi_{t_5t_6}\oi_{s_1s_2}\oi_{s_5s_6}\m_{t_3t_4}\m_{s_3s_4}\si_{k_2k_3}\si_{h_2h_3}\right.\nonumber\\
&&\times \left.
\cov(v_{ j t_1 k_1},v_{i t_2 k_2})\cov(v_{js_1k_1},v_{\ell s_2 h_2})\cov(v_{it_3k_3},v_{\ell s_3 h_3}) \right]\nonumber\\
&&\times\left.\cov(\varepsilon_{it_5},\varepsilon_{j t_6})\cov(\varepsilon_{\ell  s_5},\varepsilon_{ js_5})\cov(
\varepsilon_{i t_4},\varepsilon_{\ell s_4})\right]\nonumber\\
&\leq &\frac{\ka}{(NT)^2}\E\left\{
\sup_{t_5,t_6}\left|\oi_{t_5t_6}\right|\sum_i\sum_{t_5,t_6}\left|\cov(
\varepsilon_{i t_5},\varepsilon_{jt_6})\right|
\sup_{s_5,s_6}\left|\oi_{s_5s_6}\right|\sum_{\ell}\sum_{s_5,s_6}\left|\cov(
\varepsilon_{\ell s_5},\varepsilon_{js_6})\right|\nonumber\right.\\
&& \times \sup_{i,\ell}\left[\sup_{t_1,t_2}|\oi_{t_1,t_2}|\sum_{t_1,t_2}\cov\left(v_{jt_1k_1},v_{it_2k_2}\right)
\sup_{s_1,s_2}|\oi_{s_1,s_2}|\sum_{s_1,s_2}\cov\left(v_{js_1k_1},v_{\ell s_2 h_2}\right)\right.\nonumber\\
&&\times\left.\left.\sup_{t_4s_4}\left|\sum_{t_3,s_4}\m_{t_3t_4}\cov(v_{it_3k_4},v_{\ell s_3 h_3}) \m_{s_3t_4})\right|\sum_{t_4,s_4}|\cov(\varepsilon_{it_4},\varepsilon_{\ell s_4})|
\right]\right\}\nonumber\\
&=&\frac{\ka}{(NT)^2}O(T)O(T)O(T)O(T)O(1)O(T)= O\left(\frac{T^3}{N^2}\right).
\label{BoundC3}
\eerr

Using the same techniques illustrated above, it can be shown that  the terms in  (\ref{eq:ctre2}) and (\ref{eq:ctre3}) are $\max\left(O_p(1),O_p\left(T^{\nicefrac{3}{2}}/N\right)\right)$. The proof is tedious, due the the high order of the cumulants involving the $v_{ikt}$s, and hence omitted to save space. It might be worth to note that bound of the latter term requires the high order moment assumptions on the process  $\{v_{itk}\}$ in Assumption \ref{ass V}.

Concerning (\ref{eq:ctre4}) we have
\berr
&&\lb\frac{1}{TN}\sum_{i=1}^N\vj'\SSI\v\rem^{(3)}\v'\md\epi\epi'\SSIb\epj\rb\\
&\leq& 
\frac{1}{N}\sum_{i=1}^N\lb\frac{\vj'\SSI\v}{T}\rb \lb\rem^{(3)}\rb \lb\v'\md\epi\rb \lb\epi'\SSIb\epj\rb=O(1).
\eerr
completing the proof of (\ref{eq:BoundCA}).

For the term (\ref{eq:cuno5b}) we have (see equation (\ref{eq:chibound}))
\berr 
&&\lb\frac{1}{TN}\sum_{i=1}^N\vj'\SSI\v\left(\frac{\xx'\xx}{T}\right)^{-1}\Gi\f'\md\epi\epi'\SSIb\epj\rb^2\nonumber\\
&\leq & \lb\frac{1}{\sqrt{T}N}\sum_{i=1}^N\vj'\SSI\v\bm{\chi_i}'\epi'\SSIb\epj\rb^2
\\
&+&\lb\frac{1}{\sqrt{T}N}\sum_{i=1}^N\vj'\SSI\v\psinv\dif\bm{\chi_i}'\epi'\SSIb\epj\rb^2\\
&+&\lb\frac{1}{\sqrt{T}N}\sum_{i=1}^N\vj'\SSI\v\left[\psinv\dif\right]^2\bm{\chi_i}'\epi'\SSIb\epj\rb^2\\
&+&\lb\frac{1}{TN}\sum_{i=1}^N\vj'\SSI\v\rem^{(3)}\Gi'\f\md\epi\epi'\SSIb\epj\rb^2\\
&=& \max\left(O_p(1),O_p\left(\frac{T^{\nicefrac{3}{2}}}{N}\right)\right).
\eerr
The latter bound proceeding as for terms (\ref{eq:sei}) and (\ref{eq:cuno5a}); the proof is omitted. 

\item
%
%
Setting $J=1$ in the expansion in Lemma \ref{expansion}, we get
\ber
&&\lb\frac{1}{T^2N}\sum_{i=1}^N\vj'\SSIb
\v\left(\frac{\xx'\xx}{T}\right)^{-1}\x'\md\epi\epi'
\md\x\left(\frac{\xx'\xx}{T}\right)^{-1}\v'\SSIb\epj\rb\nonumber\\
&\leq & 
\lb\frac{1}{T^2N}\sum_{i=1}^N\vj'\SSIb
\v\psinv\x'\md\epi\epi'
\md\x\psinv\v'\SSIb\epj\rb+O_p(1)\nonumber\\
&=&\max\left(O_p(1,)\left(\frac{T}{N}\right)\right). \label{eq:BOUNDT1d}
\eer
The proof for the latter bound follows closely that of part $(c)$, and hence is omitted.
\end{Part}
%
%
%
%
\subsubsection*{Bound Term III}
Using the expansion in Lemma~\ref{expansion} for $J=2$,
\ber
&&\lb \vj'\SSIb\left(\bm{\hat{C}_{2N}}-\BTN\right)\SSIb\epj\rb\\
&=&\lb\frac{1}{N}\sum_{i=1}^N\vj'\SSIb\left(\v\xx^+\ff\b\b'\ff'\left(\xx^+\right)'\v'-\BT\right)\SSIb\epj
\rb\nonumber\\
&\leq&\lb 
\frac{1}{N}\sum_{i=1}^N\vj'\SSIb\left(\v\psinv\frac{\xx'\ff}{T}\b\b'\frac{\ff'\xx}{T}\psinv\v'-\BT\right)\SSIb\epj\rb\label{eq:W1}\\
&+&\lb 
\frac{1}{T^2N}\sum_{i=1}^N\vj'\SSIb\v\psinv\dif\psinv\right.\nonumber\\
&&\times\left.\xx'\ff\b\b'\ff'\xx\psinv\v'\SSIb\epj\label{eq:W2}\rb\\
&+&\lb 
\frac{1}{T^2N}\sum_{i=1}^N\vj'\SSIb\v\psinv\xx'\ff\b\b'\ff'\xx\right.\nonumber\\
&&\times\left. \psinv\dif\psinv\v'\SSIb\epj\label{eq:W3}\rb\\
&+&\lb 
\frac{1}{T^2N}\sum_{i=1}^N\vj'\SSIb\v\left[\psinv\dif\right]^2\psinv
\right.\nonumber\\
&&\times \left.
\xx'\ff\b\b'\ff'\xx\psinv\v'\SSIb\epj\rb\label{eq:W4}\\
&+&\lb 
\frac{1}{T^2N}\sum_{i=1}^N\vj'\SSIb\v\psinv\xx'\ff\b\b'\ff'\xx\right.\nonumber\\
&&\times \left.\left[\psinv\dif\right]^2\psinv\v'\SSIb\epj\rb\label{eq:W5}\\
&+&
\frac{1}{N}\sum_{i=1}^N\lb \vj'\SSIb\v\rb \lb \psinv \rb \lb \left(\frac{\xx'\xx}{T}\right)^{-1} \rb \lb \rem^{(3)}\rb\lb\v'\SSIb\epj\rb  \lb \frac{\xx'\ff}{T}\rb   \lb \b \rb^2 \label{eq:W6}\\
&=&O_p(1).\nonumber
\eer
To prove the bound, we start considering 
the norm in equation (\ref{eq:W1}). The latter can be further bounded by
\ber
&&\lb 
\frac{1}{N}\sum_{i=1}^N\vj'\SSIb\left(\v\psinv\frac{\xx'\ff}{T}\b\b'\frac{\ff'\xx}{T}\psinv\v'-\BT\right)\SSIb\epj\rb\nonumber\\
&\leq &\lb
\frac{1}{N}\sum_{i=1}^N\vj'\SSIb\left(\v\psinv\Gi'\frac{\ff'\ff}{T}\b\b'\frac{\ff'\ff}{T}\Gi\psinv\v'-\BT\right)\SSIb\epj\rb\label{eq:WW1}\\
&+ &\lb
\frac{1}{N}\sum_{i=1}^N\vj'\SSIb\v\psinv\frac{\v'\md\ff}{T}\b\b'\frac{\ff'\ff}{T}\Gi\psinv\v'\SSIb\epj\rb
\label{eq:WW2}
\\
&+ &\lb
\frac{1}{N}\sum_{i=1}^N\vj'\SSIb\v\psinv\Gi'\frac{\ff'\ff}{T}\b\b'\frac{\ff'\md\v}{T}\psinv\v'\SSIb\epj\rb
\label{eq:WW3}
\\
&+ &\lb
\frac{1}{N}\sum_{i=1}^N\vj'\SSIb\v\psinv\frac{\v'\md\ff}{T}\b\b'\frac{\ff'\md\v}{T}\psinv\v'\SSIb\epj\rb .
\label{eq:WW4}
\eer
%
%
%
\bigskip
We first consider the expectation of the term (\ref{eq:WW1}).
%
%
\berr
&&\E\left[\frac{1}{N^2T^4}\sum_{t_1,\dots,t_6}\sum_{s_1,\dots,s_6}\sum_{k_1,\dots,k_5}
\sum_{h_2,\dots,h_5}\sum_{m_1,\dots,m_4}\sum_{n_1,\dots,n_4}
\right.\nonumber\\
&&\times f_{m_1t_3}f_{t_3m_2}f_{m_3t_4}f_{t_4m_4} \oi_{t_1t_2}\oi_{t_5t_6}\oi_{s_1s_2}\oi_{s_5s_6}\E(\varepsilon_{jt_6}\varepsilon_{js_6})\\
&&\times \E\left[ v_{jk_1t_1}v_{jk_1s_1}\left(v_{it_2k_2}v_{ik_5t_5}-\E v_{it_2k_2}v_{ik_5t_5}\right)
\left(v_{\ell s_2h_2}v_{\ell h_5t_5}-\E v_{\ell s_2h_2}v_{\ell h_5t_5}\right)\right]
\nonumber
\\
&&\times\left.
\si_{ik_2k_3}\si_{ik_4k_5}\si_{\ell h_2h_3}\si_{\ell h_42h_5}
\gamma_{ik_3m_1}\gamma_{im_4k_4}\gamma_{\ell h_3n_1}\gamma_{\ell n_4h_4}b_{im_2}b_{im_3}b_{\ell n_2}b_{\ell n_3}
\right]^{\nicefrac{1}{2}}\\
&\leq&\ka
\E\left\{\left(\frac{1}{T}\sum_{m}\sum_{t}f^4_{mt}\right)^2\sup_{t_5,s_5}\left|\sum_{t_6s_6}\oi_{t_5,t_6}\cum(\varepsilon_{jt_5},\varepsilon_{js_6})\oi_{s_5s_6}\right|\right.\label{eq:ta110}\\
&&\sup_{k_1,k_2,k_5}\sup_{h_2,h_5}
\left[\sup_{t_2,s_2}\left|\sum_{t_1,s_1}\oi_{t_1t_2}\cum(v_{jk_1t_1},v_{jk_1s_1})\oi_{s_1s_2}\right|\right.\\
&&\times\left.
\frac{1}{N^2}\sum_{i,\ell}\sum_{t_2,t_5}\sum_{s_2,s_5}\left|\cum\left( v_{it_2k_2}v_{ik_5t_5},v_{\ell s_2h_2}v_{\ell h_5s_5}\right)
\right|\right.
\label{eq:ta11}\\
&&+\left.\left.
\frac{1}{N^2}\sum_{i,\ell}
\sum_{t_1,t_2,t_5}\sum_{s_1,s_2}
|\cum\left(
v_{jk_1t_1}v_{jk_1s_1},
v_{it_2k_2}v_{ik_5t_5},v_{\ell s_2h_2}v_{\ell h_5s_5}
\right)
\right]\right\}^{\nicefrac{1}{2}}
\label{eq:ta11a}\\
&=&\E\lb\frac{\ff'\ff}{T}\rb^2 O(1)\left[O(1)O\left(\frac{T^2}{N}\right)+O\left(\frac{T^3}{N^2}\right)\right]O(1)=O\left(\frac{T^2}{N}\right).
\eerr
About term (\ref{eq:WW2})
\ber
&&
\E\lb
\frac{1}{N}\sum_{i=1}^N\vj'\SSIb\v\psinv\frac{\v'\md\f}{T}\b\b'\frac{\ff'\ff}{T}\Gi\psinv\v'\SSIb\epj\rb\nonumber\\
\leq 
&&\E\left[
\frac{1}{N^2T^4}\sum_{t_1,\dots, t_4,t_6,t_7}\sum_{s_1,\dots, s_4,s_6, s_7}\sum_{k_1,\dots,k_5}\sum_{h_2,\dots,h_5}\sum_{i,\ell}
\right.\nonumber\\
&&
\E\left(v_{k_1jt_1}v_{k_2it_2}v_{k_3it_3}v_{k_5it_6}
v_{k_1js_1}v_{h_2\ell s_2}v_{h_3\ell s_3}v_{h_5\ell s_6}\right)\label{eq:jan2018a}\\
&&\E\left(\varepsilon_{jt_7}\varepsilon_{js_7}\right)\oi_{t_1t_2}\oi_{s_1s_2}\oi_{t_6t_7}\oi_{s_6s_7}\nonumber\\
&&\m_{t_3t_4}\m_{s_3s_4}\mathrm{f}_{m_1t_4}\mathrm{f}_{n_1s_4}\si_{ik_3k_4}\si_{ik_4k_5}\si_{\ell h_3h_4}\si_{\ell h_4h_5}\gamma_{im_3k_4}\gamma_{\ell n_3h_4}b_{im_1}b_{im_2}b_{\ell n_1}b_{\ell n_2}\nonumber\\
&&\left.
\sum_{t_5}f_{m_2t_5}f_{m_3t_5}
\sum_{s_5}f_{n_2s_5}f_{n_3s_5}
\right]^{\nicefrac{1}{2}}. \nonumber
\eer
where $\mathrm{f}_{ts}$ and $f_{ts}$ denote the elements for the matrices $\f$ and $\ff$, respectively.
The expectation in (\ref{eq:jan2018a}) can be rewritten as $\sum_{\pi=\{a_1,\dots a_l\}\in\pp(a)}\cum(\bm{v^{,a_1}},\dots,\bm{v^{,a_l}})$, for $l=1,2,3,4$, and 
$$
a=\left\{
k_1jt_1,k_1js_2,
k_2it_2,k_3it_3,k_5it_7,
 h_2\ell s_2,h_3\ell s_3,h_5\ell s_7
\right\}.
$$
We focus on the partition $\pi$, where the blocks have the same cross-sectional indexes:
$$
\pi=\left\{ 
\left\{k_1jt_1,k_1js_2\right\},
 \left\{k_2it_2,k_3it_3,k_5it_7\right\},
 \left\{h_2\ell s_2,h_3\ell s_3,h_5\ell s_7\right\}
\right\},
$$
that leads to bound the following expectation
\berr
&&\E\left[
\frac{1}{N^2T^4}\sum_{t_1,\dots, t_4,t_6,t_7}\sum_{s_1,\dots, s_4,s_6, s_7}\sum_{k_1,\dots,k_5}\sum_{h_2,\dots,h_5}\sum_{i,\ell}
\cum\left(v_{k_1jt_1},v_{k_1js_1}\right)
\right.\\
&&
\cum\left( v_{k_2it_2},v_{k_3it_3},v_{k_5it_7} \right)
\cum\left( v_{h_2\ell s_2},v_{h_3\ell s_3},v_{h_5\ell s_7}\right)
\E\left(\varepsilon_{jt_7}\varepsilon_{js_7}\right)\oi_{t_1t_2}\oi_{s_1s_2}\oi_{t_7t_8}\oi_{s_7s_8}\\
&&\mu_{t_3t_4}\mu_{s_3s_4}f_{m_1t_4}f_{n_1s_4}\si_{ik_3k_4}\si_{ik_4k_5}\si_{\ell h_3h_4}\si_{\ell h_4h_5}\gamma_{im_3k_4}\gamma_{\ell n_3h_4}b_{im_1}b_{im_2}b_{\ell n_1}b_{\ell n_2}\\
&&\left.
\sum_{t_5}f_{m_2t_5}f_{m_3t_5}
\sum_{s_5}f_{n_2s_5}f_{n_3s_5}
\right]^{1/2}\\
&
\leq& \ka\E\left\{\left(\frac{1}{T}\sum_{m}\sum_{t} f^2_{mt}\right)^2
\left[\frac{1}{T^2}\sum_{m_1,m_2}\sum_{n_1,n_2}\sum_{k_1,\dots,k_5}\sum_{h_2,\dots,h_5}
\sup_{t_2,s_2}\left|\sum_{t_1,s_1}\oi_{t_1t_2}\cum\left(v_{k_1jt_1},v_{k_1js_1}\right)\oi_{s_1s_2}\right|
\right.\right.\\
&&
\times
\sup_{i}\left|\sum_{t_2,t_3,t_7}\cum\left( v_{k_2it_2},v_{k_3it_3},v_{k_5it_6} \right)\sum_{t_4}\m_{t_3t_4}\mathrm{f}_{m_1t_4}\right|\sup_{t_6,s_6}\left|\sum_{t_7,s_7}\oi_{s_6s_7}\cum\left(\varepsilon_{jt_7},\varepsilon_{js_7}\right)\oi_{s_7s_6}\right|\\
&&\times\left.\left.
\sup_{\ell}\left|\sum_{s_2,s_3,s_7}\cum\left( v_{h_2\ell s_2},v_{h_3\ell s_3},v_{h_5\ell s_7} \right)\sum_{s_4}\m_{s3s_4}\mathrm{f}_{n_1s_4}\right|
%
%
\right]^{1/2}\right\}\\
&=&O\left(\E\left(\frac{1}{T}\sum_{m}\sum_{t} f^2_{mt}\right)^{3/2}\right)O(1)=O(1),
\eerr
because by Lemma~\ref{matrix}, $\left|\sum_{t_1,t_2,t_3}\cum\left( v_{k_1 i t_1},v_{k_2 i t_2},v_{k_3 i t_3} \right)\sum_{t_4}\m_{t_3t_4}f_{m t_4}\right|\leq \ka \sqrt{T}\|\ff\|$
For the remaining partitions $\pp(a)\backslash\pi$, suing the same techniques one can show that 
\berr
&&\E\left[
\frac{1}{N^2T^4}\sum_{t_1,\dots, t_4,t_7,t_8}\sum_{s_1,\dots, s_4,s_7, s_8}\sum_{k_1,\dots,k_5}\sum_{h_2,\dots,h_5}\sum_{i,\ell}
\right.\\
&&
\sum_{\sigma=\{a_1,\dots,a_l\in \pp(a)\backslash\pi} \cum(\bm{v^{,a_1}\cdots v^{,a_l}} )
\E\left(\varepsilon_{jt_7}\varepsilon_{js_7}\right)\oi_{t_1t_2}\oi_{s_1s_2}\oi_{t_6t_7}\oi_{s_6s_7}\\
&&\m_{t_3t_4}\m_{s_3s_4}\mathrm{f}_{m_1t_4}\mathrm{f}_{n_1s_4}\si_{ik_3k_4}\si_{ik_4k_5}\si_{\ell h_3h_4}\si_{\ell h_4h_5}\gamma_{im_3k_4}\gamma_{\ell n_3h_4}b_{im_1}b_{im_2}b_{\ell n_1}b_{\ell n_2}\\
&&\left.
\sum_{t_5t_6}f_{m_2t_5}f_{m_3t_5}
\sum_{s_5s_6}f_{n_2s_5}f_{n_3s_5}
\right]^{1/2}=O\left(\sqrt{\frac{T}{N}} \right).
\eerr
Proceeding in a similar way, terms (\ref{eq:WW3})-(\ref{eq:WW4}) can be show to be $O_p(1)$, whereas the 
 terms (\ref{eq:W2})-(\ref{eq:W5}) are $O_p(1)$. Finally, the term in (\ref{eq:W6}) is $O_p(1)$, using the results in Lemmas \ref{boundinverse}, {\ref{cor} and \ref{extra1}, Remark \ref{mreg} and noting that $\lb\rem^3\rb=O_p(T^{\nicefrac{3}{2}})$.
%

%
%
%
%
\subsubsection*{Bound Term IV}
In the sequel we will show that 
\berr
&&\lb\vvj'\SSI\left(\bm{\hat{C}_{3N}}+\bm{\hat{C}'_{3N}}\right)\SSI\ej\rb\\
&\leq &
\lb\frac{1}{N}\sum_{i=1}^N\vvj'\SSI\sum_{h=1}^{13}\bm{\hat{C}_{3i,h}}\SSI\ej\rb+
\lb\frac{1}{N}\sum_{i=1}^N\vvj'\SSI\sum_{h=1}^{13}\bm{\hat{C}_{3i,h}'}\SSI\ej\rb\\
&=&\max\left\{O_p(1),\frac{T^2}{N}\right\}.
\eerr
We will split each norm in thirteen terms ($h=1,\dots,13$). To save notation, we write $\bm{\hat{C}_{3i,h'}}$ to indicate $\bm{\hat{C}'_{3i,h}}$. For example, $h,h'=4$ mean that we are considering
the terms  $N^{-1}\sum_{i=1}^N\vj'\SSI\bm{\hat{C}_{3i,4}}\SSI\epj$ and $N^{-1}\sum_{i=1}^N\vj'\SSI\bm{\hat{C}'_{3i,4}}\SSI\epj$, respectively. Due the high number of terms involved, and because the derivation of the bounds exploit similar techniques to the ones previously used, we will provide very concise proofs. Since most of the $\hat{C}_{3i,h}$s are function $\xp$, the bounds' derivations will require the use of Lemma \ref{expansion}. For each term, the order of the expansion will be specified.
\begin{itemize}
\item[$\bm{h=1}:$] Setting $J=0$,
\berr
&&\lb\frac{1}{N}\sum_{i=1}^N \vvj' \SSI\bm{\hat{C}_{3i,1}}\SSI\ej\rb\\
&=&
\lb  \frac{1}{N}\sum_{i=1}^N\vvj' \SSI \vv \psinv \Gi'\frac{\ff'\ff}{T}\b\b'\ff'\SSI\ej\rb+O_p(1)\\
&\leq &  \underbrace{\lb\frac{1}{N}\sum_{i=1}^N\vvj' \SSI \vv \psinv \Gi'\frac{\ff'\ff}{T}\b\b'\rb}_{O_p\left(T/N\right)+O_p\left(\sqrt{T/N}\right)}
\underbrace{\lb \ff'\SSI\ej\rb}_{O_p(T^{-\nicefrac{1}{2}})}+O_p(1)=O_p\left(1\right).
\eerr
\item[$\bm{h'=1}:$] Setting $J=0$,
\berr
&&\lb\frac{1}{N}\sum_{i=1}^N \vvj' \SSI\bm{\hat{C}'_{3i,1}}\SSI\ej\rb\\
 &\leq&\underbrace{\lb \vvj'\SSI\ff \rb}_{O_p\left(T^{-\nicefrac{1}{2}}\right)}
 \frac{1}{N}\sum_{i=1}^N \underbrace{\lb\b\b'\rb\lb\Gi'\rb\lb\frac{\ff'\ff}{T}\rb\lb\psinv\rb}_{O_p(1)}\underbrace{\lb\vv'\SSI\ej\rb}_{O_p(\sqrt{T})}+O_p(1)=O_p\left(1\right).
\eerr
%
%
%
%
\item[$\bm{h=2}:$] Similarly to the case $h=1$,
\berr
&&\lb\frac{1}{N}\sum_{i=1}^N \vvj' \SSI\bm{\hat{C}_{3i,2}}\SSI\ej\rb\\
 &=&
\lb\frac{1}{N}\sum_{i=1}^N  \vvj' \SSI \vv\psinv\Gi'\frac{\ff'\ff}{T}\b\b' \frac{\ff'\ff}{T}\Gi\psinv\Gi'\ff'\SSI\ej\rb+O_p(1)\\
&=&O_p\left(1\right).
\eerr
\item[$\bm{h'=2}:$] Similarly to the case $h'=1$,
$$
\lb\frac{1}{N}\sum_{i=1}^N \vvj' \SSI\bm{\hat{C}'_{3i,2}}\SSI\ej\rb= O_p(1).
$$
\item[$\bm{h=3}:$] 
$$
\lb\frac{1}{N}\sum_{i=1}^N \vvj' \SSI\bm{\hat{C}_{3i,3}}\SSI\ej\rb
\leq\underbrace{\lb \vvj'\SSI\ff \rb}_{O_p\left(T^{-\nicefrac{1}{2}}\right)}
\underbrace{\lb \frac{1}{N}\sum_{i=1}^N \b\epi'\SSIb\epj\rb}_{O_p(T/N)+O_p(\sqrt{T/N})}=O_p\left(1\right).
$$

\item[$\bm{h'=3}:$]
$$
\lb\frac{1}{N}\sum_{i=1}^N \vvj' \SSI\bm{\hat{C}'_{3i,3}}\SSI\ej\rb \leq  \underbrace{\lb \vvj'\SSI\ei\rb}_{O_p\left(T^{\nicefrac{1}{2}}\right)}
\underbrace{\lb \b\rb}_{O(1)} \underbrace{\lb\ff'\SSI\ej\rb}_{O_p\left(T^{-\nicefrac{1}{2}}\right)}=O_p(1).
$$
\item[$\bm{h=4}:$] Setting $J=1$,
\berr
&&\lb\frac{1}{N}\sum_{i=1}^N \vvj' \SSI\bm{\hat{C}_{3i,4}}\SSI\ej\rb\\
 &\leq &
\underbrace{\lb  \frac{1}{T^2N}\sum_{i=1}^N\vvj' \SSI \vv \psinv \Gi'\ff'\ff\b\ei'\ff\Gi\psinv\vv'\SSI\ej\rb
}_{O_p\left(T/\sqrt{N}\right)}
\\
&&+\underbrace{\lb  \frac{1}{T^2N}\sum_{i=1}^N\vvj' \SSI \vv \psinv \vv'\ff\b\ei'\ff\Gi\psinv\vv'\SSI\ej\rb
}_{O_p(\sqrt{T/N})}
\\
&&+
\underbrace{\lb  \frac{1}{T^2N}\sum_{i=1}^N\vvj' \SSI \vv \psinv \Gi'\ff'\ff\b\ei'\vv\psinv\vv'\SSI\ej\rb
}_{O_p\left(T/\sqrt{N}\right)}
\\
&&+\underbrace{\lb  \frac{1}{T^2N}\sum_{i=1}^N\vvj' \SSI \vv \psinv \vv'\ff\b\ei'\vv\psinv\vv'\SSI\ej\rb
}_{O_p\left(\sqrt{T/N}\right)}
\\
&+&
  \underbrace{\lb\frac{1}{TN}\sum_{i=1}^N\vvj' \SSI \vv \psinv\dif\psinv\Gi' \frac{\ff'\ff}{T}\b\ei'\xx\psinv\vv'\SSI\ej\rb}_{O_p(\sqrt{T/N})}\\
&+&
\underbrace{\lb  \frac{1}{TN}\sum_{i=1}^N\vvj' \SSI \vv \psinv \Gi'\frac{\ff'\ff}{T}\b\ei'\xx\psinv\dif\psinv\vv'\SSI\ej\rb}_{O_p(\sqrt{T/N})}+O_p(1)\\
&=&\max\left\{O_p(1),O_p\left(\frac{T}{\sqrt{N}}\right)\right\}.
\eerr
\item[$\bm{h'=4}$:] 
$
\lb N^{-1}\sum_{i=1}^N \vvj' \SSI\bm{\hat{C}'_{3i,4}}\SSI\ej\rb=\max\left\{O_p(1),O_p\left(\frac{T}{\sqrt{N}}\right)\right\}.
$
\item[$\bm{h=5}$:]
\berr
&&
\lb\frac{1}{N}\sum_{i=1}^N \vvj' \SSI\bm{\hat{C}_{3i,5}}\SSI\ej\rb\\
 &\leq&
\underbrace{
\lb  \frac{1}{N}\sum_{i=1}^N\vvj' \SSI \vv \psinv\sum_{s=0}^3\left[\dif\psinv\right]^s\Gi' \frac{\ff'\ff}{T}\b\ei'\SSI\ej\rb}_{O_p\left(T^2/N\right)}\\
&+&
\underbrace{
\lb  \frac{1}{TN}\sum_{i=1}^N\vvj' \SSI \vv \sum_{s=0}^2\left(\psinv\dif\right)^s \vv'\ff\b\ei'\SSI\ej\rb}_{O_p(\sqrt{T^3/N^2})}+O_p(1)\\
&=&\max\left\{O_p(1),O_p\left(\frac{T^{2}}{N}\right)\right\}.
\eerr
\item[$\bm{h'=5}$:] Setting $J=1$
\berr
&&
\lb\frac{1}{N}\sum_{i=1}^N \vvj' \SSI\bm{\hat{C}'_{3i,5}}\SSI\ej\rb+O_p(1) = 
\max\left\{O_p(1),O_p\left(\frac{T}{\sqrt{N}}\right)\right\}.
\eerr
\item[$\bm{h=6}$:]
\berr
&&\lb\frac{1}{N}\sum_{i=1}^N\vvj'\SSI\bm{\hat{C}_{3i,6}}\SSI\ej\rb\\
&\leq &\sum_{i=1}^N\underbrace{\lb\vvj'\SSI\ff\rb}_{T^{-\nicefrac{1}{2}}}
\underbrace{
\lb\b\rb\lb\frac{\ei'\xx}{\sqrt{T}}\rb
\lb\left(\frac{\xx'\xx}{T}\right)^{-1}\rb
\lb\frac{\vv\SSI\ej}{\sqrt{T}}\rb}_{O_p(1)}
=O_p\left(\frac{1}{\sqrt{T}}\right).
\eerr
\item[$\bm{h'=6}$:]
\berr
&&\lb\frac{1}{N}\sum_{i=1}^N\vvj'\SSI\bm{\hat{C}'_{3i,6}}\SSI\ej\rb\\
 &\leq & \sum_{i=1}^N\underbrace{\lb\frac{\vvj'\SSI\vvj}{T}\rb
\lb\left(\frac{\xx'\xx}{T}\right)^{-1}\rb}_{O_p(1)}
\underbrace{\lb\xx'\ei\rb}_{O_p(\sqrt{T})}
\underbrace{\lb\b\rb\lb\ff\SSI\ej\rb}_{T^{-\nicefrac{1}{2}}}=O_p\left(1\right).\\
\eerr
\item[$\bm{h=7}$:] Setting $J=0$,
\berr
&&\lb\frac{1}{N}\sum_{i=1}^N\vvj'\SSI\bm{\hat{C}_{3i,7}}\SSI\ej\rb\leq \underbrace{\lb \vvj'\SSI\ff \rb}_{O_p(T^{-\nicefrac{1}{2}})}.\\
 &\times & 
\underbrace{\lb \frac{1}{N}\sum_{i=1}^N 
\Gi\psinv\Gi' \frac{\ff'\ff}{T}\b\ei'\SSI\ej\rb}_{O_p(\sqrt{T/N})}\frac{1}{N}\sum_{i=1}^N \lb\Gi\rb^2 \lb\psinv\rb^2\underbrace{\lb\dif\rb}_{O_p(T^{-\nicefrac{1}{2}})}\underbrace{\lb\ei'\SSI\ej\rb}_{O_p(T)} \\
&=&O_p\left(1\right).\nonumber\\
\eerr
%
\item[$\bm{h'=7}$:]
$$
 \lb\frac{1}{N}\sum_{i=1}^N\vvj'\SSI\bm{\hat{C}'_{3i,7}}\SSI\ej\rb\leq\frac{1}{N}\sum_{i=1}^N
\lb  \vvj'\SSI\ej \rb \lb\b\rb \lb\frac{\ff'\xx}{T}\rb\lb\Gi\rb^2 \lb \ff'\SSI\ej \rb =O_p(1).\nonumber
$$
\end{itemize}  
The remaining terms ($h,h'=8,9,\dots,13$) can be shown to be $O_p(1)$ using similar arguments to those used for the term labelled $h'=7$. 
%
%
%
%
%
%
%
%

\subsubsection{Proof of Lemma~\ref{mmtt2}}}
The proof is similar to the proof of Lemma \ref{mmtt1}, and hence omitted.
%
%
%
%
\subsubsection{Proof Lemma \ref{mmtt3}}
The proof is similar to the proof of Lemma~\ref{mmtt1} and hence omitted.
%
\subsubsection{Proof Lemma \ref{mmttF}}
%
Proceeding as in the proof of Lemma \ref{mmtt1}, and using Lemma \ref{cor}
\berr
&&\lb (\sh-\SS)\SSI\ff \rb\\
&\leq & \overbrace{\lb \ff\left(\han-\an\right)\ff'\SSI\ff \rb}^{Term \;I^*}+    
\overbrace{\lb\left(\bm{\hat{C}_{1N}}-\db'\Hn\db\right)\SSI\ff \rb }^{Term \;II^*}\\ 
&&+\underbrace{\lb \left(\bm{\hat{C}_{2N}}-\db'\BTN\db\right)\SSI\ff \rb}_{Term \;III^*} 
+\underbrace{\lb\left(\bm{\hat{C}_{3N}}+\bm{\hat{C}'_{3N}}\right)\SSI\ff \rb }_{Term \;IV^*} = O_p\left(1\right).
\eerr
%
The Term $I^*$ is bounded by $\lb \ff\rb \lb\han-\an\rb \lb\ff'\SSI\ff\rb\leq O_p(\sqrt{T})O_p\left(T^{-\nicefrac{1}{2}}\right)O_p\left(1\right)$. Concerning the Term $II^*$, he have 
\berr
\lb\bm{\hat{C}_{1N}}\SSI\ff\rb&\leq & 
\frac{1}{N}\sum_i\lb\ei-\frac{\vv}{\sqrt{T}}\left(\frac{\xx'\xx}{T}\right)^{-1}\frac{\xx'\ei}{\sqrt{T}}\rb
\lb\ei'\SSI\ff-\frac{\ei'\xx}{T}\left(\frac{\xx'\xx}{T}\right)^{-1}\vv'\SSI\ff\rb\\
&\leq& \frac{1}{N}\sum_i\frac{1}{\sqrt{T}}\lb\ei+O_p(1)\rb \lb\sqrt{T}\ei'\SSI\ff+o_p(1)\rb=O_p(1),
\eerr
and $\lb\db'\Hn\db\SSI\ff\rb\leq \lb\Hn\rb_{sp}\lb\SSI\ff\rb=O_p\left(T^{-\nicefrac{1}{2}}\right)$. Using similar argument, it can be shown that the Term $III^*$ is $O_p(1)$. For the last term ($IV^*$), we will show that
\berr
&&\lb \frac{1}{N}\sum_{i=1}^N\left(\sum_{h=1}^{13}\bm{\hat{C}}_{3i,h}+\bm{\hat{C}'}_{3i,h}\right)\SSI\ff \rb=O_p\left(\frac{T}{N}\right).
\eerr
We will consider only the terms for $h=1$. By Lemma \ref{approx inv} we find 
\berr
\lb\bm{\hat{C}}_{3i,1}\SSI\ff\rb&=&\lb\frac{1}{N}\sum_{i=1}^N\vv\xp\ff\b\b'\ff'\SSI\ff\rb\\
&\leq &
\lb\db'\rb\lb \frac{1}{N}\sum_{i=1}^N\v\psinv\Gi'\frac{\ff'\ff}{T}\b\b'\rb\lb\ff\SSI\ff \rb+O_p(1).
\eerr
Because
\berr
&&\E\lb \frac{1}{N}\sum_{i=1}^N\v\psinv\frac{\ff'\ff}{T}\b\b'\rb\\
&=& \ka\frac{1}{T}\sum_{t_1}\sum_{m_1,m_2}\E|f_{m_1t_1}f_{t_1m_2}|
\left[\frac{1}{N}
\sum_{t_2}\sum_{i,\ell}|\cov(v_{i t_2k_1},v_{\ell t_2 k_2})|\right]^{\nicefrac{1}{2}}=O\left(\sqrt{\frac{T}{N}}\right),
\eerr
we conclude that $\lb\bm{\hat{C}}_{3i,1}\SSI\ff\rb=O_p(1)$. For $h'=1$ we have
\berr
\lb\bm{\hat{C}'}_{3i,1}\SSI\ff\rb
&\leq &
 \frac{1}{N}\sum_{i=1}^N\lb\frac{\ff}{\sqrt{T}}\rb
\lb\psinv\rb 
\lb\Gi\rb\lb\frac{\ff'\ff}{T}\rb \lb\b\rb^2\lb\sqrt{T}\vv'\SSI\ff\rb +O_p(1)=O_p(1).
\eerr
Using the same arguments for $h,h'=2,\dots,13$ we prove the results.
\subsubsection{Proof Lemma \ref{mmtt4}}
Let $\bm{\bar{\mathcal{H}}}:=\db\MAI\db'$, and denote by $\bar{\mathcal{h}}_{ts}$ its $(t,s)$-entry, for $s,t=1,\dots,(T-S).$ Then,  equations (\ref{eq:MMM1}) and (\ref{eq:silver}), (\ref{eq:qnormale}) and (\ref{eq:qbreve}) entail that
\berr
&&\lb\vvj'\left(\SSI-\MAI\right)\ej\rb=\lb\vj'\SSIb\frac{1}{N}\sum_{i}\left(\BT-\T\right)\bm{\bar{\mathcal{H}}}\epj\rb\\
&=& \lb\frac{1}{N}\sum_{i=1}^N\sum_{k,\ell=1}^K\vj'\SSIb\left[\bar{\mathcal{q}}_{ik\ell}\E\left({\bf{v}}_{ik}{\bf{v}}'_{i\ell}\right)\right]\bm{\bar{\mathcal{H}}}\epj\rb\leq
\frac{1}{N}\sum_{i=1}^N\sum_{k,\ell=1}^K
\left|\bar{\mathcal{q}}_{ik\ell}\right|
\lb\vj'\SSIb\E\left({\bf{v}}_{ik}{\bf{v}}'_{i\ell}\right)\bm{\bar{\mathcal{H}}}\epj\rb .
\\
\eerr
where $\bar{\mathcal{q}}_{ik\ell}:=\breve{\mathcal{q}}_{ik\ell}-\mathcal{q}_{ik\ell}$ and ${\bf{v}}_{ik}$ denotes the $k-$th column of the matrix $\v$. Note that $\bar{\mathcal{q}}_{ik\ell}\xrightarrow{p}0$. From arguments similar  to those used in proof of Lemmata III, if follow that 
$\lb\vj'\SSIb\E\left({\bf{v}}_{ik}{\bf{v}}'_{i\ell}\right)\bm{\bar{\mathcal{H}}}\epj\rb=O_p(\sqrt{T})$, concluding the proof.

\section{Common regressors}\label{app2}

In this Section we provide some abbreviated proofs for the theorems of Section~\ref{common}. Further details are available upon request.

\begin{theorem}\label{TI1}
Under Assumption~2.4:
\[
\left(   \begin{array}{c}  \widetilde{\bm \alpha }_i^{GLS}  \\  \widetilde{\bm \beta }_i^{GLS}  \end{array}  \right) = \left(\z_i'  \stldp \z_i \right)^{+}
\z_i'     \stldp  \bm{Y}_i =   \left(  \begin{array}{c}  \bm{0} \\  \hat{\bm{ \beta }}_i^{GLS}   \end{array}  \right) .
\]
\end{theorem}
\begin{proof}
Using Equation (1) in \cite[Section 3.5.3]{ltk}  for the partitioned inverse of $\zi'\zi$, with $\zi=\left[\d, \x\right]$, $\mzi$ can be rewritten as  
$$
\mzi=\db\left(
\bm{I_{T-k}}-\db'\x\left(
\x'\db\db'\x
\right)^{-1}\x'\db
\right)\db'=\db\mxx\db' .
$$
Equation (\ref{eq:gidio}) entail that $\hui=\mzi\ui=\db\mxx\uu$; hence
 $
\stldp=\db\left(\bm{\hat{\mathcal{S}}_N}\right)^{-1}\db'
$,  by
Lemma \ref{fact6410}, with $\bm{\hat{\mathcal{S}}_N}$  defined in (\ref{eq:charlie2}).
It follows that
 \be 
\left( \z_i'   \stldp   \z_i\right)^+ =
   \left[
 \begin{array}{cc}   \mathbf{0}  &    \mathbf{0} \\
   \mathbf{0} &
 \left(\xx'\shi\xx\right)^{-1}
     \end{array} \right],\quad 
  \z_i'   \stldp \y = 
     \left[
 \begin{array}{c}   \mathbf{0}   \\
 \xx'\shi\yy
     \end{array} \right] ,
\ee
completing the proof.
\end{proof}

\begin{theorem}\label{TI2}  Under  the assumptions of Theorem~\ref{Theorem_GLS}, if 
\[
\d' \f = \bf{0},
\]
then,
\[
  \sqrt{T}  ( {\bf\cal L }_{i}^\alpha )^{-{1 \over 2}} ( \widehat{\bm  \alpha }^{(\hat{\beta}_i)}_i - \bm {\alpha }_{i0} ) \xrightarrow{d} \norm( \bm{0},  \I_S  ) ,
\]
where
\berr
\hspace{-0.5in} {\bf\cal L }_{i}^\alpha  &:=&  \left(\frac{\d'\d}{ T} \right) \left(   \d' \x \left( \v'\cni\v  \right)^{-1},\; \I_S \right) 
 \Big(  \begin{array}{c} \vv'\cni \\ \d'  \end{array}  \Big) { \Hi \over T }   \\
 &&
 \times \left(   \begin{array}{cc} \cni \v, \; \d  \end{array}     \right)  \left(  \begin{array}{c}  ( \v'\cni\v   )^{-1} ( \x ' \d  )  \\  \I_S \end{array} \right)  \left({\d'\d\over T} \right).
  \eerr
\end{theorem}
\begin{proof} Rewrite
\berr
 \widehat{\bm  \alpha }^{(\hat{\beta}_i)}_i &=& \left(\d' \d )^{-1} \d' (  \bm{Y}_i - \x \bm{\hat{ \beta }_i^{GLS}}  )  = (\d' \d )^{-1} \d' (  \d \ba_{i0}  + \x \bb_{i0}  + \ui - \x \bm{\hat{ \beta }_i^{GLS}} \right )\\
&=& \ba_{i0}- \left(\d' \d \right)^{-1}\left[\d'\x\left(\bm{\hat{ \beta }_i^{GLS}}-\bb_{i0}\right)+\d'\epi\right],
\eerr
where we used $ \d' \ui = \d' \epi $. Then apply Theorem~\ref{Theorem_GLS} together with:
\[
{  \d' \epi  \over \sqrt{T} }  \xrightarrow{d} \norm ( \bm{0} , \bm{\Sigma }_{\d' \Hi \d } ).
\]
The asymptotic covariance matrix   follows  by taking into account the covariance between $ T^{\nicefrac{1}{2}}( \hat{\bm{ \beta }}_i^{GLS} - \bb_{i0} ) $ and $T^{-\nicefrac{1}{2}}  \d' \epi $. 
\end{proof}

\begin{theorem}  Under  the assumptions of Theorem~\ref{Theorem_GLS}, if 
\be
\d' \f = \bf{0},   \label{iden}
\ee
then,
\[
\sqrt{T}  {\breve{\bf\cal L }}_{i}^{-{1 \over 2}}  \left[ \left(   \begin{array}{c}  \breve{\bm \alpha }_i^{GLS  }    \\  \breve{\bm \beta }_i^{GLS}  \end{array}  \right) - 
 \left(   \begin{array}{c}  {\bm \alpha }_{i0}  \\  {\bm \beta }_{i0}  \end{array}  \right) \right]
  \xrightarrow{d} \norm (  \bm{0},   \bold{I}_{K+S} )\quad  \mbox{ as }  \quad {1 \over T }  + { T^3 \over N^2 } \rightarrow 0,
\]
where
\[
 {\breve{\bf\cal L} }_{i} := {1 \over T }
  \left( \bf{Z}_i'    \breve{\bm{S}}_N^{-1}  \bf{Z}_i \right)^{-1}  
  \left(     \begin{array}{l}    \d'   \\    \x' \MAI   +\x' \pd  \end{array} \right)  \Hi
   \left(     \begin{array}{c}    \d   \\    \MAI \x  +\x \pd  \end{array} \right)  
  \left( \bf{Z}_i'    \breve{\bm{S}}_N^{-1}  \bf{Z}_i \right)^{-1}  .
 \]
\end{theorem}

\begin{proof}  For the matrix  $\breve{\bm{S}}_N$ defined in (\ref{eq:gidio3}), it is easy to verify that (see Property 15 in \cite[3.6.2]{ltk}
$$
\breve{\bm{S}}_N^{-1}=\stldp+\pd=\db\shi\db'+\pd ,
$$
yielding 
\berr
\left( \bf{Z}_i'    \breve{\bm{S}}_N^{-1}  \bf{Z}_i \right)^{-1} 
&=&   \left[     \begin{array}{ccc}    \d'\d & &\d'\x \\  \x' \d && \x'\db\shi \db' \x +  \x' \pd\x       \end{array} \right]^{-1} \\
&=&   \left[     \begin{array}{ccc}    {\bf E_i}  & &  - {\bf E_i }   \d' \x  \A_i^{-1}  \\  -\A_i^{-1} \x' \d {\bf E_i }  & &  \A_i^{-1} + \A_i^{-1}  \x' \d  {\bf E_i }  \d' \x  \A_i^{-1}   \end{array} \right], \\
\eerr
with  $\A_i:=   \x' \db   \shi \db' \x +  \x' \pd\x    $ and $ {\bf E_i  }:= ( \d' (  {\bf I } -   \x \A^{-1}  \x '  )\d )^{-1}  $.
 Similarly, 
\berr
\bf{Z}_i' \breve{\bm{S}}_N^{-1}   \bf{y}_i   && = \left(\bf{Z}_i'   \breve{\bm{S}}_N^{-1}  
 \bf{Z}_i\right)  \left(  \begin{array}{c}   \bm{\alpha_{i0}} \\ \bm{\beta_{i0}}   \end{array}  \right) +
 {\bf Z}_i'   \breve{\bm{S}}_N^{-1} 
  \ui   
 \eerr
  where, using the identification assumption (\ref{iden}). 
$$
   {\bf Z}_i'    \breve{\bm{S}}_N^{-1}   \ui  =  \left[     \begin{array}{l}    \d'  \epi  \\    \xx' \shi  \uu  +\x' \pd \epi \end{array} \right],
$$
and, from the proof of Theorem \ref{Theorem_GLS}, $T^{-\nicefrac{1}{2}}\xx' \shi  \uu\approx T^{-\nicefrac{1}{2}}\xx' \MAI  \ei$.
 \end{proof}

\section{Different factor structure}\label{app3}
Following up the discussion in Section \ref{diff-fac}, we now consider the feasible GLS. In this section we allow for  $\d\neq 0$, that is we consider the model
\be\label{eq:diffact2}
\y=\d\bm{\alpha}_{i0}+\x\bm{\beta}_{i0}+\ui, \quad \x=\f_1\Gi+\v, \quad \ui=\f_2\b+\epi .
\ee
The aim of this section is to provide an heuristic proof that the results of Theorem \ref{Theorem_GLS} apply to model (\ref{eq:diffact2}). Proceeding along the lines of Appendix \ref{auxiliary}, we define
\berr
\sh^{(\gg)} &:=&\frac{1}{N}\sum_{i=1}^N\mxx\ui \ui ' \mxx\\
&=&\frac{1}{N}\sum_{i=1}^N\left(\I-\ff_1\Gi\xp-\vv\xp\right)\left(\ff_2\b+\ei \right)\left(\ff_2\b+\ei \right)'  
\left(\I-\ff_1\Gi\xp-\vv\xp\right)'\\
&=& \frac{1}{N}\sum_{i=1}^N\left(\gg\bm{\hat{A}^{(\gg)}_{1i}}\gg'+
\bm{\hat{C}^{(\gg)}_{1i}}+\bm{\hat{C}^{(\gg)}_{2i}}+\bm{\hat{C}^{(\gg)}_{3i}}
\right)=\gg\bm{\hat{A}_N}^{(\gg)} \gg' + \bm{\hat{C}_N}^{(\gg)}.
\eerr
The matrices $\gg$ and $\hat{\A}_i$ are defined below. Noting that
\berr
&&\left(\I-\ff_1\Gi\xp\right)\ff_2\b=
\left(\I-\left(\projp_{\ff_2}+\projm_{\ff_2}\right)\ff_1\Gi\xp\right)\ff_2\b\\
&=&\ff_2\left[\I-\left(\ff_2'\ff_2\right)^{-1}\ff_2'\ff_1\xp\ff_2\right]\b
+\projm_{\ff_2}\ff_1\Gi\xp\ff_2\b,
\eerr
the matrix
$$
\left(\I-\ff_1\Gi\xp\right)\ff_2\b\b'\ff_2'
\left(\I-\ff_1\Gi\xp\right)'+\left(\ff_1\Gi\xp\right)\ei\ei'
\left(\ff_1\Gi\xp\right)'
$$
can be rewritten as $\gg\han^{(\gg)}\gg'$, with $\gg:=\left[\ff_2,\;\projm_{\ff_2}\ff_1\right]$ of dimension $T\times (M_1+M_2)$. Write
$$
\bm{\hat{\mathcal{K}}_i}:=
\left[
\begin{array}{c}
\I-\left(\ff_2'\ff_2\right)^{-1}\ff_2'\ff_1\Gi\xp\ff_2\\
\Gi\xp\ff_2
\end{array}
\right],\quad
\bm{\hat{\mathcal{L}}_i}:=
\left[
\begin{array}{c}
\left(\ff_2'\ff_2\right)^{-1}\ff_2'\ff_1\xp\\
\Gi\xp\ff_2
\end{array}
\right].
$$
Then,
$$
\han^{(\gg)}=\frac{1}{N}\sum_{i=1}^N \left(\bm{\hat{\mathcal{K}}_i}\b\b'\bm{\hat{\mathcal{K}}'_i}
+\bm{\hat{\mathcal{L}}_i}\ei\ei'\bm{\hat{\mathcal{L}}'_i}\right).
$$

Likewise
$$
\bm{\hat{C}^{(\gg)}_{1i}}= \left(\I-\vv\xp\right)\ei\ei'
\left(\I-\vv\xp\right)',\qquad
\bm{\hat{C}_{2i}}^{(\gg)}=\left(\vv\xp\right)\ff_2\b\b'\ff_2'
\left(\vv\xp\right)'.
$$
The terms $\bm{\hat{C}^{(\gg)}_{3i}}$ is defined as
$
\bm{\hat{C}^{(\gg)}_{3i}}=\sum_{j=1}^{13} \left(
\bm{\hat{C}^{(\gg)}_{3i,j}}+\bm{\hat{C}^{(\gg)'}_{3i,j}}\right),
$
where
$$
\begin{array}{lcl}
\bm{\hat{C}^{(\gg)}_{3i,1}}= -\vv\xp\ff_2\b\b'\ff_2', & \phantom{ccccc}&
%
%
\bm{\hat{C}^{(\gg)}_{3i,2}}= \vv\xp\ff_2'\b\b'\ff_2'
\left(\ff_1\Gi\xp\right)'  , \\
\bm{\hat{C}^{(\gg)}_{3i,3}}=\ff_2\b\ei', & &
\bm{\hat{C}^{(\gg)}_{3i,4}}=\vv\xp\ff_2\b\ei'\left(\vv\xp\right)' , \\
\bm{\hat{C}^{(\gg)}_{3i,5}}=-\vv\xp\ff_2\b\ei', & &
\bm{\hat{C}^{(\gg)}_{3i,6}}=-\ff_2\b\ei'\left(\vv\xp\right)' , \\
\bm{\hat{C}^{(\gg)}_{3i,7}}=-\ff_1\Gi\xp\ff_2\b\ei', &&
\bm{\hat{C}^{(\gg)}_{3i,8}}=\vv\xp\ff_2\b\ei'\left(\ff_1\Gi\xp\right)', \\
\bm{\hat{C}^{(\gg)}_{3i,9}}=\ff_1\Gi\xp\ff_2'\b\ei'\left(\vv\xp\right)', && 
\bm{\hat{C}^{(\gg)}_{3i,10}}=-\ff_1\Gi\xp\ei\ei' ,\\ 
\bm{\hat{C}^{(\gg)}_{3i,11}}=\ff_1\Gi\xp\ei\ei'\left(\vv\xp\right)', & &
\bm{\hat{C}^{(\gg)}_{3i,12}}=\ff_1\Gi\xp\ff_2\b\ei'\xp\Gi'\ff_1' ,
\\
\bm{\hat{C}^{(\gg)}_{3i,13}}=
-\ff_2\b\ei'\left(\ff_1\Gi\xp\right)'  .&& 
\end{array}
$$
Next define the matrices 
$$
\SS^{(\gg)}:=\gg \an^{(\gg)}\gg'+\db'\bcn^{(\gg)}\db,\quad
\bcn^{(\gg)}:= \frac{1}{N}\sum_{i=1}^N \left(\Hi+\BT^{(\gg)}\right),\quad
$$
and
\berr
\an^{(\gg)}&:=&\frac{1}{N}\sum_{i=1}^N\left(\I_T-\ff_2^+\ff_1\Gi\bm{\Psi_{1i}}^{-1}\Gi'\frac{\ff_1'\ff_2}{T}\right)
\b\b'
\frac{1}{N}\sum_{i=1}^N\left(\I_T-\ff_2^+\ff_1\Gi\bm{\Psi_{1i}}^{-1}\Gi'\frac{\ff_1'\ff_2}{T}\right)',\\
 \BT&:=&\E\left[\v\psinv \Gi\frac{\ff_1'\ff_2}{T}\b\b'\frac{\ff_2'\ff_1}{T}\Gi'\bm{\Psi_{1i}}^{-1}\v\v'\left|\sa(\z)\right.\right],
\eerr
with $\ff_2^+=\left(\ff_2'\ff_2\right)^{-1}\ff_2$, and
\be
\bm{\Psi}_{1i}:=\Gi'\frac{\ff_1'\ff_1}{T}\Gi+\SO_{\v'\v}.
\ee
The key difference with respect to what discuss in Section \ref{diff-fac} lies  in the fact that the inverse of $\SS^{(\gg)}$ is orthogonal both to $\ff_1$ and $\ff_2$.

\bibliography{cross-section_PZ}
\end{document}